\pdfoutput=1
\documentclass[11pt]{article}
\PassOptionsToPackage{numbers,compress}{natbib}
\input commands.sty
\usepackage{tikz}

\usetikzlibrary{arrows.meta, decorations.pathreplacing, calc}

\newcommand{\A}{\ms{A}}
\newcommand{\OB}{\ms{O}}
\newcommand{\Gg}{\mc{G}}
\newcommand{\Ff}{\mc{F}}
\newcommand{\Mm}{\mc{M}}
\newcommand{\lgp}{\log_2^{+}}

\newtheorem{openprob}[thm]{Open Problem}

\crefname{thm}{Theorem}{Theorems}
\Crefname{thm}{Theorem}{Theorems}
\crefname{lem}{Lemma}{Lemmas}
\Crefname{lem}{Lemma}{Lemmas}
\crefname{cor}{Corollary}{Corollaries}
\Crefname{cor}{Corollary}{Corollaries}
\crefname{prop}{Proposition}{Propositions}
\Crefname{prop}{Proposition}{Propositions}
\crefname{openprob}{Open Problem}{Open Problems}
\Crefname{openprob}{Open Problem}{Open Problems}
\crefname{remark}{Remark}{Remarks}
\Crefname{remark}{Remark}{Remarks}
\crefname{claim}{Claim}{Claims}
\Crefname{claim}{Claim}{Claims}

\title{Tight Bounds on the Cost of Adaptivity for the Meyerson Sketch}

\author{%
  Edith Cohen\thanks{Google Research and Tel Aviv University.
    \texttt{edith@cohenwang.com}}
  \and
  Elena Gribelyuk\thanks{Princeton University.
   \texttt{eg5539@princeton.edu}}
  \and
  Pasin Manurangsi\thanks{Google Research.
    \texttt{pasin@google.com}}
  \and
  Uri Stemmer\thanks{Tel Aviv University and Google Research.
    \texttt{u@uri.co.il}}%
}

\date{}

\begin{document}
\maketitle

\begin{abstract}
In online facility location, points arrive one at a time, and the
algorithm must either open a facility at the arriving point or route
the point to an existing facility. The Meyerson sketch opens a facility at each
arriving point with probability proportional to the point's distance
to the closest open facility, and requires no state beyond the set of open
centers. Due to its simplicity, space efficiency, and strong guarantees against the
offline optimum, the Meyerson sketch has become a workhorse of streaming and online
clustering. In many such applications, however, the set of open facilities are visible to the process that generates the stream, which can adaptively select future points
based on the algorithm's past random choices, voiding its classical
guarantees. In this work, we quantify the effect of such adaptivity. We compare an
adaptively generated run of the sketch against an \emph{oblivious
replay}, an independent execution, with fresh coins, on the very same generated sequence, and study the \emph{adaptivity ratio} of
expected adaptive cost to expected replay cost. We determine the
worst-case ratio in both directions: adaptivity can neither inflate nor
deflate the expected cost, or the number of open facilities, by more than an $O(\log\Delta/\log\log\Delta)$ factor,
where $\Delta$ is the aspect ratio of the input points (the ratio of the
largest to the smallest pairwise distance). This is asymptotically tight as
there are deterministic generators on the real line that inflate or deflate the cost by an $\Om(\log\Delta/\log\log\Delta)$ factor. We show that these robustness guarantees carry over to Meyerson-based sketches for approximate $k$-clustering with sketch size $O(k\,\mathrm{polylog}(n))$.

\end{abstract}

\section{Introduction}

In the \emph{online facility location} (OFL) problem, points of an underlying metric space
$(\Mm, d)$ arrive one at a time; upon each arrival, the algorithm must
irrevocably either open a facility at this point and pay a fixed price
$f$, or choose to route the point to the closest existing facility and pay for the routing distance. The total cost is defined as $C = f \cdot K + Q$, where $K$ denotes the total number of opened facilities (centers) and $Q$ is the total routing cost. The task of the algorithm is to open a set of facilities such that the total cost is close to the \emph{offline optimum} 
\begin{equation} \label{optf:eq}
\mathrm{OPT}_f(x) = \min_{F \subseteq \Mm}\, \bigl(f\,|F| +
\sum_{t} d(x_t, F)\bigr)
\end{equation}
on a given input sequence $x = (x_1, \ldots x_T)$.
A common performance measure of an online algorithm is its
\emph{competitive ratio}, defined as the supremum of
$\E[C]/\mathrm{OPT}_f(x)$ over all input sequences $x$ of length $n$ (with the expectation over the algorithm's coins).

\subsection{The Meyerson Sketch}

\citet{meyerson2001online} proposed a strikingly simple randomized rule
(\Cref{alg:meyerson}): for each arriving point, open a new facility at this point with probability
proportional to its current distance to the closest facility, capped at one. The rule
rent-or-buys in expectation: since a point at distance $d < f$ has opening probability $d/f$, the expected total facility opening cost matches the routing cost. Moreover, the sketch is space efficient in that its state is only the current set of opened facilities, which generally is much smaller than the input size, and decreases with $f$.

\begin{algorithm2e}[t]
\caption{The Meyerson sketch with facility price $f$.}
\label{alg:meyerson}
$S \gets \emptyset$\;
\ForEach{arriving point $x$}{
  $d \gets d(x, S)$ \tcp*{convention: $d(x, \emptyset) = +\infty$}
  draw a fresh coin $U \sim \mathrm{Unif}[0,1]$\;
  \eIf{$U \le \min\{1,\ d/f\}$}{
    open a center: $S \gets S \cup \{x\}$\;
  }{
    route $x$ to its nearest center in $S$, at cost $d$\;
  }
}
\end{algorithm2e}

Over the last few decades, a sequence of works has analyzed the competitive ratio of the Meyerson sketch under \textit{oblivious} inputs, where the input stream does not depend on the runtime choices of the algorithm and can be considered fixed in advance. In his seminal work, Meyerson proved that this algorithm is $O(\log n)$-competitive in expectation for worst-case sequences of length $n$ \citep{meyerson2001online}. \citet{fotakis2008competitive} later settled the
competitive ratio of the OFL problem at $\Theta(\log n / \log\log n)$. 
The lower bound
holds for every online algorithm, randomized or deterministic, with no
restriction on its state, even when the points lie on a line segment. The upper bound is achieved by the randomized space-efficient Meyerson sketch and  also by a deterministic algorithm that uses $\Theta(n)$ space and maintains all still-unsatisfied
past demands. 
Another thread analyzed the Meyerson sketch when the stream of points is presented in a uniformly-random order. In this setting, \citet{meyerson2001online} established that the sketch is $O(1)$-competitive and it was later shown to be $4$-competitive \citep{KNR23}; \citet{Lang:SODA2018} showed that the sketch is $\Theta(\log t/\log\log t)$-competitive on $t$-semirandom orders.
We refer to the survey of \citet{fotakis2011online} for a more thorough coverage of the surrounding literature.

Importantly, the definitions, and the $O(\log n / \log\log n)$ competitive ratio, and high probability bounds \citep{BMORST11,ShindlerWM11} 
(see \Cref{app:semiratio}), apply when $(\Mm,d)$ is a relaxation that satisfies all the axioms
of a metric, except that the triangle inequality is relaxed to the
$\alpha$-approximate form
$d(x,z) \le \alpha\bigl(d(x,y) + d(y,z)\bigr)$ for a constant
$\alpha \ge 1$. We refer to this relaxation as a \emph{semi-metric}. A very useful semi-metric, with $\alpha=2$, is the squared Euclidean distance, which is used in $k$-means clustering.  

Beyond the online facility location problem, the Meyerson sketch has found numerous applications as an online \emph{summarization} primitive: the center set, weighted by the number of associated points, constitutes an online summary of the full stream. 
One central application of the Meyerson sketch is to the $k$-clustering problem, where the objective constrains the number of
centers instead of pricing them:
\begin{equation} \label{optk:eq}
\mathrm{OPT}_k(x) \;=\; \min_{C \subseteq \Mm,\ |C| = k}\,
\sum_{t} d(x_t, C) .
\end{equation}
For the space $(\Mm,d)$, when $\alpha=1$, the objective is $k$-medians. With squared distances ($\alpha=2$), the objective is $k$-means. 
Streaming or online clustering algorithms \citep{CharikarOP03} apply the Meyerson sketch in phases: run the sketch until it exceeds a budget of centers; once the budget is exhausted, double the price, compress the opened centers into a weighted prefix for the next phase, and continue. Each center in the summary is weighted by the number of associated points. With $O(k\log^2 n)$ centers, the  service cost
is $O(1)$ times the $k$-median optimum; applying an offline
$O(1)$-approximation clustering algorithm to the weighted centers then yields
exactly $k$ centers at an $O(1)$-approximation overall
\citep{CharikarOP03}. This approach was later extended to $k$-means clustering, and sharpened to $O(k\log n)$ centers via high-probability
bounds \citep{BMORST11,ShindlerWM11,Lang:SODA2018}.

\subsection{Meyerson Sketch under Adaptive Inputs}

 A common assumption in the analysis of randomized algorithms, which also underlines prior works on the Meyerson sketch, is that the sequence of inputs or queries is chosen independently of the algorithm’s internal randomness. While natural in one-shot or offline settings, this \emph{oblivious inputs} assumption can be too strong when an algorithm is used repeatedly and its past outputs are visible. In such settings, future inputs may be chosen in response to previous answers of the algorithm, allowing the interaction history to reveal information about the algorithm’s random choices and invalidating correctness guarantees proved only for oblivious inputs. Motivated by these concerns, adaptive and adversarially chosen inputs have been studied extensively across statistical queries \citep{Freedman:1983,Ioannidis:2005,FreedmanParadox:2009,HardtUllman:FOCS2014,DworkFHPRR:STOC2015}, sketching and streaming algorithms \citep{MironovNS:STOC2008,HardtW:STOC2013,BenEliezerJWY21,DBLP:conf/nips/CherapanamjeriN20,HassidimKMMS20,WoodruffZ21,AttiasCSS21,BEO21,DBLP:conf/icml/CohenLNSSS22,TrickingHashingTrick:arxiv2022,AhmadianCohen:ICML2024,GribelyukLWYZ25,CohenNSSS:SODA2026,CohenGNS:ICML2026}, dynamic graph algorithms \citep{ShiloachEven:JACM1981,AhnGM:SODA2012,gawrychowskiMW:ICALP2020,GutenbergPW:SODA2020,Wajc:STOC2020,BKMNSS22}, and adversarial robustness in machine learning  \citep{szegedy2013intriguing,goodfellow2014explaining,athalye2018synthesizing,papernot2017practical}.

For many of the OFL applications, the algorithm's center set is not just an internal state, but is also published as a summary of the evolving input stream. Moreover, future input points may be chosen adaptively, in a way that depends on these previously published summaries.
As deterministic algorithms must satisfy correctness guarantees for any fixed input stream, deterministic algorithms automatically maintain correctness under adaptively chosen input sequences. In particular, the $\Theta(n)$ space deterministic OFL algorithm of  \citet{fotakis2008competitive} achieves the same behavior under adaptive inputs. However, all known sublinear-space OFL algorithms are randomized, the guarantees hold statistically for any fixed input under the algorithm randomness, and such analysis  may not transfer.

Adaptive attacks on sketching and streaming algorithms typically work by \emph{learning
persistent randomness}: for instance, the algorithm may sample a sketch matrix or hash function at the beginning of the stream, and
thereafter the algorithm is deterministic conditioned on the sampled random bits. In fact, adaptive attacks against such sketches operate by gradually learning the initial sampled bits, and then carefully designing hard queries on which this sampled randomness fails ~\citep{HardtW:STOC2013,BenEliezerJWY21,DBLP:conf/nips/CherapanamjeriN20,TrickingHashingTrick:arxiv2022,AhmadianCohen:ICML2024, GribelyukLinWoodruffYuZhou2024, GribelyukLWYZ25,CohenGNS:ICML2026}. 
In contrast, the structure of the Meyerson sketch is quite different: the sketch does not have any persistent randomness, as each decision is made using a fresh coin flip. It was previously observed that such randomized algorithms are often automatically robust~\citep{BravermanHMSSZ21,CohenSS:ITCS2026}, or at least can be meaningfully robustified~\cite{CohenGNS:ICML2026}. While the lack of persistent randomness may seem hopeful, each output of the algorithm reveals the entire current internal state of the algorithm (the center set), which is itself a function of the algorithm's previous coin flips. Thus, an adaptive generator which observes all of these internal states may adaptively select input points in order to steer the sketch toward a coin-correlated configuration which could not have been produced by an oblivious input sequence (with high probability). As a result, neither the adaptive attacks against sketches relying on persistent-randomness nor the
classical algorithm analyses for oblivious inputs apply to our setting. This motivates the following question.

\begin{question}\label{q:main}
Is the Meyerson sketch
adaptively robust? Quantitatively: by how much can an adaptive generator that
observes the sketch's open centers inflate or deflate its expected cost, relative
to an independent execution of the sketch, with fresh coins, on the
very same generated sequence?
\end{question}

The comparison in \Cref{q:main}, of an adaptive run versus \emph{oblivious
replay} on the same sequence, isolates the effect of adaptivity from
the difficulty of the sequence itself: an adaptive generator may of
course produce a stream that is simply hard, but that hardness is
charged to both runs. 

\subsection{Overview of the Model and Results}

An \emph{adaptive generator} interacts with one execution of the sketch, choosing each next point as a function of the transcript so far (including the current center set), and may also decide the stopping time adaptively. This interaction produces a random realized sequence $X=(x_1,\ldots,x_\tau)$. 
To isolate the effect of adaptivity from the difficulty of the resulting sequence itself, we compare this \emph{adaptive run} with an \emph{oblivious replay}: after the sequence $X$ has been generated, we run the sketch again from scratch on exactly the same realized sequence, using fresh independent coins.
Thus, conditioned on $X$, the replay is
therefore distributed exactly as the sketch on a fixed oblivious input.
For either run $\mathcal R\in \{\A,\OB\}$ ($\A$ for the adaptive run and $\OB$ for the oblivious replay), we track the center count $K^{\mathcal R}_\tau$, the routing cost
$Q^{\mathcal R}_\tau$, and the total cost $C^{\mathcal R}_\tau = fK^{\mathcal R}_\tau + Q^{\mathcal R}_\tau$. 

\subsubsection{Quantifying the price of adaptivity}
We study the adaptivity (adaptive-to-replay) ratio
$\E[C^{\A}_\tau]/\E[C^{\OB}_\tau]$. We compare \emph{expected costs}
since the realized ratio is noisy even on
oblivious inputs, as either run may open unusually few or unusually many centers just by the luck of its coins.
The adaptivity ratios for the center count and routing
components are defined analogously. We consider both directions: in particular, the \emph{inflation gap} is defined to be the supremum of the adaptivity
ratio over generators and stopping times, and the \emph{deflation
gap} is the supremum of its reciprocal.

Our main results resolve \Cref{q:main} with matching upper and lower bounds of $\Theta(\log\Delta/\log\log\Delta)$,  for both adaptivity gaps,  where $\Delta$  is  the \emph{aspect ratio} of the points (ratio of maximum to minimum non-zero distance between points in the stream).
In particular, this implies that the cost of adaptivity for the Meyerson sketch is relatively small and does not depend on the input length. 
Thus, the sketch is inherently robust to adaptive inputs as-is, up to this modest multiplicative factor, without the need to apply expensive robustification wrappers such as those based on differential privacy. 

Our upper bounds establish robustness by bounding the extent to which adaptivity can
inflate or deflate the expected cost or center count. We also provide high probability uniform bounds on the inflation or deflation:
\begin{theorem}[Robustness; informal; see \Cref{thm:upper,thm:upper_reversed,thm:upper_whp}]\label{thm:introupper}
Let $\Lambda \le \Delta$ denote the aspect ratio capped at the facility
price $f$, and let $\rho_\Lambda = O\!\bigl(\log\Lambda / \log\log\Lambda\bigr)$.
For every (semi-)metric space, every adaptive generator, and every stopping
time $\tau$,
\[
\E\bigl[C^{\A}_\tau\bigr] \;\le\; \rho_\Lambda\, \E\bigl[C^{\OB}_\tau\bigr]
\qquad\text{and}\qquad
\E\bigl[C^{\OB}_\tau\bigr] \;\le\; \rho_\Lambda\, \E\bigl[C^{\A}_\tau\bigr].
\]
Moreover, for every $\delta \in (0,\tfrac12]$, with probability at least
$1-\delta$, simultaneously for all $t \ge 0$,
\[
C^{\A}_t \;\le\; \rho_\Lambda\, C^{\OB}_t
+ \rho_\Lambda\, f \ln\tfrac1\delta
\qquad\text{and}\qquad
C^{\OB}_t \;\le\; \rho_\Lambda\, C^{\A}_t
+ \rho_\Lambda\, f \ln\tfrac1\delta .
\]
The same bounds hold for the center counts $K^{\A}_t, K^{\OB}_t$ in
place of the costs, with additive term $\rho_\Lambda \ln\tfrac1\delta$.
All upper bounds hold even for \emph{two-sided} adaptive generators, which observe the
replay's centers as well.\footnote{The default generators of our model are \emph{one-sided}: they observe only the adaptive run's centers, and the replay is an unobserved counterfactual with fresh coins. The class of two-sided adaptive generators is only an analytical tool.
}
\end{theorem}

For the lower bounds, we construct deterministic generators on a line segment with adaptivity ratios that match the upper bounds up to a constant factor:  
\begin{theorem}[informal; see \Cref{thm:anticert} and
\Cref{prop:reverse}]\label{thm:introlower}
There are deterministic generators on the real line, observing only
the adaptive transcript, that force inflation
$\E[C^{\A}_\tau] = \Om(\log\Delta/\log\log\Delta)\cdot
\E[C^{\OB}_\tau]$ or, respectively, deflation
$\E[C^{\OB}_\tau] = \Om(\log\Delta/\log\log\Delta)\cdot
\E[C^{\A}_\tau]$. Each generator forces this factor in the opening, routing, and total-cost components simultaneously.
\end{theorem}

As for the routing component of the cost, we show that the adaptivity ratio can be much larger:
\begin{prop}[informal; see \Cref{prop:starvesummary}]\label{prop:introstarve}
There are
(one-sided) generators which induce an adaptivity ratio
of $\Om(\Delta)$ for the routing cost. Moreover, for the stronger \emph{two-sided}
generators, no function of $\Delta$ bounds the routing ratio in
\emph{either direction}: there exist generators in Euclidean $\mathbb{R}^{D+1}$ which generate a stream of points with constant
aspect ratio and force
$\E[Q^{\A}_\tau] \ge 2^{\Omega(D)}\,\E[Q^{\OB}_\tau]$ or,
symmetrically, $\E[Q^{\OB}_\tau] \ge 2^{\Omega(D)}\,\E[Q^{\A}_\tau]$.
\end{prop}

\subsubsection{Competitive ratio}
The competitive ratio compares the cost to the offline optimum of the generated stream.
By combining the robustness expectation upper bound with the optimal oblivious-sequence competitive ratio \citep{fotakis2008competitive} (\Cref{thm:semiratio} for semi-metrics), we bound the competitive ratio under adaptively generated streams:
\begin{cor}[Adaptive competitiveness; informal; see \Cref{cor:transfer}]
\label{cor:introopt}
For every semi-metric space, adaptive generator, and stopping time $\tau$,
\[
\E\bigl[C^{\A}_\tau\bigr]
\;=\; O\Bigl(\tfrac{\log\Lambda}{\log\log\Lambda}\cdot
\tfrac{\log n}{\log\log n}\Bigr)\cdot
\E\bigl[\mathrm{OPT}_f(x_1,\dots,x_\tau)\bigr] .
\]
\end{cor}
Here the factor $\tfrac{\log\Lambda}{\log\log\Lambda}$ is the cost of
adaptivity, and the $\tfrac{\log n}{\log\log n}$ factor is the optimal
competitive ratio for oblivious streams.

\subsubsection{Streaming $k$-clustering}
In streaming $k$-clustering, the task of the algorithm is to maintain a small sketch of the input stream at at each time $t$, to report a set of $k$ centers on the prefix of the stream so far, such that the clustering cost approximates the optimum \eqref{optk:eq}. We present an adaptively robust $k$-clustering sketch (\Cref{alg:multiprice}) that  uses $O(\log n)$ concurrent capacitated copies of the Meyerson sketch. Compared with the known guarantees for oblivious streams, the approximation factor of our algorithm only degrades by the Meyerson sketch's price of adaptivity.
\begin{thm}[Online adaptive $k$-clustering; informal; see
\Cref{thm:kmedianmain}]
\label{thm:introkmedian}
There is an online $k$-clustering sketch ($k$-median, and $k$-means via squared distances) built from $O(\log n)$ copies of the Meyerson sketch and a black-box offline $\rho$-approximation algorithm for the respective objective. 
On a stream of $n$ points with 
aspect ratio $\Delta$\footnote{We may assume $\Delta = n^{O(1)}$ without loss of generality; see the footnote to \Cref{thm:kmedianmain}.}, it publishes, after every point on demand, a set of $k$ centers together with a cost estimate, using 
$O\bigl(k \cdot \tfrac{\log\Delta}{\log\log\Delta} \cdot \log^2 n\bigr) = O(k\,\log^3 n/\log\log n)$ words of memory. Against any adaptive generator that observes all published center sets, the algorithm guarantees:
\begin{itemize}
\item Deterministically, the published estimate is a correct upper bound on the clustering cost of the published centers;
\item With high probability, simultaneously at every time $t$, the published centers cost at most
$O\bigl(\rho \cdot \tfrac{\log\Delta}{\log\log\Delta}\bigr) \cdot \mathrm{OPT}_k(P_t)$, whenever the prefix has more than $k$ distinct locations. 
\end{itemize}
The worst-case update time is linear in the space, measured in distance
evaluations, and reporting runs the offline algorithm on $O(\log n)$ weighted
instances with $O(k\,\mathrm{polylog}\, n)$ points in total.
\end{thm}

Note that the factor $\tfrac{\log\Delta}{\log\log\Delta}$ is the price of adaptivity, and a $O(\rho)$-approximation is known for oblivious streams \citep{CharikarOP03,BMORST11}.
Previously, \citet{BravermanHMSSZ21} obtained an adversarially robust $(1+\epsilon)$-approximate clustering algorithm by applying the merge-and-reduce framework to strong coresets for the \emph{Euclidean} metric. On the other hand, our construction gives a different tradeoff: it is built directly from the Meyerson summary and applies more broadly to general metrics (and semi-metrics, with constants depending on $\alpha$). Our sketch uses $O(\kappa_\Delta k\log^2 n)$ words, and outputs a solution up to an $O(\rho\,\kappa_\Delta)$ approximation factor. We discuss potential approaches to improve our result in \Cref{thm:introkmedian} in \Cref{op:samplinglayer}.
\subsection{Technical Overview}

We first observe a helpful property of the Meyerson sketch that allows us to account for center
openings as a proxy for the total cost: the expected cost of each step is within a factor of $2$ of $f$
times the opening probability of this point. In fact, this property is implicit in the previous analyses since \citet{meyerson2001online}. We show that this
property remains valid even when the inputs and stopping time are generated adaptively  (\Cref{sec:prelims}); therefore, adaptivity-ratio bounds for center
counts transfer to total costs, within a factor of $4$. With this in mind, we now describe the ideas behind our bounds.

\paragraph{Inflation upper bound (\Cref{thm:introupper}; \Cref{thm:upper}):}
 Our goal is to upper bound the expected number of
centers opened by the adaptive run by the expected number opened in the replay execution, up to some small multiplicative factor. To this end, our proof proceeds by carefully charging center openings in the adaptive run to openings in the replay execution as follows:
\begin{itemize}
\item \textbf{Accounting identity.} First, note that the expected number of centers opened is simply the sum of the
opening probabilities. At each step, the coin is fair given everything seen so far, and this holds under an adaptively chosen sequence of points and an adaptive choice of the
stopping time. Therefore, it suffices to relate the sums of opening probabilities, which are determined by the distances to the
closest open center that the two runs see on the same sequence.
\item \textbf{Cheap steps.} Suppose that the shortest distance from the new point to some existing center in the adaptive execution is within a factor of $3$ of the shortest distance in the replay execution. In fact, such ``cheap'' steps can be easily accounted for by noting that in expectation, the adaptive run will open at most $3$ times as many centers as the replay. Indeed, this $O(1)$ factor blow-up is acceptable to us, as we aim to show an $O\left(\frac{\log \Delta}{\log \log \Delta}\right)$ upper bound for the adaptivity ratio.
\item \textbf{Dear steps.} The tricky steps are those when the new point
is far from every adaptive center yet close to some replay center $w$:
in this case, the adaptive run will open new centers with much higher probability, while the replay will likely elect to pay the routing cost (and not open a new center). For each such step, we charge this ``dear'' opening to the closest replay center $w$. 
\item \textbf{Few charges per center.} The heart of the analysis is
showing that no replay center can absorb too many charges, since
each charge ``blocks'' its own scale (\Cref{fig:halving}): the opening
plants an adaptive center next to $w$, say at distance $r$. After this, any later point which is dear and whose closest replay center $w$ must be at distance less than $r/2$ from $w$. Therefore, the radius for potential future dear charges to $w$ decreases by a factor of $2$ on each such charge. Since distances matter only between $f/\Lambda$ and $f$
(above distance $f$, every point opens a center with probability $1$), this implies that there can be at most 
$\log_2\Lambda + O(1)$ charges per replay center.
\item \textbf{A tighter bound.}
As described, with a dearness threshold $\theta=3$, the cheap steps are accounted for by the replay openings up to an $O(1)$ factor, whereas the dear steps have a $O(\log_2\Lambda)$ factor blowup. The inflation upper bound can be tightened to $O(\log \Lambda/\log\log\Lambda)$ by selecting the threshold $\theta$ such that the cheap and dear costs are balanced.
\end{itemize}
This charging argument echoes the ring decompositions of classical OFL analyses
\citep{meyerson2001online,fotakis2008competitive,Lang:SODA2018}, but in our setting the charges are assigned to the replay's random centers rather than a fixed optimal
solution. Since it relies only on the freshness of the coins in each step, our
bound holds even for against powerful generators which observe the replay's
centers as well.

\paragraph{Deflation upper bound (\Cref{thm:introupper}; \Cref{thm:upper_reversed}):}
The inflation bound holds even for two-sided generators, which observe the outputs of both runs, and for stopping times that depend on both runs. This class is symmetric in the two runs, so by exchanging their roles, we directly obtain the deflation bound as a corollary of the inflation bound argument above (See \Cref{rem:exchange}).
We note that usual (one-sided) generators which only observe outputs from the adaptive run do not have this symmetry, so the two attacks below require separate constructions.

\paragraph{Inflation attack (\Cref{thm:introlower}; \Cref{thm:anticert}):}
We describe the generator that achieves the lower bound. The construction is parametrized by $R$, which will be related to the aspect ratio of the constructed sequence of points.
\begin{itemize}
\item \textbf{Finding an Anti-certificate.} The generator designs a point
$z^*$ that the adaptive run does not hold as a center, but we can certify with probability at least $\approx 1-O(1/R)$ that the replay does hold $z^*$ as a center (See~\Cref{fig:preparation}). 
A certification attempt opens at most a \textit{constant} number of centers and succeeds with constant probability.
\item \textbf{Extraction Phase.} Next, the adaptive generator presents points on a sequence of $R$ locations
(points) that get geometrically closer to $z^*$: each location in the sequence is a factor of $R$ closer to $z^*$ than the previous one. For each location  in this sequence, the adaptive generator issues copies of the same point until a center is opened there (See~\Cref{fig:extraction}). 

\item \textbf{One center per point vs.\ $O(1)$ in total.} For each point, the closest replay center is $z^*$, even if the replay holds prior locations. On the other hand, the closest adaptive center is at least a factor $R$ further. Now, we examine the outcomes of the two executions: observe that the adaptive run opens a center at every location and $R$ centers in total, while the replay opens only $1/R$ centers per location (in expectation), resulting in a total of only $O(1)$ centers in expectation. Therefore, the adaptivity ratio of this input sequence is $\Om(R)$.
\item \textbf{Aspect ratio:}  The aspect ratio of this layout is $\Delta \approx R^{\Theta(R)}$. By expressing the adaptivity ratio in terms of the aspect ratio $\Delta$, we obtain
$\Om(\log\Delta/\log\log\Delta)$.
\end{itemize}

\paragraph{Deflation attack (\Cref{thm:introlower}; \Cref{prop:reverse}):}
The adaptive generator for the deflation attack mirrors the approach in the previous section with some necessary adaptations. In particular, the attack finds a point $z^*$
that the \emph{adaptive} run holds as a center, while the replay does not hold this point with probability at least $\tfrac12$. Note that this is easy to do by observing the outputs of the adaptive execution (in contrast, for the inflation attack to work we needed high certainty that the \textit{replay} holds a particular point $z^*$). The generator then exploits $z^*$ by using exactly the same layout as in the inflation attack: concretely, the generator inserts points into $R$ locations which get geometrically closer to $z^*$, decreasing the distance of each new location to $z^*$ by a factor of $R$. At each such location, the adaptive generator presents many copies of the same point, such that the following holds: given that the adaptive run previously opened a center at $z^*$, the adaptive execution is very unlikely to open a new center at these presented points; on the other hand, the replay's closest center is at least a factor of $R$ further, and thus the replay is $R$ times more likely to open a center. Note that
in contrast with the inflation attack, the adaptive generator cannot observe the replay openings. To handle this, we set the number of copies a point (in each location) to be fixed so that the replay execution will eventually open a center with probability $\approx \frac12$. 
Then, at each location there are $\Omega(1)$ replay centers, and only $O(1/R)$ adaptive centers in expectation. Over all $R$ locations, it follows that the adaptive generator induces a gap of $\Omega(R)$ versus $O(1)$ openings between the two runs. Using the same calculation as before, we can express $R$ in terms of the aspect ratio $\Delta \approx R^R$ to obtain the adaptivity gap 
$R = O(\log\log\Delta/\log\Delta)$.%

\paragraph{Routing cost.}
The routing component of the cost is dominated by 
the openings component $\E[Q_\tau] \le f\,\E[K_\tau]$, also under adaptive generation (see 
\Cref{sec:prelims}).
The inflation and deflation attacks (\Cref{thm:introlower}) extend to the routing
component: On each location, the copies issued before the opening pay their distance in routing;
the replay, holding $c^*$, routes each copy almost for free.
The same extraction sequence therefore that separates the openings separates the routing costs as well and with a similar ratio. The stronger bounds in \Cref{prop:introstarve} are obtained by constructing generators in a high Euclidean dimension or a synthetic metric. 
The generators
\emph{starve} the
replay's routing cost while the adaptive run's total cost and center
count remain no larger than the replay's
($\E[C^{\A}_\tau] \le \E[C^{\OB}_\tau]$ and
$\E[K^{\A}_\tau] \le \E[K^{\OB}_\tau]$). We call this technique \emph{relocating} cost between the
components. For completeness, we present these attacks in\Cref{app:starvation} of the Appendix, with high-level descriptions preceding the formal analysis. 

\paragraph{Roadmap.} The rest of the paper is organized as follows: First, we formalize the model in \Cref{sec:model}, and develop the fair-count identity and cost-center equivalence in \Cref{sec:prelims}. Next, we state and prove the upper bound and cost transfer corollaries in 
\Cref{sec:upperproof}.
\Cref{sec:anticertproof} presents the inflation attack and \Cref{sec:lower_reverse} gives the
deflation attack.  
\Cref{sec:onlinesummary} presents online $k$-clustering sketches which are designed by using the Meyerson sketch as a subroutine. Finally, in \Cref{sec:openproblems} we conclude with a
discussion and open problems. The appendices contain extended results for non-adaptive Meyerson (\Cref{app:semiratio}) and routing-starvation constructions (\Cref{app:starvation}).

\section{The Model}\label{sec:model}

We state the model for a mild relaxation of a metric space, which we call a \textit{semi-metric}.

\begin{defn}[Semi-metric]\label{def:semimetric}
A \emph{semi-metric space} $(\Mm, d)$ consists of a set $\Mm$ and a
map $d : \Mm \times \Mm \to [0,\infty)$ that is symmetric, with
$d(x,x) = 0$ and $d(x,y) > 0$ for distinct $x, y$, and satisfies the
\emph{$\alpha$-approximate triangle inequality}
\[
d(x,z) \;\le\; \alpha\bigl(d(x,y) + d(y,z)\bigr)
\qquad\text{for all } x, y, z \in \Mm ,
\]
for a constant $\alpha \ge 1$. A metric is the case $\alpha = 1$;
squared Euclidean distance satisfies $\alpha = 2$.
\end{defn}

Throughout, $(\Mm, d)$ is a semi-metric space with parameter
$\alpha$, and $f > 0$ denotes the facility price.
The Meyerson sketch (\Cref{alg:meyerson}) maintains a center set
$S$, which is initially empty; upon receiving a point $x$, the sketch opens a center at $x$ with
probability $\min\{1,\, d(x, S)/f\}$, and otherwise routes $x$ at cost
$d(x, S)$. We adopt the convention $d(x, \emptyset) = +\infty$, so the first
point is opened with probability $1$. Additionally, we use the convention $0 \cdot \infty = 0$, as this simplifies the accounting of the contributions of deterministically opened points to the routing cost.

\medskip\noindent\textbf{Randomness.}
Let $(U_t)_{t\ge 1}$ and $(U'_t)_{t\ge 1}$ be independent sequences of i.i.d.\
$\mathrm{Unif}[0,1]$ variables, and let $W$ be an independent variable (the
generator's private randomness). The coins $U$ drive the \emph{adaptive} run;
the coins $U'$ drive an independent \emph{oblivious replay} on the same
input sequence.

\medskip\noindent\textbf{Adaptive interaction.}
An \emph{adaptive generator} is a sequence of measurable maps producing
\[
x_t \;=\; g_t\bigl(W;\ (x_1, o_1), \dots, (x_{t-1}, o_{t-1})\bigr)
\;\in\; \Mm \cup \{\bot\},
\]
where $\bot$ is a halting symbol (once emitted, all later points
are $\bot$ and incur no cost). For $x_t \in \Mm$ the adaptive run computes
$d^{\A}_t = d(x_t, S^{\A}_{t-1})$ and opens iff
$o_t = \ind{\{U_t \le \min\{1, d^{\A}_t / f\}\}} = 1$, with
$S^{\A}_0 = \emptyset$. Thus $x_t$ is a function of $(W, U_{<t})$ only:
the generator observes the adaptive run's transcript (equivalently, its center
set, from which the indicators $o_s$ are reconstructible) but is independent
of the replay coins $U'$.
The generator together with the adaptive coins $U$ therefore induces a 
random input sequence $X=(x_1,x_2,\ldots)$ (including its stopping time).
Crucially, this sequence is generated independently of the replay coins
$U'$.

\medskip\noindent\textbf{Oblivious replay.}
The oblivious run executes the sketch on this same random sequence $X$, using the independent coins $U'$. Equivalently, one may first condition
on the realized sequence generated by the adaptive interaction and then
run a fresh copy of the sketch on that fixed sequence.
Conditionally on the sequence, the replay is
distributed exactly as the sketch on a fixed input:
For $x_t \ne \bot$, $d^{\OB}_t = d(x_t, S^{\OB}_{t-1})$ and
$o'_t = \ind{\{U'_t \le \min\{1, d^{\OB}_t/f\}\}}$, with
$S^{\OB}_0 = \emptyset$. 

\medskip\noindent\textbf{Filtrations and stopping times.}
Let $\Ff_t = \sigma(W, U_{\le t})$ (the adaptive transcript filtration;
$x_{t+1}$ is $\Ff_t$-measurable) and
$\Gg_t = \sigma(W, U_{\le t}, U'_{\le t})$ (the joint filtration). Unless
stated otherwise, stopping times are with respect to $(\Ff_t)$; every
$(\Ff_t)$-stopping time is a $(\Gg_t)$-stopping time. Stopping times may be
infinite; all cost processes below are nondecreasing, so values at infinite
times are defined as monotone limits in $[0, \infty]$.

\noindent
Two natural $(\Ff_t)$-stopping rules are worth keeping in mind. A
\emph{fixed horizon} $T$ corresponds to the deterministic stopping time
$\tau=T$. A simple adaptive stopping rule is to run the sketch until it
opens $k$ centers: For
$k \ge 1$ the \emph{budget-exhaustion time} is
$\tau_k = \inf\{t : K^{\A}_t = k\} \in \mathbb{N}\cup\{\infty\}$.

\medskip\noindent\textbf{One-sided and two-sided generators.}
A generator as above observes only the adaptive transcript. We call such generators \emph{one-sided}:
\begin{defn}[One-sided generator]\label{def:onesided}
A one-sided generator chooses $x_t$ as a measurable
function of $(W, U_{<t})$, i.e., $x_t$ is $\Ff_{t-1}$-measurable. The
private randomness $W$ is independent of the coin sequences $(U, U')$.
\end{defn}
We also consider a stronger class: 
\begin{defn}[Two-sided generator]\label{def:twosided}
A \emph{two-sided} generator may choose $x_t$ as a measurable function
of $(W, U_{<t}, U'_{<t})$, i.e., $x_t$ is $\Gg_{t-1}$-measurable; in
particular, it also observes the replay's open centers at each time $t$.
\end{defn}

\begin{remark}[Exchanging the runs]\label{rem:exchange}
The replay is a counterfactual re-execution introduced for the analysis, so no environment interacting with the deployed sketch can condition on its coins. Importantly, we emphasize that one-sided generators are the operationally meaningful class. For a two-sided generator, the replay is no longer oblivious, since its coins influence the stream, and in fact the two runs become interchangeable: exchanging the coin sequences $U$ and $U'$ leaves the distribution unchanged, maps each run to the other, and the generator remains two-sided and $\tau$ remains a $(\Gg_t)$-stopping time. A robustness \emph{upper} bound is stronger when it holds against the more powerful two-sided generators and $(\Gg_t)$-stopping times, and an inflation bound directly implies the deflation bound, with the runs exchanged (\Cref{thm:upper_reversed}). An attack (\emph{lower} bound) is stronger when carried out by a one-sided generator which can only observe the outputs of the adaptive execution; the exchange maps a one-sided generator to one observing only the replay, so the attacks in the two directions require separate constructions.
\end{remark}

\medskip\noindent\textbf{Opening probabilities and costs.}
For $x_t \ne \bot$ define
\[
a_t \;=\; \frac{\min\{f,\, d^{\A}_t\}}{f},
\qquad
b_t \;=\; \frac{\min\{f,\, d^{\OB}_t\}}{f},
\]
to be the center opening probabilities of the two executions; both are $\Gg_{t-1}$-measurable,
and $o_t, o'_t$ are conditionally $\Bern(a_t)$, $\Bern(b_t)$ given
$\Gg_{t-1}$, since $U_t, U'_t$ are drawn independently (For $x_t = \bot$ set
$a_t = b_t = o_t = o'_t = 0$). Center counts, routing (connection) costs,  and facility-location costs up to
time $t \in \mathbb{N}\cup\{\infty\}$ are given by
\[
K^{\A}_t = \sum_{s\le t} o_s, \qquad Q^{\A}_t \;=\; \sum_{s\le t} (1 - o_s)\, d^{\A}_s,
\qquad
C^{\A}_t = f\,K^{\A}_t + Q^{\A}_t = \sum_{s \le t}\bigl[(1-o_s)\, d^{\A}_s + o_s f\bigr],
\]
and $K^{\OB}_t$, $Q^{\OB}_t$, $C^{\OB}_t$  are defined with $o'_s, d^{\OB}_s$ analogously.

\medskip\noindent\textbf{Aspect ratios.}
Let $d_{\min}$ and $d_{\max}$ be the smallest and largest distances between
distinct points of the input stream, let $\Delta = d_{\max}/d_{\min}$, and define
\[
\Lambda \;=\; \frac{\min\{f,\, d_{\max}\}}{d_{\min}} \;\le\; \Delta
\]
to be the \emph{capped} aspect ratio (distances above $f$ are equivalent for the
sketch, whose opening probability is saturated at $1$). Write
$\lgp x = \max\{0, \log_2 x\}$.

\medskip\noindent\textbf{Adaptivity gaps.} Fix a generator, which carries its ambient data, the metric space
it presents points from, and the facility opening cost  $f > 0$. For any cost component $X \in \{K, C, Q\}$ (center count, total cost,
and routing cost respectively) and a stopping time $\tau$, the \emph{adaptivity ratio} is
\[
\frac{\E\bigl[X^{\A}_\tau\bigr]}{\E\bigl[X^{\OB}_\tau\bigr]}\,,
\]
and is defined whenever the denominator is positive. 

For a cost component $X \in \{K, C, Q\}$, the \emph{adaptivity
inflation gap} of $X$ at aspect ratio $\Delta$ is the supremum of the
adaptivity ratio of $X$, and the \emph{adaptivity deflation gap} is
the supremum of its inverse --- in both cases over all one-sided
generators and stopping times $\tau$ such that the presented stream
has aspect ratio at most $\Delta$ almost surely. Both gaps are at
least $1$, witnessed by oblivious generators. Two-sided gaps are
defined identically, with the supremum over two-sided generators
(\Cref{def:twosided}); they are at least their one-sided
counterparts. When the component is not specified we mean the total
cost, and we abbreviate \emph{inflation gap} and \emph{deflation
gap}.

\begin{remark}[Randomized vs.\ deterministic generators]
By conditioning on $W$, it suffices to prove all upper bounds for
deterministic generators; all of our lower-bound constructions are
deterministic.
\end{remark}

\section{Costs, Counts, and the Fair-Count Identity}\label{sec:prelims}
The classic analysis of the Meyerson
sketch~\citep{meyerson2001online} relies on the facts that  the expected number of opened centers equals the sum of the
opening probabilities (linearity of expectation), and that each step's expected cost is within a factor of $2$ of $f$ times its opening
probability. With adaptivity, we need to use the
martingale form of the same statement: the count minus its running sum of conditional opening probabilities has
expectation zero at every stopping time. We call the resulting identity
the \emph{fair-count identity}.

We use the standard martingale machinery: Fix a filtration $(\Gg_t)_{t \ge
0}$, an increasing family of $\sigma$-algebras; $\Gg_t$ represents the
history through step $t$ (transcript including both
runs' coins). A process $(Y_t)$ is \emph{adapted} if $Y_t$ is
$\Gg_t$-measurable for every $t$ (determined by the history through step $t$), and
\emph{predictable} if $Y_t$ is $\Gg_{t-1}$-measurable for every $t$ (determined
already by the history through step $t-1$, before the step-$t$ coins are
flipped). 
In our setting: the opening probability $a_t$ is predictable
as it is a function of the new point $x_t$ and the current center set, which are both fixed before the step-$t$ coin is flipped. On the other hand, the opening indicator $o_t$ is adapted but not predictable, since it depends on the outcome of the coin toss at time $t$.
An adapted process $(M_t)$ with $\E|M_t| < \infty$ is a
\emph{martingale} if $\E[M_t \mid \Gg_{t-1}] = M_{t-1}$ for all $t$:
conditionally on the past, its increments have mean zero. A random time
$\tau$ is a \emph{stopping time} if the event $\{\tau \le t\}$ is
$\Gg_t$-measurable for every $t$: that is, the decision to stop by time $t$
may use only the history through $t$.

\begin{fact}[Optional stopping; Doob~\cite{Doob53}, see
also~{\cite[Thm.~4.8.5]{Durrett19}}]\label{fact:doob}
Let $(M_t)_{t \ge 0}$ be a martingale with respect to $(\Gg_t)$ and let
$\sigma$ be a bounded stopping time. Then $\E[M_\sigma] = \E[M_0]$.
\end{fact}

\begin{lem}[Fair-count identity and cost--center equivalence]
\label{lem:equiv}
For every generator, one-sided or two-sided (\Cref{def:twosided}), every
$(\Gg_t)$-predictable process $h_t \in [0,1]$, and every stopping time
$\tau$,
\begin{equation}\label{eq:faircount}
\E\Bigl[\sum_{t\le\tau} o_t h_t\Bigr] = \E\Bigl[\sum_{t\le\tau} a_t h_t\Bigr],
\qquad
\E\Bigl[\sum_{t\le\tau} o'_t h_t\Bigr] = \E\Bigl[\sum_{t\le\tau} b_t h_t\Bigr],
\end{equation}
with values in $[0,\infty]$; in particular
$\E[K^{\A}_\tau] = \E[\sum_{t\le\tau} a_t]$ and
$\E[K^{\OB}_\tau] = \E[\sum_{t\le\tau} b_t]$. Moreover, for either run,
\begin{equation}\label{eq:costcenter}
f\,\E[K_\tau] \;\le\; \E[C_\tau] \;\le\; 2f\,\E[K_\tau].
\end{equation}
\end{lem}

\begin{proof}
Fix a predictable $h_t \in [0,1]$. The only thing to verify is that the
opening probabilities are predictable and are the conditional means of
the opening indicators. Whether the generator is one-sided or two-sided,
$x_t$ is $\Gg_{t-1}$-measurable, and the center sets $S^{\A}_{t-1}$,
$S^{\OB}_{t-1}$ are functions of $(x_s, o_s)_{s<t}$ and
$(x_s, o'_s)_{s<t}$; hence $a_t$ and $b_t$ are $\Gg_{t-1}$-measurable.
Moreover, the random bits $U_t$ and $U'_t$ are independent of
$\Gg_{t-1} = \sigma(W, U_{<t}, U'_{<t})$, so
$\E[o_t \mid \Gg_{t-1}] = a_t$ and $\E[o'_t \mid \Gg_{t-1}] = b_t$.
Consequently $M_T = \sum_{t \le T}(o_t - a_t)h_t$ is a
$(\Gg_t)$-martingale. Applying \Cref{fact:doob} at the bounded stopping
time $\tau \wedge T$ gives
$\E[\sum_{t\le\tau\wedge T} o_t h_t] = \E[\sum_{t\le\tau\wedge T} a_t h_t]$
for every finite $T$; letting $T \to \infty$, both sides converge
monotonically (as $h_t \ge 0$) to the first identity
of~\eqref{eq:faircount}. The second is identical with
$(o'_t, b_t, U'_t)$ in place of $(o_t, a_t, U_t)$.

For~\eqref{eq:costcenter}, condition on $\Gg_{t-1}$: the expected cost
of step $t$ is $m_t = (1 - q_t)\, d_t + q_t f$, where
$q_t = \min\{1, d_t/f\}$ is the run's opening probability and $d_t$ its
current distance. If $d_t \ge f$ then $q_t = 1$ and
$m_t = f = \min\{f, d_t\}$. If $d_t < f$ then
$m_t = d_t + \tfrac{d_t}{f}(f - d_t) \in [\,d_t,\, 2 d_t\,]$. In both cases
\[
\min\{f, d_t\} \;\le\; m_t \;\le\; 2 \min\{f, d_t\} \;=\; 2 f q_t .
\]
Summing conditional expectations (monotone convergence as above) gives
$\E[C_\tau] = \E[\sum_{t\le\tau} m_t]$, and the claim follows from the case
$h \equiv 1$ of~\eqref{eq:faircount}.
\end{proof}

\begin{cor}\label{cor:equiv}
For every generator (one- or two-sided) and stopping time $\tau$,
\[
\frac{1}{2}\cdot
\frac{\E[K^{\A}_\tau]}{\E[K^{\OB}_\tau]}
\;\le\;
\frac{\E[C^{\A}_\tau]}{\E[C^{\OB}_\tau]}
\;\le\;
2\cdot \frac{\E[K^{\A}_\tau]}{\E[K^{\OB}_\tau]} .
\]
Hence bounds on the adaptivity gap for center count transfer to total
cost and vice versa, losing a factor of at most $4$.
\end{cor}

Since $C_\tau = f K_\tau + Q_\tau$, \eqref{eq:costcenter} also yields,
for either run,
\begin{equation}\label{eq:routingdom}
\E\bigl[Q_\tau\bigr] \;\le\; f\,\E\bigl[K_\tau\bigr].
\end{equation}
In particular, the routing cost component of the total cost never
dominates the center-count cost component.

\section{Robustness Guarantees} \label{sec:upperproof}

In this section we establish that the adaptivity inflation and deflation gaps of the center count are bounded by 
$\kappa_\Lambda$, and of the total cost by at most $2 \kappa_\Lambda$, where $\kappa_\Lambda$ is defined as follows:

\begin{fact}[Threshold optimization]\label{fact:theta}
Let $\alpha \ge 1$ be the semi-metric parameter (\Cref{def:semimetric}).
For $\Lambda \ge 2$ let
\begin{equation} \label{CLambda:eq}
    \kappa_\Lambda \;=\; \inf_{\theta > 2\alpha}\,
\Bigl(\theta + \log_{(\theta-\alpha)/\alpha}\Lambda + 2\Bigr).
\end{equation}
Then $\kappa_\Lambda \le 3\alpha + 2 + \log_2\Lambda$, and
$\kappa_\Lambda = O(\alpha\log\Lambda/\log\log\Lambda)$; explicitly,
\[
\kappa_\Lambda \;\le\; \frac{4\alpha\log_2\Lambda}{\log_2\log_2\Lambda}
\qquad\text{whenever } \log_2\Lambda \ge 16 .
\]
For a metric ($\alpha = 1$), \eqref{CLambda:eq} is replaced by
$\inf_{\theta > 2}\bigl(\theta + \log_{\theta-1}\Lambda + 2\bigr)$.
\end{fact}

\begin{proof}
The first bound is the choice $\theta = 3\alpha$, with base
$(\theta - \alpha)/\alpha = 2$. For the second, write
$L = \log_2\Lambda \ge 16$ and take
$\theta^* = \alpha\bigl(1 + \lceil\sqrt{L}\,\rceil\bigr)$; note
$\theta^* > 2\alpha$ since $L \ge 4$. Since
$(\theta^* - \alpha)/\alpha \ge \sqrt{L}$,
\[
\log_{(\theta^*-\alpha)/\alpha}\Lambda
\;=\; \frac{L}{\log_2\bigl((\theta^* - \alpha)/\alpha\bigr)}
\;\le\; \frac{L}{\log_2 \sqrt{L}}
\;=\; \frac{2L}{\log_2 L}\,,
\]
and $\theta^* + 2 \le \alpha(\sqrt{L} + 2) + 2 \le \alpha(\sqrt{L} + 4)$.
It remains to check
$\alpha(\sqrt{L} + 4) + 2L/\log_2 L \le 4\alpha L/\log_2 L$; dividing
by $\alpha \ge 1$, it suffices that
$\sqrt{L} + 4 \le 2L/\log_2 L$ for all $L \ge 16$, i.e.\ that
$g(L) = 2L/\log_2 L - \sqrt{L}$ satisfies $g(L) \ge 4$. At $L = 16$,
$g(16) = 8 - 4 = 4$ with equality. For $L \ge 16$, $g$ is
nondecreasing:
\[
g'(L)
\;=\; \frac{2\ln 2\,(\ln L - 1)}{\ln^2 L} - \frac{1}{2\sqrt{L}}
\;\ge\; 0,
\]
since the inequality rearranges to
$4\ln 2 \cdot \sqrt{L}\,(\ln L - 1) \ge \ln^2 L$, and using
$\sqrt{L} \ge \ln L$ (valid for $L \ge 16$) it suffices that
$4\ln 2\,(\ln L - 1) \ge \ln L$, i.e.\
$\ln L \ge \tfrac{4\ln 2}{4\ln 2 - 1}$. As the right hand side is $ <2$, this inequality holds for
$L \ge 16$. Hence
$\theta^* + \log_{(\theta^*-\alpha)/\alpha}\Lambda + 2 \le 4\alpha L/\log_2 L$.
\end{proof}

We note that the infimum in \eqref{CLambda:eq} is attained: the
objective is continuous on $(2\alpha,\infty)$ and diverges as $\theta \to 2\alpha^+$
and as $\theta \to \infty$. We write $\theta_\Lambda$ for a minimizer;
trivially $\theta_\Lambda \le \kappa_\Lambda$.

\subsection{The Dear Charging Lemma}

The proofs of the robustness theorems utilize the following charging argument. Call a step \emph{dear} if the adaptive opening probability exceeds $\theta$ times the opening probability in the replay execution. We deterministically bound the number of adaptive centers opened at dear steps by a multiple of the total number of centers held by the replay.

\begin{figure}[t]
\centering
\begin{tikzpicture}[line cap=round]
  \fill[gray!12,even odd rule] (0,0) circle (3.9) (0,0) circle (1.5);
  \fill[gray!25,even odd rule] (0,0) circle (1.5) (0,0) circle (0.625);
  \draw (0,0) circle (3.9);
  \draw[dashed] (0,0) circle (1.5);
  \draw[dashed] (0,0) circle (0.625);
  \draw[dotted] (0,0) circle (0.2);
  \fill (0,0) circle (2pt);
  \node[below right=1pt] at (0,0) {$w$};
  \coordinate (x1) at (115:3.0);
  \draw[dashed] (0,0) -- (x1) node[midway, above, sloped] {\small $\beta_1$};
  \draw[thick] ($(x1)+(-0.11,-0.11)$) -- ($(x1)+(0.11,0.11)$);
  \draw[thick] ($(x1)+(-0.11,0.11)$) -- ($(x1)+(0.11,-0.11)$);
  \node[above left=1pt] at (x1) {$x_{t_1}$};
  \coordinate (x2) at (205:1.25);
  \draw[dashed] (0,0) -- (x2) node[midway, above, sloped] {\small $\beta_2$};
  \draw[thick] ($(x2)+(-0.11,-0.11)$) -- ($(x2)+(0.11,0.11)$);
  \draw[thick] ($(x2)+(-0.11,0.11)$) -- ($(x2)+(0.11,-0.11)$);
  \node[below left=1pt] at (x2) {$x_{t_2}$};
  \coordinate (x3) at (350:0.5);
  \draw[thick] ($(x3)+(-0.11,-0.11)$) -- ($(x3)+(0.11,0.11)$);
  \draw[thick] ($(x3)+(-0.11,0.11)$) -- ($(x3)+(0.11,-0.11)$);
  \node[above right=0pt] at (x3) {$x_{t_3}$};
  \node[right=2pt] at (3.9,0) {$f$};
  \node[above right=0pt] at (45:1.5) {\small $\beta_1/2$};
  \node[above right=-1pt] at (100:0.625) {\small $\beta_2/2$};
  \node[gray, below=2pt] at (270:0.2) {\small $f/\Lambda$};
  \node[align=center, gray!60!black] at (90:2.65)
    {\small blocked by the\\[-2pt]\small 1st charge};
  \node[align=center, gray!60!black] at (250:1.05)
    {\small by the 2nd};
\end{tikzpicture}
\caption{Charging a replay center for dear steps (metric case $\alpha = 1$, threshold $\theta=3$). Each dear charge poisons its own scale.  Dear openings charged
to a replay center $w$ plant adaptive centers (crosses) at radii
$\beta_1 > \beta_2 > \cdots$ with $\beta_{j+1} < \beta_j/2$: once
$x_{t_j}$ is planted, every point of the shaded annulus is within
distance $3\,d(\cdot, w)$ of an adaptive center, so it can never again
produce a dear charge to $w$. All radii lie in $[f/\Lambda,\, f]$
(above $f$ every point opens with probability one), so at most
$\lgp\Lambda + O(1)$ charges fit.\label{fig:halving}}
\end{figure}

\begin{lem}[Charging lemma]\label{lem:charge}
Fix $\theta > 2\alpha$ and call step $t$ \emph{dear} if $a_t > \theta b_t$. Then for every realization and every $t \ge 0$,
\[
\#\bigl\{\,s \le t \;:\; o_s = 1 \text{ and } s \text{ is dear}\,\bigr\} \;\le\; \bigl(\log_{(\theta-\alpha)/\alpha}\Lambda + 2\bigr)\, K^{\OB}_t .
\]
\end{lem}
 
\begin{proof}
At a dear step $s$ we have $b_s < a_s/\theta \le 1/\theta < 1$, so $d(x_s, S^{\OB}_{s-1}) = b_s f < f$; in particular $S^{\OB}_{s-1} \ne \emptyset$ and there is a nearest replay center $w \in S^{\OB}_{s-1}$ with $d(x_s, w) = b_s f$ (break ties by age).
\emph{If the adaptive run opens a center at a dear step, charge that opening to this $w$} (\Cref{fig:halving}).
 
Fix any realization and any replay center $w$; let $t_1 < t_2 < \cdots < t_J$ be the dear steps whose adaptive openings are charged to $w$, with radii
\[
\beta_j \;=\; \frac{d(x_{t_j}, w)}{f} \;=\; b_{t_j} \;<\; \frac1\theta .
\]
Since the charge records an adaptive opening, $o_{t_j} = 1$ and the point $x_{t_j}$ is an \emph{adaptive center} from time $t_j$ onward: 
$x_{t_j} \in S^{\A}_{t}$ for all $t \ge t_j$. Hence at step $t_{j+1}$, by the $\alpha$-approximate triangle inequality through $w$, it follows that
\[
a_{t_{j+1}}
\;\le\; \frac{d\bigl(x_{t_{j+1}},\, x_{t_j}\bigr)}{f}
\;\le\; \alpha\,\frac{d(x_{t_{j+1}}, w) + d(w, x_{t_j})}{f}
\;=\; \alpha\bigl(\beta_{j+1} + \beta_j\bigr) ,
\]
while dearness at $t_{j+1}$ gives $a_{t_{j+1}} > \theta\,\beta_{j+1}$.
Combining these inequalities, we obtain
\begin{equation}\label{eq:halving}
\beta_{j+1} \;<\; \frac{\alpha\,\beta_j}{\theta - \alpha}
\end{equation}
In particular, this means that consecutive radii shrink by at least the factor $(\theta - \alpha)/\alpha > 1$ for each charge. A charge of radius $0$, i.e. the adaptive run opening at the location of $w$ itself, is necessarily the last charge to $w$: afterwards, the adaptive run holds a center at that location, so for any point $x$ whose nearest replay center is $w$, $a \le d(x, w)/f = b \le \theta b$, and no later dear step will be charged to $w$. Every nonzero radius satisfies
$\beta_j \ge d_{\min}/f$, since $x_{t_j}$ and $w$ are distinct presented points, and $\beta_1 \le \min\{1/\theta,\, d_{\max}/f\}$.
By~\eqref{eq:halving} the number of nonzero charges is at most
$\log_{(\theta-\alpha)/\alpha}\Lambda + 1$, so
$J \le \log_{(\theta-\alpha)/\alpha}\Lambda + 2$. Every charged center belongs to $S^{\OB}_{s-1} \subseteq S^{\OB}_t$ for some $s \le t$; summing over $w \in S^{\OB}_t$ proves the claim. 
\end{proof}

\subsection{Inflation and Deflation Upper Bounds}

\begin{thm}[Robustness to inflation]\label{thm:upper}
For every semi-metric space, every facility price $f > 0$, every adaptive generator, one-sided or two-sided (\Cref{def:twosided}), and
every $(\Gg_t)$-stopping time $\tau$,\footnote{Quantifying over
$(\Gg_t)$-stopping times \emph{strengthens} the theorem: the bound
holds under every stopping rule, including hypothetical ones depending
on both runs. In particular, the class contains the model's default
$(\Ff_t)$-stopping times. }
\[
\E\bigl[K^{\A}_\tau\bigr]
\;\le\; \kappa_\Lambda \, \E\bigl[K^{\OB}_\tau\bigr],
\]
and consequently
$\E[C^{\A}_\tau] \le 2\kappa_\Lambda\, \E[C^{\OB}_\tau]$, where
$\kappa_\Lambda$ is as defined in \Cref{fact:theta}.
\end{thm}

\begin{proof}
Fix a deterministic dearness threshold $\theta>2\alpha$. Call step $t$
\emph{cheap} if $a_t \le \theta b_t$ and \emph{dear} otherwise; the
cheapness indicator $h_t$ is $\Gg_{t-1}$-measurable, hence predictable.
Decompose $K^{\A}_\tau$ into adaptive openings at cheap steps and at
dear steps. Openings at dear steps are bounded pathwise by
\Cref{lem:charge}:
\[
\#\{\text{adaptive openings at dear steps up to } \tau\}
\;\le\; \bigl(\log_{(\theta-\alpha)/\alpha}\Lambda + 2\bigr)\, K^{\OB}_\tau .
\]
For cheap steps, by \Cref{lem:equiv} applied with $h_t$ and then with
$h \equiv 1$, and the pointwise inequality
$a_t h_t \le \theta b_t h_t \le \theta b_t$,
\[
\E\Bigl[\sum_{t\le\tau} o_t h_t\Bigr]
\;=\; \E\Bigl[\sum_{t\le\tau} a_t h_t\Bigr]
\;\le\; \theta\,\E\Bigl[\sum_{t\le\tau} b_t\Bigr]
\;=\; \theta\,\E\bigl[K^{\OB}_\tau\bigr].
\]
Taking expectations of the dear bound and adding gives
$\E[K^{\A}_\tau] \le (\theta + \log_{(\theta-\alpha)/\alpha}\Lambda + 2)\,
\E[K^{\OB}_\tau]$ for every deterministic $\theta > 2\alpha$; taking the
infimum over $\theta$ proves the center-count claim. The cost claim
follows from \Cref{cor:equiv} via
$\E[C^{\A}_\tau] \le 2f\,\E[K^{\A}_\tau]$ and
$f\,\E[K^{\OB}_\tau] \le \E[C^{\OB}_\tau]$. For $\tau = \infty$, apply
the argument at $\tau \wedge T$ and let $T \to \infty$; all quantities
are monotone.

The proof applies as-is to two-sided generators (\cref{def:twosided}): The replay's randomness was used only through the fair-count identities~\eqref{eq:faircount} and the predictability of the cheap/dear indicator, which are both valid for two-sided generators, since $a_t$ and $b_t$ are $\Gg_{t-1}$-measurable (\Cref{lem:equiv}); the dear-step bound (\Cref{lem:charge}) holds deterministically.
\end{proof}

\begin{cor}[Robustness to deflation]\label{thm:upper_reversed}
For every semi-metric space, every facility price $f > 0$, every adaptive
generator, one-sided or two-sided (\Cref{def:twosided}), and
every $(\Gg_t)$-stopping time $\tau$,
\[
\E\bigl[K^{\OB}_\tau\bigr]
\;\le\; \kappa_\Lambda\, \E\bigl[K^{\A}_\tau\bigr]
\qquad\text{and}\qquad
\E\bigl[C^{\OB}_\tau\bigr] \;\le\; 2\kappa_\Lambda\, \E\bigl[C^{\A}_\tau\bigr]:
\]
adaptivity cannot \emph{deflate} the expected center count or cost by more than the factor of \Cref{thm:upper}.
\end{cor}
 
\begin{proof}
Exchange the coin sequences $U$ and $U'$ (\Cref{rem:exchange}). The exchanged generator is two-sided, $\tau$ remains a $(\Gg_t)$-stopping time, and the stream has the same distribution, so $\Lambda$ is unchanged; Thus, the claim follows by applying \Cref{thm:upper} to the exchanged generator.
\end{proof}

One application of \Cref{thm:upper} is converting per-sequence guarantees into adaptive ones:

\begin{cor}[Transfer of oblivious guarantees]\label{cor:transfer}
Let $\Phi$ map finite sequences over a semi-metric space $\Mm$ to $[0,\infty]$ and suppose the
sketch satisfies $\E[C(x)] \le \Phi(x)$ for every fixed finite sequence $x$,
where the expectation is taken over the sketch's coins. Then for every adaptive
generator and every a.s.\ finite stopping time $\tau$,
\[
\E\bigl[C^{\A}_\tau\bigr]
\;\le\; 2 \kappa_\Lambda\cdot \E\bigl[\Phi(X_{\le\tau})\bigr],
\]
where $X_{\le\tau} = (x_1, \dots, x_\tau)$. Concretely, the per-sequence
competitive guarantees $\Phi(x) = \beta\cdot\mathrm{OPT}(x)$ extend to adaptive generators at an $O(\log\Lambda/\log\log\Lambda)$ loss.
By instantiating $\Phi$ with the known fixed-sequence competitive ratio of $\beta= O(\log n/\log\log n)$ \citep{fotakis2008competitive} (see \cref{thm:semiratio} for the extension to semi-metrics), we conclude that the sketch is
$O\bigl(\tfrac{\log\Lambda}{\log\log\Lambda}\cdot\tfrac{\log n}{\log\log n}\bigr)$-competitive against OPT when the request sequence is produced by an adaptive generator. 
\end{cor}

\subsection{High-probability Robustness}\label{sec:whp}

The proofs above use the fair-count
identity to equate, in expectation at stopping times, the number of
openings with the accumulated sum of opening probabilities. We now establish
that these two quantities stay close with high
probability, and uniformly over all times at once. This allows for high probability and all times bounds on inflation or deflation.
The argument is a Chernoff bound made time-uniform: instead of
applying Markov's inequality to $e^{\gamma N_t}$ at a fixed time $t$, we
apply Ville's inequality, a maximal version of Markov's inequality
for nonnegative supermartingales, to an exponential process that
tracks all times simultaneously. 
The exponential supermartingale goes back to \citet{Freedman1975}, and the time-uniform technique is developed systematically in \citet{HowardRMS20}; we include the short proof to keep the section self-contained.

\begin{fact}[Ville's inequality~\cite{Ville39}]\label{fact:ville}
Let $(Z_t)_{t \ge 0}$ be a nonnegative supermartingale with respect to a
filtration $(\Gg_t)$, i.e., $Z_t \ge 0$ and
$\E[Z_t \mid \Gg_{t-1}] \le Z_{t-1}$ for all $t \ge 1$. Then for every
$a > 0$,
\[
\Pr\Bigl[\,\sup_{t \ge 0} Z_t \;\ge\; a\,\Bigr]
\;\le\; \frac{\E[Z_0]}{a}.
\]
\end{fact}

\begin{lem}[Time-uniform fair count ]\label{lem:ville}
Let $(X_t)_{t \ge 1}$ be adapted to a filtration $(\Gg_t)$ with
$X_t \in [0,1]$, and let $\mu_t = \E[X_t \mid \Gg_{t-1}]$. Write
$N_t = \sum_{s \le t} X_s$ and $A_t = \sum_{s \le t} \mu_s$. Then for
every $\gamma \in (0,1]$ and $\delta \in (0,1)$, each of the following
holds with probability at least $1 - \delta$, simultaneously for all
$t \ge 1$:
\begin{enumerate}
\item[(i)] $\displaystyle
N_t \;\le\; (1+\gamma)\,A_t + \tfrac{2}{\gamma}\ln\tfrac1\delta$;
\item[(ii)] $\displaystyle
A_t \;\le\; (1+\gamma)\,N_t + \tfrac{2}{\gamma}\ln\tfrac1\delta$.
\end{enumerate}
\end{lem}

\begin{proof}
Both parts follow from the exponential-supermartingale argument of \citet{HowardRMS20}; we spell it out to fix the constants. 

For (i), let
$Z_t = \exp\bigl(\gamma N_t - (e^{\gamma}-1)A_t\bigr)$, $Z_0 = 1$.
Since $e^{\gamma x} \le 1 + x(e^{\gamma}-1)$ for $x \in [0,1]$
(by convexity),
\[
\E[e^{\gamma X_t} \mid \Gg_{t-1}]
\;\le\; 1 + \mu_t(e^{\gamma}-1)
\;\le\; e^{\mu_t(e^{\gamma}-1)},
\]
so $(Z_t)$ is a nonnegative supermartingale. By \Cref{fact:ville} with
$a = 1/\delta$, with probability at least $1-\delta$ we have
$Z_t < 1/\delta$ for all $t$; taking logarithms,
\[
N_t \;\le\; \frac{e^{\gamma}-1}{\gamma}\,A_t
+ \frac{1}{\gamma}\ln\frac1\delta
\qquad\text{for all } t.
\]
Since $\frac{e^{\gamma}-1}{\gamma} = 1 + \frac{\gamma}{2} +
\frac{\gamma^2}{6} + \cdots \le 1 + (e-2)\gamma \le 1 + \gamma$ for
$\gamma \le 1$, bound (i) follows (weakening $\tfrac1\gamma$ to
$\tfrac2\gamma$ to match (ii)).

For (ii), let
$Z'_t = \exp\bigl((1-e^{-\gamma})A_t - \gamma N_t\bigr)$. Note that
$e^{-\gamma x} \le 1 - x(1-e^{-\gamma})$ for $x \in [0,1]$ (by convexity), so
$\E[e^{-\gamma X_t} \mid \Gg_{t-1}]
\le 1 - \mu_t(1-e^{-\gamma})
\le e^{-\mu_t(1-e^{-\gamma})}$,
so $(Z'_t)$ is a nonnegative supermartingale. Also, by \Cref{fact:ville}, we have that with probability at least $1-\delta$,
\[
A_t \;\le\; \frac{\gamma}{1-e^{-\gamma}}\,N_t
+ \frac{1}{1-e^{-\gamma}}\ln\frac1\delta
\qquad\text{for all } t.
\]
Since $1 - e^{-\gamma} \ge \gamma(1-\tfrac{\gamma}{2}) \ge
\tfrac{\gamma}{2}$ for $\gamma \le 1$, we have
$\frac{\gamma}{1-e^{-\gamma}} \le \frac{1}{1-\gamma/2} \le 1+\gamma$
and $\frac{1}{1-e^{-\gamma}} \le \frac{2}{\gamma}$. This implies the claim in (ii).
\end{proof}

Substituting \Cref{lem:ville} for the fair-count identity in the proof of \Cref{thm:upper} gives the inflation bound with high probability, uniformly in time; the deflation bound follows by \Cref{rem:exchange}, which applies to probability bounds as it does to expectations, since the exchange leaves the distribution of the whole trajectory unchanged.
 
\begin{thm}[High-probability robustness, uniform in time]\label{thm:upper_whp}
For every semi-metric space, every facility price $f > 0$, every adaptive generator, one-sided or two-sided, and every $\gamma \in (0,1]$,
$\delta \in (0,1)$: with probability at least $1 - 2\delta$,
simultaneously for all $t \ge 0$,
\[
K^{\A}_t \;\le\; (1+\gamma)^2\, \kappa_\Lambda\, K^{\OB}_t + \eta_\delta,
\qquad
\eta_\delta = \frac{5\,\kappa_\Lambda}{\gamma}\ln\frac1\delta .
\]
The same holds with the runs exchanged, and both bounds hold simultaneously with probability at least $1 - 4\delta$.
\end{thm}
\begin{proof}
Let $\theta = \theta_\Lambda$ and let $h_t = \ind{\{a_t \le \theta b_t\}}$
be the cheapness indicator, which is $(\Gg_t)$-predictable. By the proof of
\Cref{lem:equiv}, each of the pairs
$(X_t, \mu_t) = (o_t \psi_t,\, a_t \psi_t)$ and
$(X_t, \mu_t) = (o'_t \psi_t,\, b_t \psi_t)$ satisfies the hypotheses of
\Cref{lem:ville}, for any $(\Gg_t)$-predictable $\psi_t \in [0,1]$.
Instantiate \Cref{lem:ville} twice, each with failure probability
$\delta$:
\begin{itemize}
\item $E_1$: part (i) for $(o_t h_t,\, a_t h_t)$, so that for all $t$,
$\sum_{s\le t} o_s h_s \le (1+\gamma)\sum_{s\le t} a_s h_s
+ \tfrac2\gamma\ln\tfrac1\delta$;
\item $E_2$: part (ii) for $(o'_t,\, b_t)$ (that is, $\psi \equiv 1$), so
that for all $t$, $\sum_{s\le t} b_s \le (1+\gamma)K^{\OB}_t
+ \tfrac2\gamma\ln\tfrac1\delta$.
\end{itemize}
On $E_1 \cap E_2$, for every $t$, splitting $K^{\A}_t$ at the cheap/dear
boundary, bounding the dear part pathwise by \Cref{lem:charge}, and chaining through the pointwise inequality
$a_s h_s \le \theta b_s$:
\begin{align*}
K^{\A}_t
&\;\le\; (1+\gamma)\sum_{s\le t} a_s h_s
+ \tfrac2\gamma\ln\tfrac1\delta
+ \bigl(\log_{(\theta-\alpha)/\alpha}\Lambda + 2\bigr)K^{\OB}_t \\
&\;\le\; (1+\gamma)\,\theta \sum_{s\le t} b_s
+ \tfrac2\gamma\ln\tfrac1\delta
+ \bigl(\log_{(\theta-\alpha)/\alpha}\Lambda + 2\bigr)K^{\OB}_t \\
&\;\le\; \Bigl((1+\gamma)^2\theta + \log_{(\theta-\alpha)/\alpha}\Lambda + 2\Bigr)
K^{\OB}_t
+ \bigl((1+\gamma)\theta + 1\bigr)\tfrac2\gamma\ln\tfrac1\delta .
\end{align*}
The additive term is at most
$(2\theta + 1)\tfrac2\gamma\ln\tfrac1\delta
\le \tfrac{4\theta+2}{\gamma}\ln\tfrac1\delta
\le \tfrac{5\,\kappa_\Lambda}{\gamma}\ln\tfrac1\delta = \eta_\delta$,
using $\gamma \le 1$, $\theta_\Lambda \le \kappa_\Lambda$, and
$\kappa_\Lambda > 4$. Since
$(1+\gamma)^2\theta + \log_{(\theta-\alpha)/\alpha}\Lambda + 2
\le (1+\gamma)^2\bigl(\theta + \log_{(\theta-\alpha)/\alpha}\Lambda + 2\bigr)
= (1+\gamma)^2 \kappa_\Lambda$, the bound holds on
$E_1 \cap E_2$, an event of probability at least $1 - 2\delta$. The bound with the runs exchanged is the same event for the exchanged generator, which has the same probability (\Cref{rem:exchange}); a union bound gives both simultaneously.
\end{proof}
 
The cost version needs two further events, which replace \eqref{eq:costcenter}.

\begin{cor}[High-probability cost robustness]\label{cor:upper_whp_cost}
For every $\eps \in (0,1]$ and $\delta \in (0,\tfrac12]$, in the setting
of \Cref{thm:upper_whp}: with probability at least $1 - \delta$,
simultaneously for all $t \ge 0$,
\[
C^{\A}_t \;\le\; 2(1+\eps)\, \kappa_\Lambda\, C^{\OB}_t
+ O\Bigl(\tfrac{\kappa_\Lambda}{\eps}\, f \ln\tfrac1\delta\Bigr),
\]
and the same holds with the runs exchanged (\Cref{rem:exchange}). The implied constant is
absolute.
\end{cor}
\begin{proof}
We prove the statement with $(1+\gamma)^4$ in place of $(1+\eps)$ and
failure probability $4\delta$; the stated form follows by taking
$\gamma = \eps/8$ (so that $(1+\gamma)^4 \le e^{4\gamma} =
e^{\eps/2} \le 1+\eps$ for $\eps \le 1$) and replacing $\delta$ by
$\delta/4$, using $\ln(4/\delta) = O(\ln(1/\delta))$ for
$\delta \le \tfrac12$.

The realized step cost of either run lies in $[0, f]$, and by the proof
of \Cref{lem:equiv} its conditional mean $m_t$ satisfies
$m_t \in [f q_t,\, 2 f q_t]$, where $q_t$ is that run's opening
probability; in particular the normalized step costs are adapted,
$[0,1]$-valued, and satisfy the hypotheses of \Cref{lem:ville}. Apply
\Cref{lem:ville}(i) to $X_t = (\text{step cost})/f$, with
$\mu_t = m^{\A}_t/f$, then bound $m^{\A}_s/f \le 2 a_s$ and pass to
$K^{\A}_t$ via a further instance of \Cref{lem:ville}(ii), for $(o_t, a_t)$:
\[
C^{\A}_t/f
\;\le\; (1+\gamma)\sum_{s \le t} m^{\A}_s/f + \tfrac2\gamma\ln\tfrac1\delta
\;\le\; 2(1+\gamma)\sum_{s\le t} a_s + \tfrac2\gamma\ln\tfrac1\delta
\;\le\; 2(1+\gamma)^2 K^{\A}_t
+ \bigl(4(1+\gamma)+2\bigr)\tfrac1\gamma\ln\tfrac1\delta .
\]
Chain with \Cref{thm:upper_whp} and conclude with the pathwise
inequality $f K^{\OB}_t \le C^{\OB}_t$ (as $C_t = f K_t + Q_t$).
Union bound: two events beyond the two of \Cref{thm:upper_whp}, four in total.
\end{proof}

\begin{remark}[The additive term is necessary]\label{rem:additive}
Some additive dependence on $\delta$ is unavoidable even for
oblivious inputs. To see this, consider the example of a ``star'' configuration: an anchor followed by points
$z_1, \dots, z_m$, each at distance $\eps f$ from the anchor and
pairwise at distance $2\eps f$, with $m\eps = 1$. Openings among the
$z_i$ are independent $\mathrm{Bern}(\eps)$ in either run (an opened
$z_j$ lies at distance $2\eps f > \eps f$ from every other $z_i$, so the
anchor remains nearest), whence each count is $1$ plus an approximately
Poisson$(1)$ variable, and the two runs are independent given the
stream. With probability $\Omega(\delta)$ the adaptive count exceeds
$\Omega(\log(1/\delta)/\log\log(1/\delta))$ while the replay count is
$1$. The gap between this lower bound and the additive term
$O\bigl((\kappa_\Lambda/\gamma)\log(1/\delta)\bigr)$ of
\Cref{thm:upper_whp} is open.
\end{remark}

\section{The Inflation Attack} \label{sec:anticertproof}
We next present a deterministic adaptive generator, on the real line, which forces the adaptive execution of the Meyerson sketch to incur total cost $\Omega(\log \Delta/\log \log \Delta)$ times the cost of the replay execution. 

\begin{theorem}[Inflation attack]
\label{thm:anticert}
There are absolute constants $C_0 \le 25$ and $R_0$ ($R_0 = 30$ suffices) such that for every $R \ge R_0$ there exist a finite metric space $\Mm \subset \mathbb{R}$ (with the line metric) of aspect ratio $\Delta = \exp(\Theta(R\log R))$, a deterministic one-sided adaptive
generator over $\Mm$, and an a.s.\ finite $(\Ff_t)$-stopping time $\tau$ with
\[
\E\bigl[K^{\A}_\tau\bigr] \;\ge\; R - 1,
\qquad
\E\bigl[K^{\OB}_\tau\bigr] \;\le\; C_0,
\]
and, separately in the routing components,
\[
\E\bigl[Q^{\A}_\tau\bigr] \;\ge\; (R-2)\,f,
\qquad
\E\bigl[Q^{\OB}_\tau\bigr] \;\le\; C_0\, f .
\]
Consequently the adaptivity inflation gaps of the total cost, the
center count, and the routing component are all
$\Om(\log\Delta/\log\log\Delta)$.
\end{theorem}

The remainder of this section establishes  \Cref{thm:anticert} and two extensions: a high probability bound in \Cref{cor:whpattack} and an implementation of the generator with canonical stopping rules in \Cref{rem:budget}.

\subsection{The Inflation Generator} \label{infgen:sec}

\medskip\noindent\textbf{Description.}
The configuration includes $N_2$ identically configured regions spaced far apart.
The attack processes the regions in order until it certifies, in some region, a point $z^*$ that is not a center point in the adaptive execution but is a center in the replay, with probability $1 - O(1/R)$; if no region is certified, we say the attack fails.
The partial configuration and processing of each region until declaring success or failure is illustrated in \Cref{fig:preparation}. 
If such a $z^*$ is certified, the attack then proceeds to an ``extraction'' process, in which the adaptive execution opens $R$ centers and the replay only opens $O(1)$ in expectation. The extraction configuration and process are illustrated in \Cref{fig:extraction}.

\begin{figure}[t]
\centering
\begin{tikzpicture}[line cap=round]
  \draw (-0.5,0) -- (2.1,0);
  \draw (2.35,-0.14) -- (2.65,0.14);
  \draw (2.55,-0.14) -- (2.85,0.14);
  \draw (3.1,0) -- (13.0,0);
  \fill (0,0) circle (3pt);
  \node[below=3pt] at (0,0) {$h$};
  \fill[gray!55] (3.8,0) circle (2pt);
  \node[below=2pt, gray!75!black] at (3.8,0) {$u_1$};
  \fill[gray!55] (4.9,0) circle (2pt);
  \fill[gray!55] (6.0,0) circle (2pt);
  \fill[gray!55] (7.1,0) circle (2pt);
  \draw[red!70!black, thin] (6.0,0) circle (5pt);
  \fill[blue!65!black] (8.2,0) circle (3pt);
  \node[below=3pt, blue!65!black] at (8.0,0) {$u^*$};
  \fill[red!70!black] (8.62,0) circle (3pt);
  \node[below=3pt, red!70!black] at (8.85,0) {$z^*$};
  \draw[red!70!black, thin] (8.74,0.30) circle (2pt);
  \draw[red!70!black, thin] (8.86,0.52) circle (2pt);
  \node[gray!75!black, right=2pt] at (8.98,0.46)
    {\small presented $n_2$ times};
  \draw[gray!45] (9.8,0) circle (2pt);
  \draw[gray!45] (10.9,0) circle (2pt);
  \draw[gray!45] (12.0,0) circle (2pt);
  \node[gray!60, below=2pt] at (10.9,-0.10) {\small never presented};
  \draw[<->] (0,-0.85) -- (3.8,-0.85)
    node[midway, below=2pt] {$\approx f/R$};
  \draw[<->] (3.8,-0.85) -- (4.9,-0.85)
    node[midway, below=2pt] {$\delta_\star$};
  \draw[decorate,decoration={brace,amplitude=3pt,raise=7pt}]
    (8.2,0) -- (8.62,0) node[midway, above=11pt, xshift=-3pt] {$\rho_0$};
\end{tikzpicture}
\caption{Inflation generator: The certification phase in a single region. 
A region is configured with a ``hub'' point $h$, and within distance $\approx f/R$ from it, a set of $n_1$ evenly-spaced attempt points $(u_i)$, distance of $\delta_\star = f/(Rn_1)$ apart, in an interval of size $f/R$. For each such point $u_i$, there is a respective certification point $z_i$ within distance $\eta = \delta_\star/\lceil\ln R\rceil$ (and, not shown here, a respective set of $R$ extraction points that get geometrically closer to the certification point, with distances $(\rho_j)$).  \\
When processing the region, the generator first requests the hub $h$ (which is opened in both the adaptive and replay executions). It then
requests the attempt points $(u_i)$ in order until a center is opened at one of them, i.e. $u_i\equiv u^*$. If all attempts fail, then certification failed. Otherwise, the attack places $n_2 = (f/\delta_\star)\lceil\ln R\rceil$ requests on the respective  certification point $z_i\equiv z^*$. The region is certified if \emph{none} of these $n_2$ requests opened a center $z^*$. Due to the spacings, the replay is unlikely (probability at most $2/R$) to hold $u^*$; not holding it, it holds $z^*$ after the $n_2$ requests with probability at least $1-1/R$. The adaptive run, which holds $u^*$, evades $z^*$ with constant probability. Overall, the region is certified (the adaptive run holds a $u^*$ and does not hold $z^*$) with constant probability, and on certification the probability of a \emph{silent failure} (the replay does not hold $z^*$) is at most $3/R$.
}
\label{fig:preparation}
\end{figure}
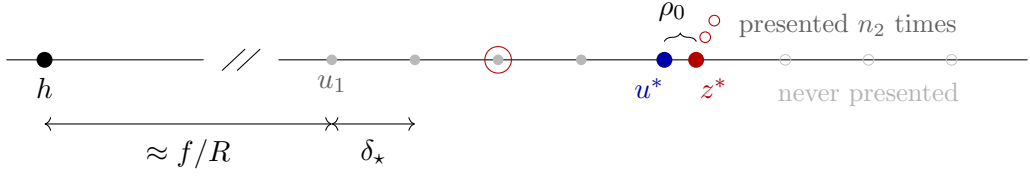

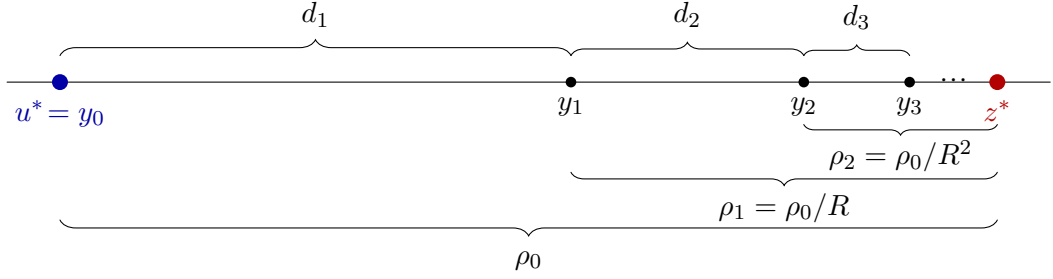
\begin{figure}[t]
\centering
\begin{tikzpicture}[line cap=round]
  \draw (-0.7,0) -- (13.1,0);
  \coordinate (y0) at (0,0);
  \coordinate (y1) at (6.76,0);
  \coordinate (y2) at (9.84,0);
  \coordinate (y3) at (11.24,0);
  \coordinate (cs) at (12.4,0);
  \fill (y1) circle (2pt);
  \node[below=2pt] at (y1) {$y_1$};
  \fill (y2) circle (2pt);
  \node[below=2pt] at (y2) {$y_2$};
  \fill (y3) circle (2pt);
  \node[below=2pt] at (y3) {$y_3$};
  \node at (11.82,0.02) {$\cdot\!\cdot\!\cdot$};
  \fill[blue!65!black] (y0) circle (3pt);
  \node[below=3pt, blue!65!black] at (y0) {$u^*\!=y_0$};
  \fill[red!70!black] (cs) circle (3pt);
  \node[below=3pt, red!70!black] at (cs) {$z^*$};
  \draw[decorate,decoration={brace,amplitude=6pt,raise=10pt}]
    (y0) -- (y1) node[midway, above=17pt] {$d_1$};
  \draw[decorate,decoration={brace,amplitude=5pt,raise=10pt}]
    (y1) -- (y2) node[midway, above=17pt] {$d_2$};
  \draw[decorate,decoration={brace,amplitude=4pt,raise=10pt}]
    (y2) -- (y3) node[midway, above=16pt] {$d_3$};
  \draw[decorate,decoration={brace,mirror,amplitude=4pt,raise=16pt}]
    (y2) -- (cs) node[midway, below=18pt] {$\rho_2=\rho_0/R^2$};
  \draw[decorate,decoration={brace,mirror,amplitude=5pt,raise=34pt}]
    (y1) -- (cs) node[midway, below=38pt] {$\rho_1=\rho_0/R$};
  \draw[decorate,decoration={brace,mirror,amplitude=6pt,raise=52pt}]
    (y0) -- (cs) node[midway, below=60pt] {$\rho_0$};
\end{tikzpicture}
\caption{Inflation attack: Exploiting a certified region.  The starting state has
an adaptive-only center $u^*$ (blue) and a replay-only center
$z^* = u^* + \rho_0$ (red); everything else, in both runs, is far away. The extraction configuration has locations $y_i = z^* - \rho_i$ with $\rho_i = \rho_0/R^i$ getting geometrically closer to $z^*$. In round $i$ the generator
presents copies of $y_i$ until the adaptive run opens one, so each $y_i$ becomes an adaptive center in turn. The adaptive run opens each copy with probability $d_i/f$, where $d_i = \rho_{i-1} - \rho_i$ is the hop from the previous location; the replay, which holds $z^*$, opens with probability at most $\rho_i/f$, a factor $R-1$ smaller.}
\label{fig:extraction}
\end{figure}

Throughout this section, we assume $R \ge R_0$, where $R_0$ an
absolute constant chosen so that the finitely many explicit numerical
inequalities below hold ($R_0 = 30$ suffices); no constant has been
optimized.

\medskip\noindent\textbf{Layout.}
The layout is configured using the following parameters.
\[
s = \frac{f}{R}, \qquad
n_1 = R\lceil \ln R\rceil, \qquad
\delta_\star = \frac{s}{n_1} = \frac{f}{R^2\lceil \ln R \rceil}, \qquad
\eta = \frac{\delta_\star}{\lceil \ln R\rceil }= \frac{f}{R^2\lceil \ln R \rceil^2} ,
\]
\[
n_2 = R^2 \lceil \ln R \rceil^2 , \qquad
\rho_j = \eta R^{-j} \ (0 \le j \le R), \qquad
N_2 = \lceil 5 \ln R\rceil .
\]
Note the exact identities $n_2 = f/\eta$ and
$n_2\,\delta_\star/f = \lceil\ln R\rceil$, which the certification
analysis uses.

The configuration consists of $N_2$ identically laid-out \emph{regions}, indexed by $c \in [N_2]$.%
The layout of region $c$ includes a hub $h_c = 8fc$, $n_1$ evenly spread
\emph{attempt positions} $u_{c,i} = h_c + s + (i-1)\,\delta_\star$ for $i \in [n_1]$, and 
for each $i$, a \emph{certification point} $z_{c,i} = u_{c,i} + \eta$. Finally, we define 
\emph{extraction points} $y_{c,i,j} = z_{c,i} - \rho_j$ for $j \in [R]$, which get geometrically closer to the certification point. Note that $y_{c,i,0} = u_{c,i}$, so the extraction points climb from the attempt point toward its certification point. 

Specifically, attempt points are at distance between
$s$ and $2s$ from the hub, with spacing  $\delta_\star$ apart; $\eta < \delta_\star$ is the
distance of a certification point from its attempt point; and
$\rho_j$ for $j\ge 1$ are geometrically shrinking distances of extraction points to the respective certification point, with
$\rho_0 = \eta$.

\medskip\noindent\textbf{Layout Properties.}
All points of region $c$ lie in $[h_c,\, h_c + 2s + \eta] \subseteq
[h_c,\, h_c + 3s]$, so all within-region distances are
at most $3s = 3f/R < f$, while distinct regions are at distance
$\ge 8f - 3s \ge 6f$. Therefore: 
\begin{itemize}
    \item [\emph{(P0)}]
When the first point of a region (its hub) is presented, every previously
presented point is at distance $\ge 6f \ge f$, so \emph{both} runs open
every hub with probability $1$.
    \item [\emph{(P1)}] Once a region's hub is placed, a within-region center at distance $< f$ always exists, so centers
in other regions (at distance $\ge 6f$) never determine an opening
probability. The relevant exact within-region distances:
$d(u_{c,i}, h_c) = s + (i-1)\,\delta_\star \in [s,\, 2s)$;
$d(u_{c,i}, u_{c,i'}) = |i - i'|\,\delta_\star$;
$d(z_{c,i}, u_{c,i}) = \eta$, with $z_{c,i}$ lying strictly between
$u_{c,i}$ and $u_{c,i+1}$; and
$d(y_{c,i,j}, z_{c,i}) = \rho_j$, $d(y_{c,i,j}, u_{c,i}) = \eta - \rho_j = \rho_0 - \rho_j$, with the extraction points lying strictly between $u_{c,i}$ and $z_{c,i}$.
\end{itemize}

\medskip\noindent\textbf{The generator.}
For $c = 1, 2, \dots, N_2$, or until some region is certified:
\begin{itemize}[leftmargin=2em]
\item \emph{Phase 1 (spread attempts).} Present $h_c$; then present
$u_{c,1}, u_{c,2}, \dots$ in order, stopping at the first index $T = T_c$
at which the \emph{adaptive} run opens; we then say the region
\emph{selects} $u^* := u_{c,T}$. If no opening occurs within $n_1$
attempts, the region fails; continue with region $c+1$.
\item \emph{Phase 2 (certification).} Write $z^* = z_{c,T}$. Present up
to $n_2$ identical copies of $z^*$. If the
adaptive run opens a copy, the region fails; continue with region $c+1$.
If it
rejects all $n_2$ copies, the region is \emph{certified}. 
\end{itemize}
\emph{Extraction.} At the first certified region (write $y_j = y_{c,T,j}$):
for $j = 1, \dots, R$, present identical copies of $y_j$ until the adaptive
run opens one.
Let $\tau$ be the end of round $R$, or the halting time if
all $N_2$ regions failed.

Every decision above depends only on the adaptive
transcript, so this is a deterministic one-sided generator and $\tau$ is an
$(\Ff_t)$-stopping time; $\tau$ is a.s.\ finite because, by
\Cref{clm:biases} below, every waiting time (determined by the number of identical copies of $z^*$ presented until
an opening) is distributed as a geometric random variable with a deterministic positive parameter.

\subsection{Analysis and Openings Bounds}

We first give the high-level argument; the formal proof is presented afterwards.
\begin{itemize}
\item {\bf Certifying a region.}
    A region is certified if (i)~the adaptive run opens some $u_i\equiv u^*$, and (ii)~the adaptive run then rejects all $n_2$ requests at the respective $z^*$. For (i): each presented attempt is opened with probability at least $1/R$, so selection fails only with probability $(1-1/R)^{n_1}\le 1/R$. For (ii): conditioned on selection, each copy of $z^*$ is opened with probability $\eta/f = 1/n_2$, so all $n_2$ are rejected with probability $(1-1/n_2)^{n_2}\approx 1/e$. Overall, a region is certified with probability at least $1/5$.
\item {\bf A silent failure} occurs when a region is certified but the replay does not hold $z^*$. This is bounded by the probability of the event that either $u^*$ is also a replay center (probability at most $2/R$), or that, the replay does not hold $u^*$ but it still rejected all $n_2$ requests at $z^*$. Since each copy is opened with probability at least $\delta_\star/f = 1/(R^2\lceil\ln R\rceil)$, this has probability $(1-\delta_\star/f)^{n_2}\le 1/R$. The silent failure probability is therefore at most $3/R$.
\item {\bf Extraction.}
The adaptive run opens exactly one center per location, $R$ in total. Absent a silent failure, the replay holds $z^*$ and opens only $\approx 1/R$ centers per location in expectation, $O(1)$ in total; on a silent failure it may open up to one center per location, at most $R$ in total --- but a silent failure only occurs with probability
$O(1/R)$.
\item {\bf Bounding expected openings.}
Each region is certified with constant probability, so only $O(1)$ regions are processed in expectation, and a failed region costs both runs $O(1)$ openings in expectation. By combining the above, it follows that the adaptive execution opens $\Omega(R)$ centers in expectation while the replay opens $O(1)$.
\item {\bf Resources.}
Each extraction location shrinks the scale by a factor of $R$, so the configuration spans $R$ geometric scales below $\eta$ and the aspect ratio is $\Delta = R^{\Theta(R)}$. By solving for $R$, we obtain $R = \Theta(\log\Delta/\log\log\Delta)$.
\end{itemize}

\begin{claim}[Exact opening probabilities]\label{clm:biases}
Deterministically, for every region $c$ that is processed:
\emph{(i)} when attempt $u_{c,i}$ is presented, its adaptive opening
probability equals $(s + (i-1)\,\delta_\star)/f \in [1/R,\, 2/R\,]$, and its
oblivious opening probability is at most $2/R$;
\emph{(ii)} in Phase~2, every copy of $z^*$ presented before the adaptive
run's first opening among the copies has adaptive opening probability
exactly $\eta/f$;
\emph{(iii)} in extraction round $j$, every copy of $y_j$ presented before
the adaptive run's first opening in that round has adaptive opening
probability exactly $a_j = (\rho_{j-1} - \rho_j)/f$, uniformly in
$j \in [R]$ (with the convention $y_0 = u^*$).
\end{claim}

\begin{proof}
By (P1), only within-region centers matter; all quoted distances are the
exact ones from the layout. We now show each of the claims below.

\emph{(i)} Suppose the adaptive centers
of region $c$ are exactly $\{h_c\}$ and attempts $1, \dots, i-1$ were rejected
(otherwise Phase~1 would have ended at a smaller index) and no other point
of the region has been presented. Then,  the adaptive distance is
$s + (i-1)\,\delta_\star < 2s < f$. The oblivious centers of region $c$
form a subset of $\{h_c, u_{c,1}, \dots, u_{c,i-1}\}$ containing $h_c$
(by (P0)), so the oblivious distance is at most $2s$.

\emph{(ii)} The adaptive centers of region $c$ are exactly $\{h_c, u^*\}$:
all other presented attempts were rejected, and the copies of $z^*$
presented so far were rejected. The nearest is $u^*$ at distance $\eta$
(the hub is at distance $\ge s$).

\emph{(iii)} The adaptive centers of region $c$ are $\{h_c, u^*\}$ together
with one opened copy from each of extraction locations $1, \dots, j-1$; no copy of $z^*$
is an adaptive center, since the region is certified. The nearest is the location-$(j-1)$ center $y_{j-1}$ (for $j = 1$: $y_0 = u^*$), at distance
$\rho_{j-1} - \rho_j$: earlier locations $y_{j'}$ with $j' < j-1$ are at the
larger distances $\rho_{j'} - \rho_j$, $u^*$ is at
$\rho_0 - \rho_j \ge \rho_{j-1} - \rho_j$, and the hub is at distance
$\ge s$.
\end{proof}

The replay-side arguments rely on the following observation: 
The generator itself is deterministic, so conditioning
on the adaptive run's coins determines the stream, and the replay is then
just the Meyerson sketch on a fixed sequence.

\begin{lemma}[Conditional replay]\label{lem:condreplay}
The presented sequence is a function of the adaptive coins $U$ alone, and $U' \perp U$. Hence,
conditionally on $\sigma(U)$, the stream is fixed, and with it every quantity it
determines: the phase and round lengths, all region outcomes, the first
certified region, and $\tau$. The oblivious run is
exactly the Meyerson sketch executed with i.i.d.\ fresh coins on this
fixed sequence; in particular, each point presented at time $t$ is
opened, given $\sigma(U)$ and $U'_{<t}$, with probability
$\min\{1,\, d(x_t, S^{\OB}_{t-1})/f\}$.
\end{lemma}

\begin{proof}
$x_t$ is $\sigma(U_{<t})$-measurable by induction on the transcript, and
$U'_t$ is independent of $(U, U'_{<t})$.
\end{proof} 

\begin{claim}[Selection and the anti-certificate]\label{clm:anticert}
For each region $c$ that is started: \emph{(i)}
$\Pr[\,T_c > n_1\,] \le 1/R$; \emph{(ii)} let
$B_c := \{\text{the oblivious run opened } u_{c,T_c}\}$ (defined on
$\{T_c \le n_1\}$); then $\Pr[\,B_c \mid \sigma(U)\,] \le 2/R$ on
$\{T_c \le n_1\}$.
\end{claim}

\begin{proof}
\emph{(i)} By \Cref{clm:biases}(i), each presented attempt is opened by the
adaptive run with probability at least $1/R$ by a fresh coin $U$, so
$\Pr[T_c > n_1] \le (1-1/R)^{n_1} \le e^{-n_1/R} = e^{-\lceil \ln R\rceil} \le 1/R$.
\emph{(ii)} Condition on $\sigma(U)$: the index $T_c$ and the entire
presented prefix are fixed, and attempt $u_{c,T_c}$ was presented exactly
once, at a fixed time. By \Cref{lem:condreplay}, its replay coin is fresh
and its conditional opening probability, given additionally the replay
coins of the previously presented points, is the oblivious bias of
\Cref{clm:biases}(i), which is at most $2/R$ \emph{pathwise}.
Integrating over those coins gives the claim.
\end{proof}

Conditioned on the event
$\{T_c \le n_1\} \cap \neg B_c$, at the start of
Phase~2 \emph{every} oblivious center of region $c$ is at distance
strictly greater than $\delta_\star$ from $z^*$: no point of the region to
the right of $u^*$ has yet been presented (the points $u_{c,j}$ with
$j > T_c$ are never presented; $z^*$ and the $y_j$ come later), so every
oblivious center other than the hub sits at some $u_{c,j}$ with
$j \le T_c - 1$, at distance
$(T_c - j)\,\delta_\star + \eta \ge \delta_\star + \eta > \delta_\star$
from $z^*$, while the hub is at distance $\ge s > \delta_\star$. 

\begin{claim}[Certification]\label{clm:cert}
For each region that selects: \emph{(i)} the region is certified with
probability exactly $(1-\eta/f)^{n_2} = (1 - 1/n_2)^{n_2} \ge 1/4$, and
this event is
$\sigma(U)$-measurable; \emph{(ii)} let
\[ \mathrm{HIT} := \{\text{the oblivious run opens at least one copy of } z^*\} ;\] then
\[
\Pr\bigl[\,\neg \mathrm{HIT} \;\big|\; \sigma(U),\ \neg B_c\,\bigr]
\;\le\; \Bigl(1 - \frac{\delta_\star}{f}\Bigr)^{n_2} \;\le\; \frac1R
\qquad \text{on the event that the region is certified.}
\]
\end{claim}

\begin{proof}
\emph{(i)} By \Cref{clm:biases}(ii) each copy is rejected by the adaptive
run with probability exactly $1 - \eta/f$, by fresh coins $U$;
certification is
the intersection of $n_2$ such rejections. Since $\eta/f = 1/n_2$ and
$(1 - 1/m)^{m}$ is increasing in $m$ with value $1/4$ at $m = 2$,
$(1-\eta/f)^{n_2} \ge 1/4$.
\emph{(ii)} Condition on $\sigma(U)$ and on the replay coins of the
attempts (which determine $B_c$). On the stated event
the region is certified, so all $n_2$ copies were presented, and by the
paragraph above every oblivious center of the region is at distance
$> \delta_\star$
from $z^*$ at Phase-2 start. This remains true until the replay's first
opening among the copies (only copies of $z^*$ are presented during
Phase~2), so by \Cref{lem:condreplay} each copy's conditional opening
probability on the surviving path is at least
$\min\{1, \delta_\star/f\} = \delta_\star/f$. Telescoping the conditional
survival probabilities over the $n_2$ fresh copy coins,
$\Pr[\neg \mathrm{HIT} \mid \cdot\,] \le (1-\delta_\star/f)^{n_2}
\le e^{-n_2 \delta_\star/f} = e^{-\lceil\ln R\rceil} \le 1/R$.
\end{proof}

\begin{claim}[Extraction accounting]\label{clm:extract}
Let $\mathrm{SUCC}$ be the event that some region is certified. On
$\mathrm{SUCC}$, extraction round $j$ lasts $T_j \sim \Geom(a_j)$ copies
($\sigma(U)$-measurable, a.s.\ finite) and contributes exactly one adaptive
opening, so $K^{\A}_\tau \ge R$ on $\mathrm{SUCC}$. Moreover,
\[
\E\bigl[\#\{\text{oblivious openings during extraction}\}\bigr]
\;\le\; \underbrace{\tfrac{R}{R-1}}_{\text{on } \mathrm{HIT}}
\;+\; \underbrace{R \cdot \Pr[\,\neg\mathrm{HIT} \cap \mathrm{SUCC}\,]}_{\text{on } \neg\mathrm{HIT}}
\;\le\; 1.1 + 3 \;=\; 4.1 .
\]
\end{claim}

\begin{proof}
The round structure and $K^{\A}_\tau \ge R$ are \Cref{clm:biases}(iii). On
$\mathrm{HIT}$, every copy of $y_j$ is at distance exactly $\rho_j$ from the oblivious
center at $z^*$, so by \Cref{lem:condreplay} its conditional opening
probability given $\sigma(U)$ and the prior replay coins is at most
$\rho_j/f$. Since $T_j$ is $\sigma(U)$-measurable,
$\E[\#\text{openings in round } j;\, \mathrm{HIT} \mid \sigma(U)] \le
(\rho_j/f)\, T_j$, and taking expectations,
$\E[\cdot\,; \mathrm{HIT}] \le (\rho_j/f)\,\E[T_j] = \rho_j/(f a_j)
= \rho_j/(\rho_{j-1} - \rho_j) = \tfrac{1}{R-1}$, uniformly in $j$;
summing over the $R$ rounds gives at most $\tfrac{R}{R-1} \le 1.1$. On $\neg \mathrm{HIT}$, the copies within a round
are identical, so after the replay's first opening in a round every further
copy of that round is at distance $0$ and is never opened: at most one
oblivious opening per round, hence at most $R$ in total on this event.
Finally, all region outcomes and the identity of the first certified region
are $\sigma(U)$-measurable (\Cref{lem:condreplay}), so
\Cref{clm:anticert,clm:cert} apply conditionally on
$\mathrm{SUCC}$, bounding the \emph{silent failure} probability ($\neg \mathrm{HIT}$ on
a certified region):
\[
\Pr[\,\neg\mathrm{HIT} \mid \mathrm{SUCC}\,]
\;\le\; \Pr[B_{c^*}\mid\cdot]
+ \Pr[\neg\mathrm{HIT} \mid \cdot, \neg B]
\;\le\; \frac{2}{R} + \frac1R \;=\; \frac{3}{R}. \qedhere
\]
\end{proof}

\begin{proofof}{\Cref{thm:anticert}}
\emph{Certification probability and region count.} By
\Cref{clm:anticert}(i) and \Cref{clm:cert}(i), each started region is
certified
with probability at least $(1 - 1/R)\cdot\tfrac14 \ge \tfrac15$
(independently across regions, as all the coins involved are fresh). Hence
the number of started regions is stochastically dominated by a
$\Geom(1/5)$ variable, $\E[\#\text{regions}] \le 5$, and
$\Pr[\neg\mathrm{SUCC}] \le (4/5)^{N_2} \le e^{-5\ln R\cdot\ln(5/4)}
\le R^{-1.1} \le 1/R$.

\emph{Adaptive count.} By \Cref{clm:extract},
$\E[K^{\A}_\tau] \ge R\,\Pr[\mathrm{SUCC}] \ge R(1 - 1/R) = R - 1$.

\emph{Oblivious count.} Per started region, conditionally on it being
started (a $\sigma(U)$-event): the hub contributes exactly $1$; the
attempts contribute, by \Cref{clm:biases}(i) and \Cref{lem:condreplay},
at most $(2/R)\,\E[T_c \wedge n_1] \le (2/R) \cdot R = 2$
(the domination $\Pr[T_c > i] \le (1-1/R)^i$ gives
$\E[T_c \wedge n_1] \le R$); and the copies of $z^*$ contribute at
most $1$ (identical copies). Hence
\[
\E\bigl[K^{\OB}_\tau\bigr]
\;\le\; \E[\#\text{regions}]\cdot(1 + 2 + 1) \;+\; 4.1
\;\le\; 5 \cdot 4 + 4.1 \;=\; 24.1 \;\le\; 25 \;=\; C_0 .
\]

\emph{Resources.} The smallest distance between distinct presented points
is $d(y_R, z^*) = \rho_R = \eta R^{-R}$ (the attempt spacing $\delta_\star$
exceeds $\eta$, which exceeds $\rho_1$, and consecutive extraction points are
$\rho_{j-1} - \rho_j \ge (R-1)\rho_R$ apart); the largest is at most
$8f(N_2 + 1)$. Since $f/\eta = R^2\lceil\ln R\rceil^2 = O(R^2\ln^2 R)$,
\[
\log\Delta \;=\; R\ln R + O(\log R) \;=\; \Theta(R\log R),
\]
whence $R = \Omega(\log\Delta/\log\log\Delta)$ and the main display.
\end{proofof}

\subsection{Extensions}
Note that using a single region (as described above) already suffices for the expectation gaps in the statement of \Cref{thm:anticert}. The use of multiple regions yields the following high probability claim:
\begin{cor}[High-probability form]\label{cor:whpattack}
In the construction above,
\begin{equation}\label{whpratio:eq}
\Pr\Bigl[\,K^{\A}_\tau \ge R \;\text{ and }\; K^{\OB}_\tau = O(\log R)\,\Bigr] \;\ge\; 1 - O(1/R),
\end{equation}
so $K^{\A}_\tau/K^{\OB}_\tau = \Omega(R/\log R) = \Omega\bigl(\log\Delta/(\log\log\Delta)^2\bigr)$ with probability $1 - O(1/R)$.
\end{cor}
\begin{proof}
If there is a successful region, which holds with probability $\Pr[\mathrm{SUCC}] \ge 1-R^{-1.1}$), then $K^{\A}_\tau \ge R$ deterministically. With  probability $1 - O(1/R)$, we reach a successful region after $O(\log R)$ attempts, since each attempt succeeds with constant probability, and thus $K^{\OB}_\tau = O(\log R)$.\footnote{The $\log R$ loss in \eqref{whpratio:eq} is inherent to the construction: each attempted region fails with constant probability and contributes its hub to $K^{\OB}_\tau$, so with probability $\Omega(1/R)$ there are $\Omega(\log R)$ failed regions.}
\end{proof}

\begin{remark}[Fixed horizon and budget exhaustion stopping rules]\label{rem:budget}
The two canonical stopping rules mentioned earlier (\Cref{sec:model}) achieve the same asymptotic gap via adapted constructions:

\emph{Budget exhaustion.} Stopping when the
$k$th adaptive center opens: take $R = k$ and run the generator until the
\emph{total} adaptive center count reaches $k$. A few of the $k$ openings
are spent on steps before extraction begins (in failed regions, and on
the hub and selection of the certified region); but these pre-extraction
openings, in both runs, number only $O(1)$ in expectation (at most $3$
per started region and $\E[\#\text{regions}] \le 5$). Therefore
$\Omega(k)$ of the budget is spent on extraction steps.

\emph{Fixed horizon.} Stop after a \emph{deterministic} number of steps. In this case, a single region suffices:
allot deterministic windows of $n_1$, $n_2$, and $\lceil 1/a_j\rceil$
($j \in [R]$) steps to Phase~1, Phase~2, and extraction location $j$; run each phase
inside its window; and pad every unused step (after a selection or a
failure, and all extraction windows when the region is not certified) with
requests at the hub, which cost nothing in either run. The adaptive run
then opens each extraction location with probability at least $1 - 1/e$ (a missed location
only increases later biases), so $\E[K^{\A}_\tau] = \Omega(R)$ on
certification, which retains constant probability, while the replay
accounting is unchanged. The  deflation attack (\Cref{sec:lower_reverse}) uses a similar fixed-length block structure.
\end{remark}

\subsection{Routing Bounds}

We now establish the routing cost claims of \Cref{thm:anticert}.

\begin{proofof}{\Cref{thm:anticert}, routing bounds}
We bound the 
\emph{adaptive routing cost} from below, by only considering the contribution of the extraction rounds. Condition on $\mathrm{SUCC}$ and on the
pre-extraction transcript. 
For the $j$-th extraction round: By \Cref{clm:biases}(iii), every copy of $y_j$
presented before the round-$j$ adaptive opening is at adaptive distance
exactly $f a_j$, and by \Cref{clm:extract} round $j$ lasts
$T_j \sim \Geom(a_j)$ copies, of which the first $T_j - 1$ are rejected
and each pays routing $f a_j$; the parameters $a_j$ are the same for every
region and selection index, so the conditioning does not affect the
geometric law. Hence round $j$ contributes conditional expected adaptive
routing cost
\[
f a_j \bigl(\E[T_j] - 1\bigr)
\;=\; f a_j \Bigl(\frac{1}{a_j} - 1\Bigr)
\;=\; f\,(1 - a_j),
\]
and summing over rounds, using the telescoping identity
\[
\sum_{j=1}^{R} a_j
\;=\; \frac{\rho_0 - \rho_R}{f}
\;\le\; \frac{\eta}{f} \;\le\; 1,
\]
the extraction
phase alone contributes at least $Rf - f = (R-1)f$ on $\mathrm{SUCC}$.
Routing costs are nonnegative, so dropping the pre-extraction
contribution and
using $\Pr[\mathrm{SUCC}] \ge 1 - 1/R$ from the proof of
\Cref{thm:anticert},
\[
\E\bigl[Q^{\A}_\tau\bigr]
\;\ge\; \Pr[\mathrm{SUCC}] \cdot (R-1) f
\;\ge\; \Bigl(1 - \frac{1}{R}\Bigr)(R-1) f
\;\ge\; (R-2) f .
\]

\emph{Oblivious routing cost.} By \eqref{eq:routingdom} and the oblivious count
bound of \Cref{thm:anticert},
\[
\E\bigl[Q^{\OB}_\tau\bigr]
\;\le\; f\,\E\bigl[K^{\OB}_\tau\bigr]
\;\le\; C_0\, f . \qedhere
\]
\end{proofof}

\subsection{Posterior Pricing: Why Extraction Requires Confidence}
\label{sec:posterior}

Recall that the inflation generator first manufactures a point that the
replay very likely holds as a center and the adaptive run does not, and then exploits this asymmetry $R$ times. We show here that the first step is
unavoidable: a one-sided generator cannot obtain adaptive openings
cheaply unless it is confident that the replay already covers the
presented point. This subsection motivates the design of the attack but
is not used in the analysis.

\begin{restatable}[Posterior pricing]{lem}{lemposterior}\label{lem:posterior}
For $x_t \ne \bot$ let $r_t = d(x_t, S^{\A}_{t-1})$,
$\bar b_t = \E[\,b_t \mid \Ff_{t-1}]$, and
$\pi_t = \Pr\bigl[\,S^{\OB}_{t-1} \cap B(x_t, r_t/2) = \emptyset
\mid \Ff_{t-1}\bigr]$; all three are $\Ff_{t-1}$-measurable. Then
$\bar b_t \ge \pi_t\, a_t / 2$, and consequently, for every one-sided
generator, every $\pi \in (0,1]$, and every stopping time $\tau$,
\[
\E\Bigl[\#\bigl\{\text{adaptive openings at steps } t \le \tau
\text{ with } \pi_t \ge \pi\bigr\}\Bigr]
\;\le\; \frac{2}{\pi}\,\E\bigl[K^{\OB}_\tau\bigr].
\]
\end{restatable}

\begin{proof}
On the event $\{S^{\OB}_{t-1} \cap B(x_t, r_t/2) = \emptyset\}$ we have
$d(x_t, S^{\OB}_{t-1}) \ge r_t/2$, so
$b_t \ge \min\{f, r_t/2\}/f \ge \min\{f, r_t\}/(2f) = a_t/2$; taking
conditional expectations gives $\bar b_t \ge \pi_t a_t/2$. The indicator
$h_t = \ind{\{\pi_t \ge \pi\}}$ is $\Ff_{t-1}$-measurable, hence predictable;
by \Cref{lem:equiv} and the tower property,
\[
\E\Bigl[\sum_{t\le\tau} o_t h_t\Bigr]
= \E\Bigl[\sum_{t\le\tau} a_t h_t\Bigr]
\le \frac{2}{\pi}\,\E\Bigl[\sum_{t\le\tau} \bar b_t h_t\Bigr]
= \frac{2}{\pi}\,\E\Bigl[\sum_{t\le\tau} b_t h_t\Bigr]
\le \frac{2}{\pi}\,\E\bigl[K^{\OB}_\tau\bigr]. \qedhere
\]
\end{proof}

Quantitatively, this means that any generator with $\E[K^{\A}_\tau] \ge R$ and
$\E[K^{\OB}_\tau] \le C_0$ must place at least half its adaptive
openings, in expectation, at steps with lack-posterior
$\pi_t \le 4C_0/R$. Indeed, the construction presented above is organized around this constraint: the preparation phase finds a point that the replay holds as a center with probability
$1 - \widetilde O(1/R)$ while the adaptive run does not, at the cost of $O(1)$ leaked openings, and the extraction phase then utilizes this point $R$ times.

\section{The Deflation Attack}\label{sec:lower_reverse}

The deflation direction admits a matching deterministic attack which can be carrier out by a one-sided generator on the real line. Moreover, the stopping time $T$ is fixed in advance.
\begin{restatable}[Deflation attack]{thm}{thmdeflattack}
\label{prop:reverse}
For every integer $R \ge 2$ there exist a finite point
configuration on the real line, a deterministic one-sided adaptive
generator, and a \emph{deterministic} stopping time $\tau = T$ (a fixed
number of steps) with
\[
\E\bigl[K^{\A}_\tau\bigr] \;\le\; 4,
\qquad
\E\bigl[K^{\OB}_\tau\bigr] \;\ge\; \frac{R}{7},
\]
and, separately in the routing components,
\[
\E\bigl[Q^{\A}_\tau\bigr] \;\le\; \tfrac32\, f,
\qquad
\E\bigl[Q^{\OB}_\tau\bigr] \;\ge\; \frac{R f}{28},
\]
where the presented sequence has aspect ratio at most $2R^{R}$ almost
surely. Consequently, on streams of aspect ratio $\Delta \le 2R^{R}$,
\[
\frac{\E[K^{\A}_\tau]}{\E[K^{\OB}_\tau]} \;\le\; \frac{28}{R},
\qquad
\frac{\E[C^{\A}_\tau]}{\E[C^{\OB}_\tau]} \;\le\; \frac{56}{R},
\qquad
\frac{\E[Q^{\A}_\tau]}{\E[Q^{\OB}_\tau]} \;\le\; \frac{42}{R},
\]
and $R = \Om(\log\Delta/\log\log\Delta)$, so all three deflation ratios
are $O(\log\log\Delta/\log\Delta)$.
\end{restatable}

The remainder of this section describes the deflation generator and proves \Cref{prop:reverse}. The construction requires an adaptive-only center certified with \emph{constant} (not
high, in departure from the inflation generator) confidence.

\subsection{The Deflation Generator}

\medskip\noindent\textbf{Configuration.}
Fix an integer $R \ge 2$. Let $\eps_j = \tfrac12 R^{-j}$
for $j = 0, 1, \dots, R$, and place on the real line the \emph{hub}
$h = 0$, the \emph{seed} $z = \tfrac f2$, and the \emph{extraction locations}
$y_j = z + \eps_j f$ for $j \in [R]$. The seed plays the role of the
inflation attack's certified pair: a single fresh coin will supply
constant confidence that $z$ is an adaptive-only center, where the
inflation attack spent $n_2$ requests for confidence $1 - O(1/R)$.

\medskip\noindent\textbf{The generator:} present $h$ at step $1$
and $z$ at step $2$; if the adaptive run opened $z$ (the event
$S = \{o_2 = 1\}$, observable from the adaptive transcript), present
$k_j = 2/\eps_{j-1} = 4R^{\,j-1}$ copies of $y_j$ (which we refer to as block $B_j$), for
$j = 1, \dots, R$ in order; otherwise present $h$ at every remaining
step.

The number of steps is $\tau = T = 2 + \sum_{j=1}^{R} k_j$ (deterministically) and in particular this is an $(\Ff_t)$- and
$(\Gg_t)$-stopping time. Every choice is a deterministic function of the
adaptive transcript, so this is a deterministic one-sided generator. The
padding incurs no routing cost and no openings, as a repeatedly-presented hub is at
distance $0$ from an open center in both runs.

\subsection{Analysis}

\begin{claim}[Geometry]\label{clm:revgeom}
Exactly: $d(z, h) = \tfrac f2$; $d(y_j, z) = \eps_j f$;
$d(y_j, h) = (\tfrac12 + \eps_j)f$; and
$d(y_j, y_{j'}) = (\eps_{j'} - \eps_j)f \ge (\eps_{j-1} - \eps_j)f$
for $j' < j$. Moreover
$\eps_{j-1} - \eps_j = \eps_{j-1}(1 - \tfrac1R) \ge \tfrac12
\eps_{j-1} \ge \eps_j$. On the event $S$ the distinct presented
locations are $\{h, z, y_1, \dots, y_R\}$, with
$d_{\min} = \eps_R f$ and
$d_{\max} = (\tfrac12 + \eps_1) f \le \tfrac34 f$, so pathwise
$\Lambda \le \Delta \le \tfrac{3}{4}\,\eps_R^{-1} \le 2 R^{R}$; on
$S^{c}$ only $\{h, z\}$ are presented and $\Delta = 1$.
\end{claim}

\begin{proof}
All distances are immediate from the positions; the minimum over
distinct pairs on $S$ is between $y_R$ and $z$, since
$(\eps_{j-1} - \eps_j) \ge \eps_j \ge \eps_R$ for all $j$ and the hub
is farther from every point than $z$ is. The cap at $f$ is inactive
as $d_{\max} < f$.
\end{proof}

\begin{claim}[Run analysis]\label{clm:revrun}
Write $N = \{o'_2 = 0\}$ for the event that the replay rejects the
seed. In every realization:
\emph{(i)} at step $1$ both runs face distance $+\infty$ and open $h$
surely; at step $2$ both runs' nearest center is $h$ at distance
exactly $\tfrac f2$, so $a_2 = b_2 = \tfrac12$, decided by the fresh,
mutually independent coins $U_2, U'_2$; in particular
$\Pr[S] = \Pr[N] = \tfrac12$ and $S, N$ are independent.
\emph{(ii)} On $S$, at every step of block $B_j$ the adaptive run's
nearest center is at distance at most $\eps_j f$, so
$a_t \le \eps_j$ throughout $B_j$.
\emph{(iii)} On $S \cap N$, at every step of block $B_j$ that
precedes the replay's first opening within $B_j$, the replay's
nearest center is at distance at least
$(\eps_{j-1} - \eps_j) f$, so
$b_t \ge \delta_j := \eps_{j-1}(1 - \tfrac1R) \ge \tfrac12
\eps_{j-1}$.
\end{claim}

\begin{proof}
(i) is the first-point convention and \Cref{clm:revgeom}; the coins
at step $2$ are fresh and mutually independent by the model.
(ii) On $S$ the adaptive run holds $z$, at distance $\eps_j f$ from
$y_j$; the nearest-center distance is at most this.
(iii) On $S \cap N$ the replay never holds $z$: the seed is presented
exactly once, at step $2$, and $o'_2 = 0$ on $N$. Before its first
opening within $B_j$, the replay's centers among presented locations
are contained in $\{h\} \cup \{y_{j'} : j' < j\}$ (re-presented hubs
are never opened, and earlier blocks contribute only copies of
$y_{j'}$ with $j' < j$). By \Cref{clm:revgeom} the distance from
$y_j$ to every such point is at least
$(\eps_{j-1} - \eps_j) f = \delta_j f$, and
$\delta_j \le \eps_0 = \tfrac12 < 1$, so
$b_t = \min\{1, d/f\} \ge \delta_j$. Note the claim quantifies over
all candidate center sets, so no induction on the replay's path is
needed.
\end{proof}

\begin{claim}[Block yield]\label{clm:blockyield}
Fix $j \in [R]$ and let $t_0$ be the first step of block $B_j$.
Conditionally on $\Gg_{t_0 - 1}$, on the event $S \cap N$:
\emph{(i)} the replay opens at least one center during $B_j$ with
probability at least $1 - e^{-1}$; and
\emph{(ii)} the expected replay routing incurred during $B_j$ is at
least $f/7$.
\end{claim}

\begin{proof}
Within $B_j$ only copies of $y_j$ are presented, so the replay's center
set changes during the block only when it opens a copy of $y_j$; hence,
until that first opening, the replay's nearest-center distance is a
$\Gg_{t_0-1}$-measurable constant $d$, with $d \ge \delta_j f$ by
\Cref{clm:revrun}(iii) and $d \le (\tfrac12 + \eps_1) f \le \tfrac34 f$
via the hub, which the replay holds surely. Write $q = d/f$ and let $G$
be the index of the first replay opening in the block (possibly
$\infty$). Conditionally on $\Gg_{t_0-1}$, on $S \cap N$, the coins of
the block are fresh, so each pre-plug copy is opened with probability
$q$ and, if rejected, routes $d$. Using
$q k_j \ge \delta_j \cdot \tfrac{2}{\eps_{j-1}} = 2(1 - \tfrac1R) \ge 1$:
\emph{(i)} the probability of at least one opening is
$1 - (1-q)^{k_j} \ge 1 - e^{-q k_j} \ge 1 - e^{-1}$; and \emph{(ii)}
the expected replay routing is
\[
d\;\E\bigl[\min(G - 1,\ k_j)\bigr]
\;=\; (f - d)\bigl(1 - (1-q)^{k_j}\bigr)
\;\ge\; \frac f4\,\bigl(1 - e^{-1}\bigr)
\;\ge\; \frac f7 ,
\]
using $\E[\min(G-1,k)] = \sum_{i=1}^{k}(1-q)^i =
\tfrac{1-q}{q}\bigl(1 - (1-q)^{k}\bigr)$, $d\,\tfrac{1-q}{q} = f - d
\ge \tfrac f4$, and part \emph{(i)}.
\end{proof}

\begin{proof}[Proof of \Cref{prop:reverse}]

\emph{Aspect ratio.} Pathwise from \Cref{clm:revgeom}.

\emph{Adaptive count.} The stopping time is deterministic, so
$\E[K^{\A}_\tau] = \E[\sum_{t \le T} a_t]$ by \Cref{lem:equiv} with
weight $h_t \equiv 1$ (or by linearity and the tower property). Step $1$
contributes $1$ and step $2$ contributes $\tfrac12$
(\Cref{clm:revrun}(i)). On $S^{c}$ all remaining steps re-present the
hub and contribute $0$. On $S$, by \Cref{clm:revrun}(ii), block $B_j$
contributes at most
\[
k_j\, \eps_j
\;=\; \frac{2}{\eps_{j-1}}\,\eps_j
\;=\; \frac{2}{R},
\]
so the $R$ blocks contribute at most $2$ and
\[
\E\bigl[K^{\A}_\tau\bigr]
\;\le\; \tfrac32 + \Pr[S]\cdot 2
\;=\; \tfrac52 \;\le\; 4.
\]

\emph{Oblivious count.} Let $M$ be the number of blocks in which the
replay opens at least one center; then $K^{\OB}_\tau \ge M$
pointwise. Since $S \cap N$ is $\Gg_2$-measurable,
\Cref{clm:blockyield}(i), the tower property, and summing over the
$R$ blocks give
\[
\E\bigl[K^{\OB}_\tau\bigr]
\;\ge\; \E\bigl[M \ind{S \cap N}\bigr]
\;\ge\; \Pr[S \cap N] \cdot \bigl(1 - e^{-1}\bigr) R
\;\ge\; \tfrac14 \cdot 0.63\, R
\;\ge\; \frac{R}{7}.
\]

\emph{Adaptive routing.}
Step $1$ incurs no routing (both runs open surely) and re-presented
hubs are at distance $0$. The seed routes $\tfrac f2$ iff $o_2 = 0$,
contributing $\tfrac12 \cdot \tfrac f2 = \tfrac f4$. On $S$, every
step of block $B_j$ is at adaptive distance at most $\eps_j f$
(\Cref{clm:revrun}(ii)), so its routing is at most $\eps_j f$
pointwise and block $B_j$ contributes at most
$k_j \eps_j f = (2/R) f$. Summing over the $R$ blocks,
\[
\E\bigl[Q^{\A}_\tau\bigr]
\;\le\; \frac f4 + \Pr[S]\cdot 2 f
\;=\; \frac54\, f \;\le\; \frac32\, f .
\]

\emph{Replay routing.} By \Cref{clm:blockyield}(ii), the tower
property, and summing over the $R$ blocks,
\[
\E\bigl[Q^{\OB}_\tau\bigr]
\;\ge\; \Pr[S \cap N] \cdot R \cdot \frac f7
\;=\; \frac{R f}{28} .
\]

\emph{Ratio and costs.} By the bounds just proved, the count ratio is
at most $4/(R/7) = 28/R$; \eqref{eq:costcenter} for each run gives
$\E[C^{\A}_\tau] \le 2f\,\E[K^{\A}_\tau]$ and
$\E[C^{\OB}_\tau] \ge f\,\E[K^{\OB}_\tau]$, whence the cost ratio is
at most $56/R$; and $\E[Q^{\A}_\tau] \le \tfrac32 f$ against
$\E[Q^{\OB}_\tau] \ge Rf/28$ gives routing ratio at most $42/R$.
Finally $\log_2\Delta \le R\log_2 R + 1$ gives
$R = \Om(\log\Delta/\log\log\Delta)$.
\end{proof}

\section{Online Summaries and $k$-Clustering}\label{sec:onlinesummary}

The center set of the Meyerson sketch, weighted by the points
routed to each center, is a compact \emph{summary} (or coreset) of the stream.
In this section, we analyze the correctness guarantees provided by this summary and devise online $k$-clustering sketches under an adaptively chosen point set. 

\subsection{$k$-clustering and facility location costs}\label{sec:offlineobj}
Suppose $(\Mm,d)$ satisfies \cref{def:semimetric} for some $\alpha \geq 1$.
A \emph{weighted point set} $P$ assigns
multiplicities $w(x) > 0$ to the points (unweighted sequences carry unit
weights).
For a weighted $P$ and a center set $C \subseteq \Mm$, the clustering cost is defined as
\[
\mathrm{cost}(P, C) \;=\; \sum_{x \in P} w(x)\, d(x, C).
\]
The optimal cost of the $k$-clustering for $P$ is given by
\[
\mathrm{OPT}_k(P) \;=\; \min_{|C| = k} \mathrm{cost}(P, C)\ .
\]
The $k$-clustering and facility-location optima at price $f$ are
related by
\begin{equation} \label{fl2med:eq}
    \mathrm{OPT}_f(P) \;=\;
    \min_{F \subseteq \Mm}\,\bigl\{f|F| + \mathrm{cost}(P,F)\bigr\}
    \;=\; \min_{k \ge 1}\,\bigl\{fk + \mathrm{OPT}_k(P)\bigr\}\ 
\end{equation}
For each $k$, optimal $k$-clustering centers are a feasible
facility-location solution of cost $fk + \mathrm{OPT}_k(P)$, and
conversely, among facility-location solutions with exactly $k$
centers, the routing cost $\mathrm{cost}(P, F)$ is minimized by
optimal $k$-clustering centers. (As a function of $f$,
$\mathrm{OPT}_f(P)$ is thus the lower envelope of the lines
$f \mapsto fk + \mathrm{OPT}_k(P)$.)

Consequently, by selecting the price based on the scale of the $k$-clustering objective, we can track the facility location objective. Concretely, fix $k$ and take $f = \lambda/k$ for an upper bound
$\lambda \ge \mathrm{OPT}_k(P)$; the $k$-th term of \eqref{fl2med:eq}
gives
\begin{equation} \label{med2fl:eq}
    \mathrm{OPT}_f(P) \;\le\; fk + \mathrm{OPT}_k(P) \;\le\; 2\lambda\ .
\end{equation}
In the other direction, $\mathrm{OPT}_f(P) \ge \mathrm{OPT}_k(P)$: a
facility-location solution either opens at most $k$ centers, paying
routing at least $\mathrm{OPT}_k(P)$, or opens more, paying
$fk = \lambda \ge \mathrm{OPT}_k(P)$ in facilities alone. Hence, for  $\lambda \le 2\,\mathrm{OPT}_k(P)$, we have that
\[
\mathrm{OPT}_k(P) \;\le\; \mathrm{OPT}_f(P) \;\le\; 4\,\mathrm{OPT}_k(P).
\]
In particular, a geometric grid of prices tracks the $k$-clustering optimum, for
every $k$ simultaneously up to constant factors.

\subsection{Meyerson Online Summary} \label{sec:Meyersonsummary}

Let $P_t = (x_1, \dots, x_t)$ denote the prefix of the stream until time $t$.
\begin{defn}[Online summary and its distortion]\label{def:onlinesummary}
An algorithm is an \emph{online summary} if at every time $t$, it
maintains a set $S_t \subseteq P_t$ and an implicit assignment
$\pi_t : P_t \to S_t$, of which it publishes only the weighted point set
\[
\widehat P_t \;=\; (S_t,\, w_t), \qquad
w_t(y) = \bigl|\pi_t^{-1}(y)\bigr| \ \text{ for } y \in S_t .
\]
Its \emph{distortion} is
$\mathrm{dist}_t = \sum_{s \le t} d\bigl(x_s, \pi_t(x_s)\bigr)$.
\end{defn}

\begin{fact}\label{fact:distortion}
For every $t$ and every center set $C$,
\[
\mathrm{cost}(P_t, C)
\;\le\; \alpha\bigl(\mathrm{cost}(\widehat P_t, C) + \mathrm{dist}_t\bigr)
\qquad\text{and}\qquad
\mathrm{cost}(\widehat P_t, C)
\;\le\; \alpha\bigl(\mathrm{cost}(P_t, C) + \mathrm{dist}_t\bigr) .
\]
For a metric ($\alpha = 1$), the inequalities above imply that
$\bigl|\mathrm{cost}(P_t, C) - \mathrm{cost}(\widehat P_t, C)\bigr| \le \mathrm{dist}_t$.
\end{fact}
\begin{proof}
For each point $x_s$, $d(x_s, C) \le \alpha\bigl(d(x_s, \pi_t(x_s)) + d(\pi_t(x_s), C)\bigr)$. Additionally, the same inequality holds with the roles of $x_s$ and $\pi_t(x_s)$ exchanged. By summing over all $s \le t$, the claim follows.
\end{proof}

Next, the Meyerson sketch as described in \Cref{alg:meyerson} can be augmented to maintain an online summary and an estimate of the distortion. Each arriving point is assigned to the center that determined its opening probability: concretely, the point is assigned to itself if a center is opened there, and otherwise the point is assigned to the nearest center in $S_{t-1}$. The published summary $\widehat P_t$ is then the center set weighted by the multiplicities of assigned points. By construction, the
distortion at time $t$ is equal to the routing cost at time $t$:

\begin{equation}\label{eq:distisQ}
\mathrm{dist}_t \;=\; Q_t ,
\end{equation}
which is also maintained by adding the distances of points to their assigned centers.

The lemma below presents the bounds \Cref{thm:upper,cor:transfer} and \eqref{eq:routingdom} for our online summary.

\begin{lem}[Online summaries under adaptive inputs]\label{lem:onlinesummary}
Run the Meyerson sketch at price $f > 0$ on the output of an adaptive
generator, publishing the online summary. Let
$\Phi$ bound the sketch's expected cost on fixed sequences as in
\Cref{cor:transfer}. Then for every $(\Gg_t)$-stopping time $\tau$:
\begin{enumerate}
\item[(i)] $\E\bigl[\mathrm{dist}^{\A}_\tau\bigr]
  \;\le\; \kappa_\Lambda\,\E\bigl[\Phi(P_\tau)\bigr]$;
\item[(ii)] $\E\bigl[\,\bigl|S^{\A}_\tau\bigr|\,\bigr]
  \;\le\; \kappa_\Lambda\,\E\bigl[\Phi(P_\tau)\bigr]/f$;
\item[(iii)] $\mathrm{OPT}_f(P_\tau)
  \;\le\; \mathrm{dist}^{\A}_\tau + f\bigl|S^{\A}_\tau\bigr|$ pathwise.
\end{enumerate}
Instantiating $\Phi(x) = \beta\,\mathrm{OPT}_f(x)$ with
$\beta = O(\log n/\log\log n)$ (over metrics by \citet{fotakis2008competitive} and over semi-metrics by \Cref{thm:semiratio}), and using
$\mathrm{OPT}_f \le fk + \mathrm{OPT}_k$ (the $k$-th term of
\eqref{fl2med:eq}),
\[
\E\bigl[\mathrm{dist}^{\A}_\tau\bigr]
  \le \kappa_\Lambda\beta\bigl(fk + \E[\mathrm{OPT}_k(P_\tau)]\bigr),
\qquad
\E\bigl[\,\bigl|S^{\A}_\tau\bigr|\,\bigr]
  \le \kappa_\Lambda\beta\Bigl(k + \E[\mathrm{OPT}_k(P_\tau)]/f\Bigr).
\]
Moreover, combining (iii) with \Cref{cor:transfer}: the sketch's total
cost is a pathwise upper estimate of the facility-location optimum, and
a $2\kappa_\Lambda\beta$-factor approximation of it in expectation,
\begin{equation}\label{eq:flestimate}
\mathrm{OPT}_f(P_\tau)
\;\le\; \mathrm{dist}^{\A}_\tau + f\bigl|S^{\A}_\tau\bigr|
\;=\; C^{\A}_\tau
\quad\text{pathwise},
\qquad
\E\bigl[C^{\A}_\tau\bigr]
\;\le\; 2\kappa_\Lambda\beta\,\E\bigl[\mathrm{OPT}_f(P_\tau)\bigr].
\end{equation}
\end{lem}

\begin{proof}
By \eqref{eq:distisQ}, $\mathrm{dist}^{\A}_\tau = Q^{\A}_\tau$, and
$\bigl|S^{\A}_\tau\bigr| = K^{\A}_\tau$. \Cref{thm:upper} gives
$\E[K^{\A}_\tau] \le \kappa_\Lambda\,\E[K^{\OB}_\tau]$, and
$f\,\E[K^{\OB}_\tau] \le \E[C^{\OB}_\tau] \le \E[\Phi(P_\tau)]$, the last
step by conditioning on the sequence, under which the replay is the
sketch on a fixed input; this establishes (ii). For (i),
$\E[Q^{\A}_\tau] \le f\,\E[K^{\A}_\tau]$ by \eqref{eq:routingdom}, and
the same argument applies. 
For (iii), take $F = S^{\A}_\tau$ in \eqref{fl2med:eq}: since centers are never removed, each point's arrival-time assignment is a member of $F$, and the distance to it can only exceed the distance to the closest center of $F$. Hence
$\mathrm{cost}(P_\tau, S^{\A}_\tau) \le \mathrm{dist}^{\A}_\tau$,
and adding the opening cost $f\bigl|S^{\A}_\tau\bigr|$ gives the
claim. Finally,
$\mathrm{dist}^{\A}_\tau + f\bigl|S^{\A}_\tau\bigr|
= Q^{\A}_\tau + f K^{\A}_\tau = C^{\A}_\tau$, and the expectation
bound in \eqref{eq:flestimate} is \Cref{cor:transfer} with
$\Phi = \beta\,\mathrm{OPT}_f$.
\end{proof}

\subsection{Multi-price Summaries and Online $k$-clustering}
\label{sec:multiprice}

\begin{algorithm2e}[t]
\caption{Multi-price summary}
\label{alg:multiprice}
{\small
\KwIn{
  $k$; bounds $n, \Delta$; an offline $\rho$-approximation $\mathcal{B}$ for weighted $k$-clustering; a capacity $B \in \mathbb{N} \cup \{\infty\}$.
}
\tcp{Note on capacity: $B = \infty$ is an analysis device (the coupling shadow of \Cref{lem:capping}); finite $B$ enforces memory $O(B \log(n\Delta))$ pathwise, with guarantees transferred by \Cref{lem:capping}.}
\tcc{Initialization}
$P \gets \lceil \log_2(n\Delta) \rceil$\;
\ForEach{level $j = 0, \dots, P$}{
  $f_j \gets 2^j d_{\min}/\bigl(k(1+\lceil\ln n\rceil)\bigr)$;\quad $S^{(j)} \gets \emptyset$;\quad
  $w^{(j)} \gets \text{empty}$;\quad $\mathrm{dist}^{(j)} \gets 0$\;
}
\tcc{Online Updates}
\ForEach{arriving point $x$}{
  \ForEach{level $j = 0, \dots, P$
    \tcp*{one step of \Cref{alg:meyerson} at price $f_j$, with weights}}{
    $d \gets d(x, S^{(j)})$; draw a fresh coin
    $U^{(j)} \sim \mathrm{Unif}[0,1]$\;
    \eIf{$\bigl|S^{(j)}\bigr| < B$ \textbf{and}
      $U^{(j)} \le \min\{1,\ d/f_j\}$}{
      open a center: $S^{(j)} \gets S^{(j)} \cup \{x\}$;\quad
      $w^{(j)}(x) \gets 1$\;
    }{
      $y \gets$ nearest center to $x$ in $S^{(j)}$
      \tcp*{copy at capacity; it routes but never opens}
      $w^{(j)}(y) \gets w^{(j)}(y) + 1$;\quad
      $\mathrm{dist}^{(j)} \gets \mathrm{dist}^{(j)} + d$\;
    }
  }
  publish $\widehat P^{(j)} = (S^{(j)}, w^{(j)})$ and
  $\mathrm{dist}^{(j)}$ for all $j$\;
}
\BlankLine
\tcc{$k$-clustering reporting (at any time on demand)}
\ForEach{level $j = 0, \dots, P$}{
  $C^{(j)} \gets \mathcal{B}\bigl(\widehat P^{(j)}\bigr)$;\quad
  $\mathrm{UB}^{(j)} \gets
    \alpha\bigl(\mathrm{cost}\bigl(\widehat P^{(j)}, C^{(j)}\bigr)
    + \mathrm{dist}^{(j)}\bigr)$\;
}
report $C^{(j^*)}$ with certificate $\mathrm{UB}^{(j^*)}$, where
$j^* \in \arg\min_j \mathrm{UB}^{(j)}$\;
}
\end{algorithm2e}

As evident from \eqref{med2fl:eq}, using a fixed-price  sketch to obtain an approximate $k$-clustering requires a price $f\approx \mathrm{OPT}_k(P_t)/k$. A high price loses granularity and a low price may inflate the summary size. Since $\mathrm{OPT}_k(P_t)$ is increasing, we need to adjust the price scale accordingly. Online clustering algorithms therefore use multiple scales or adjust the scale in phases~\citep{CharikarOP03,BMORST11,ShindlerWM11,LattanziV17}.

\Cref{alg:multiprice} describes the multi-price sketch and the approximate $k$-clustering reporting.
Fix $k$ and a
bound $n$ on the stream length. 
The sketch consists of $P+1$ concurrent Meyerson online summary copies (\Cref{sec:Meyersonsummary}) with prices 
\[\left(f_j = \frac{\lambda_j}{\bigl(k(1+\lceil\ln n\rceil)\bigr)}\right)_{j\in\{0,\ldots,P\}}\; \text{ where } P = \lceil \log_2 (n\Delta) \rceil\, ;\, \lambda_j := 2^j d_{\min}\ .\]
Each copy $j$ is applied to the input stream and publishes its summary in the triplet form  $(C^{(j)},\mathrm{UB}^{(j)},\mathrm{dist}^{(j)})$ of centers, weights, and distortion. 
The sketch receives as an input a capacity bound $B$ on the number of open centers per copy. Copies that reach capacity freeze their set of open centers and route all new requests to the existing set.

At each time step, the reported $k$-clustering solution is computed from the published summaries: For each summary $j$, we apply the $k$-clustering algorithm $\mathcal{B}$ to the summary to obtain $C^{(j)}$. We then apply \cref{fact:distortion} to compute an upper bound $\mathrm{UB}^{(j)}$ on the $k$-clustering cost of $C^{(j)}$ on the full prefix $P_t$ from the distortion and the cost of $C^{(j)}$ on the summary. 
We report the centers $C^{(j^*(t))}$ of the summary $j^*(t)$ with the best upper bound and the certificate value $\mathrm{UB}^{(j^*(t))}$.

\begin{thm}[Online adaptive $k$-clustering]\label{thm:kmedianmain}
Let $\mathcal{B}$ be an offline $\rho$-approximation algorithm for
weighted $k$-clustering, let $\delta \in (0,1)$, and run
\Cref{alg:multiprice} with capacity
\[
B = \Theta\bigl(\kappa_\Delta\bigl(k\log n +
\log\tfrac{\log(n\Delta)}{\delta}\bigr)\bigr)
\]
 against an adaptive
generator that observes everything published. Then:
\begin{enumerate}
\item[(i)] \emph{(Space.)} 
The algorithm uses
$O\bigl(\kappa_\Delta\bigl(k\log n +
\log\tfrac{\log(n\Delta)}{\delta}\bigr)\log(n\Delta)\bigr)$
words.\footnote{A word stores a presented point, a distance or cost value, or an integer of $O(\log(n\Delta))$ bits.}
\item[(ii)] \emph{(Certificate.)} Deterministically, at every time $t$, the reported centers satisfy
$\mathrm{cost}\bigl(P_t,\, C^{(j^*(t))}_t\bigr)
\le \mathrm{UB}^{(j^*(t))}_t$.
\item[(iii)] \emph{(Approximation.)} With probability at least
$1 - \delta$, simultaneously for every time $t$,
\[
\mathrm{cost}\bigl(P_t,\, C^{(j^*(t))}_t\bigr)
\;\le\; O(\rho)\,\kappa_\Delta
\Bigl(1 + \tfrac{1}{k\log n} \log\tfrac{\log(n\Delta)}{\delta}\Bigr)
\cdot \max\bigl(\mathrm{OPT}_k(P_t),\, d_{\min}\bigr) .
\]
\item[(iv)] \emph{(Update Time.)} Each update is processed in worst-case time
linear in the space bound of (i), with one distance evaluation per
stored center and $O(1)$ further work per level. Reporting applies
$\mathcal{B}$ to $O(\log(n\Delta))$ weighted instances of at most $B$
points each.
\end{enumerate}
In particular, at $\delta = n^{-O(1)}$: with
$O\bigl(\kappa_\Delta\, k \log^2 n \bigr)$ words of memory,
 the reported $k$ centers are with high
probability an
$O\bigl(\rho \cdot \tfrac{\log\Delta}{\log\log\Delta}\bigr)$-approximation
of the current optimum, simultaneously at every prefix with $\mathrm{OPT}_k(P_t) > 0$.\footnote{Here we assume $\Delta = n^{O(1)}$, a standard assumption in streaming
clustering. This is without loss of generality as distances below
$f_j/n^2$ contribute $O(f_j/n)$ to any copy's cost and trigger
openings with probability $O(1/n)$ in total, so each copy's
effective aspect ratio is at most $n^{O(1)}$ regardless of $\Delta$}  All constants
depend only on $\alpha$.
\end{thm}

\subsection{Analysis}\label{sec:spine}
In this subsection we prove \cref{thm:kmedianmain}. 
To aid analysis, we add the step subscript to the notation. 
For level $j$ and time $t$,
$C^{(j)}_t = \mathcal{B}(\widehat P^{(j)}_t)$,
is the set of centers, $\mathrm{dist}^{(j)}_t = Q^{(j)}_t$ is the distortion, maintained as a running counter, and the computed upper bound is
\[
\mathrm{UB}^{(j)}_t \;=\; \alpha\Bigl(\mathrm{cost}\bigl(\widehat P^{(j)}_t,\, C^{(j)}_t\bigr) + \mathrm{dist}^{(j)}_t\Bigr) .
\]
The selected index is $j^*(t) \in \arg\min_j \mathrm{UB}^{(j)}_t$, and the reported clustering and certificate are $C^{(j^*(t))}_t$ and  $\mathrm{UB}^{(j^*(t))}_t$.

We establish the space, time, and certificate claims:
\begin{proof}[Proof of \Cref{thm:kmedianmain} (i): space; (iv): time]
The algorithm stores $B$ weighted centers and one counter per level. There are $P + 1 = O(\log(n\Delta))$ levels. Therefore, the algorithm requires $O(B
\log(n\Delta))$ words of space  in total. 
An update at level $j$ computes $d(x,S^{(j)})$ with one distance evaluation per center of $S^{(j)}$ and otherwise performs $O(1)$ work, one coin and $O(1)$ counter updates, identically for frozen and unfrozen copies. Since $|S^{(j)}| \le B$ at all times, an update costs one distance evaluation per stored center plus $O(1)$ per level, worst case and linear in the space bound. A reporting step applies, at each level, the algorithm $\mathcal{B}$ to the weighted instance $\widehat P^{(j)}$ of at most $B$ points and evaluates $\mathrm{cost}\bigl(\widehat P^{(j)}, C^{(j)}\bigr)$ with $k\,\bigl|S^{(j)}\bigr|$ distance evaluations.
\end{proof}

\begin{proof}[Proof of \Cref{thm:kmedianmain} (ii): certificate]
The upper bound $\mathrm{UB}^{(j)}_t$
holds pathwise, deterministically, at all times $t$ and for every copy
$j$, by \Cref{fact:distortion} with \eqref{eq:distisQ}; in particular
it holds for the minimizer $j^*(t)$. 

The only nuance is when a copy reaches capacity. 
A frozen copy does not open new centers but continues 
to \emph{route}: each subsequent point is assigned to its nearest frozen
center, incrementing the weights and the distortion counter. The frozen copy therefore remains a valid online summary with its tracked routing cost equal to the distortion. 
\end{proof}

It remains to establish claim (iii) (approximation). 
As an analytical device, 
we define the \emph{window} of level $j$ as the  prefix of the stream during which $\mathrm{OPT}_k(P_t) \le \lambda_j := 2^j d_{\min}$. 
(Since $\mathrm{OPT}_k(P_t)$ is nondecreasing in $t$, this is a prefix).
The endpoint of this window, $T_j$, is defined as:
\[
T_j \;=\; \sup\{\, t \,:\, \mathrm{OPT}_k(P_t) \le \lambda_j \,\}
\;\in\; \mathbb{N} \cup \{\infty\} .
\]

The analysis proceeds by first establishing per-level properties at time $T_j$ and then combining them. For each level $j$, we first establish the properties in a shadow system in which copy $j$ has no capacity bound (\cref{lem:perlevel}), and then transfer them to \Cref{alg:multiprice} by bounding the probability of a \emph{premature freeze}, the event that copy $j$ is at capacity by time $T_j$ (\cref{lem:capping}).

\paragraph{Shadow system for level $j$.}
The \emph{shadow system} of level $j$ differs from
\Cref{alg:multiprice} in that copy $j$'s capacity bound is removed; the copies $i \ne j$ run at the system's capacity $B$. The shadow
system uses the same randomness $(W, (U^{(i)})_i)$ as the deployed system, with the same generator reacting to the shadow system's published summaries. The \emph{shadow} of copy $j$ is copy $j$'s run in this system.

\begin{lem}[Per-level high-probability bounds, shadow system]\label{lem:perlevel}
Fix a level $j$, a capacity $B$, and $\delta \in (0,1)$, and run the
shadow system of level $j$ against any adaptive generator. With
probability $1 - \delta$,
\[
\bigl|S^{(j)}_{T_j}\bigr|
\;=\; O\bigl(\kappa_\Delta\bigl(k\log n + \log\tfrac1\delta\bigr)\bigr)
\qquad\text{and}\qquad
\mathrm{dist}^{(j)}_{T_j}
\;=\; O(\kappa_\Delta)
\Bigl(\lambda_j + \frac{\lambda_j}{k\log n}\log\frac1\delta\Bigr) .
\]
\end{lem}

The proof applies the high-probability comparison of \Cref{thm:upper_whp,cor:upper_whp_cost} to relate the shadow to its replay, and the fixed-sequence bounds of \Cref{thm:semiratio_whp} to bound the replay in terms of the clustering optimum. To apply \Cref{thm:upper_whp,cor:upper_whp_cost} to the shadow, we verify that its environment, the generator operating in the shadow system, observing its full transcript, is a one-sided generator (\cref{def:onesided}):

\begin{claim}[Shadow environment]\label{claim:environment}
Fix a level $j$. The environment of the shadow of copy $j$ consists of the generator together with the capped copies $i \ne j$, whose published summaries it observes. This environment is a one-sided generator (\Cref{def:onesided}) with private randomness $W_{\mathrm{eff}} = (W, (U^{(i)})_{i \ne j})$ and coins $U^{(j)}$.
\end{claim}
\begin{proof}
The tuple $W_{\mathrm{eff}}$ absorbs the generator's own randomness and the internal coins of the other copies. Since the coin sequences of distinct copies are mutually independent, and independent of $W$ and of copy $j$'s replay coins (which are fresh, introduced only for the analysis), the pair $(W_{\mathrm{eff}}, U^{(j)})$ satisfies the independence required of $(W, U)$ in \Cref{def:onesided}. For measurability: copy $i$'s state at time $t-1$ is a function of the stream prefix and $U^{(i)}_{\le t-1}$, and the capacity gate $|S^{(i)}| < B$ is a function of that state, so capping does not affect the induction. Copy $j$'s own published summary is reconstructible from its transcript $(x_s, o^{(j)}_s)_{s < t}$, and the stream prefix is part of that transcript. Hence everything the generator observes, and therefore its next point, is a measurable function of $\bigl(W_{\mathrm{eff}};\ (x_s,o^{(j)}_s)_{s < t}\bigr)$.
\end{proof}

\begin{proof}[Proof of \cref{lem:perlevel}]
By \Cref{claim:environment}, copy $j$ against its replay is an
instance of the model, so \Cref{thm:upper_whp} and
\Cref{cor:upper_whp_cost} apply to it. 
Their events are uniform in
time, so the bounds hold at the random time $T_j$: except with
probability $O(\delta)$,
\begin{equation} \label{eq:a2rbound}
K^{(j)}_{T_j} \;=\; O(\kappa_\Delta)\, K^{\OB,(j)}_{T_j}
+ O\Bigl(\kappa_\Delta \log\frac1\delta\Bigr),
\qquad
C^{(j)}_{T_j} \;=\; O(\kappa_\Delta)\, C^{\OB,(j)}_{T_j}
+ O\Bigl(\kappa_\Delta\,\frac{\lambda_j}{k\log n}\log\frac1\delta\Bigr) .
\end{equation}

For the replay side, condition on the stream. The stream is a measurable function of $(W_{\mathrm{eff}}, U^{(j)})$: inductively, copy $j$'s transcript up to time $t-1$ is determined by the stream prefix and $U^{(j)}_{\le t-1}$, and the generator's next point is determined by $W_{\mathrm{eff}}$ and that transcript. The replay coins $U'^{(j)}$ are independent of this pair (\Cref{claim:environment}). Conditionally on the stream, therefore, $T_j$ is a deterministic time and the replay is the sketch with i.i.d.\ fresh coins on the fixed sequence $P_{T_j}$ of at most $n$ points.

The high-probability fixed-sequence bounds of \Cref{thm:semiratio_whp} at price $f_j$, with $\mathrm{OPT}_k(P_{T_j}) \le \lambda_j$ and $\mathrm{OPT}_k(P_{T_j})/f_j \le k(1+\lceil\ln n\rceil)$, bound the replay's count and routing cost; since $C = fK + Q$ and $f_j\bigl(k\log n + \log\tfrac1\delta\bigr) = O\bigl(\lambda_j + \tfrac{\lambda_j}{k\log n}\log\tfrac1\delta\bigr)$, they give, conditionally on every stream realization, except with probability $\delta$ over the replay coins,
\begin{equation} \label{eq:replaybound}
K^{\OB,(j)}_{T_j} \;=\; O\Bigl(k\log n + \log\frac1\delta\Bigr),
\qquad
C^{\OB,(j)}_{T_j} \;=\; O\Bigl(\lambda_j
+ \frac{\lambda_j}{k\log n}\log\frac1\delta\Bigr) .
\end{equation} 

We now combine \eqref{eq:a2rbound} and \eqref{eq:replaybound}. There are two failure events. That of \eqref{eq:a2rbound} has probability $O(\delta)$ over the joint randomness $(W_{\mathrm{eff}}, U^{(j)}, U'^{(j)})$. That of \eqref{eq:replaybound} has conditional probability at most $\delta$ given each stream realization; since this holds for every realization, its unconditional probability is at most $\delta$. By a union bound, both displays hold simultaneously except with probability $O(\delta)$. Substituting \eqref{eq:replaybound} into \eqref{eq:a2rbound} and identifying $|S^{(j)}_{T_j}| = K^{(j)}_{T_j}$ and $\mathrm{dist}^{(j)}_{T_j} = Q^{(j)}_{T_j} \le C^{(j)}_{T_j}$ yields both bounds of the lemma.
\end{proof}

\begin{lem}[Capacity transfer]\label{lem:capping}
There is a constant $C_0$ such that for
$B \ge C_0\,\kappa_\Delta\bigl(k\log n + \log\tfrac1\delta\bigr)$,
for each level $j$, except with probability $O(\delta)$: the
deployed and shadow systems of level $j$ coincide through $T_j$; in
particular, copy $j$ never reaches capacity within its window and
satisfies the bounds of \Cref{lem:perlevel} at $T_j$.
\end{lem}

\begin{proof}
The two systems differ only in copy $j$'s gate. If they agree through step $t-1$, the generator presents the same point, and every copy applies the same rule with the same coin; the opening that brings copy $j$'s count to $B$ passes the gate $|S^{(j)}| < B$. Hence the systems coincide through $\tau^{(j)}_B$, the first time copy $j$'s count reaches $B$, and can diverge only afterwards.

\Cref{lem:perlevel} applies to the shadow. On its event, the shadow's count at $T_j$ is below $B$ for $C_0$ large enough, so $\tau^{(j)}_B > T_j$, and the coincidence covers every step through $T_j$ --- including the step $T_j + 1$ that closes the window, so $T_j$ itself is common to the two systems. Copy $j$ therefore never reaches capacity within its window, coincides with its shadow there, and inherits both bounds at $T_j$. All of this fails only on the failure event of \Cref{lem:perlevel}; no conditioning and no further union bound is involved.
\end{proof}

We now combine the per-level bounds to conclude the proof. 
At every time $t$, some level's window contains $t$. Since $\mathrm{OPT}_k(P_t) \le n\, d_{\max} \le \lambda_P$, the least level $j(t)$ with $\lambda_{j(t)} \ge \mathrm{OPT}_k(P_t)$ exists, and $t \le T_{j(t)}$. By minimality, $\lambda_{j(t)} \le 2\max\bigl(\mathrm{OPT}_k(P_t),\, d_{\min}\bigr)$. Moreover, if $\mathrm{OPT}_k(P_t) > 0$ then more than $k$ distinct locations were presented, so some center of an optimal solution serves two of them, and the $\alpha$-approximate triangle inequality gives $\mathrm{OPT}_k(P_t) \ge d_{\min}/\alpha$. Hence $\lambda_{j(t)} \le 2\alpha\,\mathrm{OPT}_k(P_t)$ whenever $\mathrm{OPT}_k(P_t) > 0$. We call $j(t)$ the \emph{designated} level at time $t$.

\begin{proof}[Proof of \Cref{thm:kmedianmain}(iii)]
Apply \Cref{lem:capping} to each level at failure probability $\delta/(P+1)$; the stated capacity $B$ meets its requirement, since $\log((P+1)/\delta) = O(\log(\log(n\Delta)/\delta))$. A union bound over the $P+1 = O(\log(n\Delta))$ levels leaves one event of probability at least $1-\delta$. On this event, for every level $j$, copy $j$ never reaches capacity within its window and satisfies the bounds of \Cref{lem:perlevel} at $T_j$.

Fix any time $t$ and let $j = j(t)$ be its designated level, so that $t \le T_j$ and $\lambda_j \le 2\max\bigl(\mathrm{OPT}_k(P_t),\, d_{\min}\bigr)$. The certificate is a valid upper bound and $j^*(t)$ minimizes it, so
\[
\mathrm{cost}\bigl(P_t,\, C^{(j^*(t))}_t\bigr)
\;\le\; \mathrm{UB}^{(j)}_t
\;\le\; \alpha^2\rho\,\mathrm{OPT}_k(P_t)
+ (\alpha^2\rho + \alpha)\,\mathrm{dist}^{(j)}_{T_j} ,
\]
using \cref{fact:distortion}, the $\rho$-approximation of $\mathcal{B}$, and the monotonicity $\mathrm{dist}^{(j)}_t \le \mathrm{dist}^{(j)}_{T_j}$. Substituting the distortion bound of \Cref{lem:perlevel} at failure probability $\delta/(P+1)$, with $\lambda_j \le 2\alpha\,\mathrm{OPT}_k(P_t)$ and $\log(P+1) = O(\log\log(n\Delta))$, yields (iii).
\end{proof}

\section{Discussion and Open Problems} \label{sec:openproblems}

\paragraph{Adaptive-to-OPT ratio.}
In this work, we introduce the \textit{replay} benchmark, where the performance of a randomized algorithm on an adaptive input sequence is compared against the performance of a fresh instance of the algorithm with independent randomness (called the \textit{replay}) on the same input sequence. This benchmark is designed to pinpoint the cost of adaptivity in a way that decouples it from the hardness of the input sequence itself. The classical adaptive adversaries of competitive analysis \citep{BBKTW94}, and the competitive bounds of \citet{fotakis2008competitive} and \citet{Lang:SODA2018} for the Meyerson sketch, compare against the offline optimal solution. As a byproduct of our analysis, 
\Cref{cor:transfer} allows us to compose the adaptivity gap of the Meyerson sketch with the
per-sequence guarantee of \citet{fotakis2008competitive}, yielding an upper bound of
\begin{equation} \label{acompetitive:eq}
O\left(\frac{\log\Lambda}{\log\log\Lambda}\cdot \frac{\log n}{\log\log n}\right)=O\left(\frac{\log^2 n}{(\log\log n)^2} \right)
\end{equation} on the adaptive-to-OPT ratio. A natural question is whether this composed adaptive-to-OPT ratio is tight. In particular, note that the adaptive-to-OPT ratio is at least the \emph{maximum} of the two
factors: The generator can present a hard fixed sequence, or run
the attack of \Cref{thm:anticert}, whose replay is $O(1)$-competitive
on its own stream. We leave the gap between them open:
\begin{openprob}\label{op:adaptive2optratio}
Quantify the adaptive-to-OPT ratio of the Meyerson sketch.
\end{openprob}

\paragraph{The replay benchmark beyond Meyerson.}
Our definitions of adaptivity ratio and adaptivity gaps in \Cref{sec:model} generally apply for any randomized online optimization algorithm whose internal
state is released to the adversary.

For instance, in the \textit{resettable streaming} model, the sketch maintains an approximate statistic 
under insertions of (copies of) elements and ``reset'' operations \footnote{A reset operation on an element $x$ sets the count to $0$, i.e. $c_x \gets 0$.}. Specifically, consider the following simple sketch for counting distinct elements, by maintaining a Bernoulli sample: for each element $x$ inserted into the stream of length $T$, the element will be resampled with a fixed probability $p$, even if $x$ is already present in the sample
\citep{GemullaLH:VLDB2006,GemullaLH:PODS2007,CohenCD:SIGMETRICS2012}. If $x$ is sampled, then set $S \gets S \cup \{x\}$. Otherwise, if a reset operation occurs for $x$ and $x \in S$, then simply remove $x$ from $S$; if $x$ was not in the sample to begin with, then the algorithm does nothing and moves on to process the next stream update. 
Then, to estimate the number of distinct elements under insertions and resets, the sketch continuously reports the unbiased estimator $|S|/p$. The behavior
of this classical Bernoulli sample-based sketch under adaptive inputs was studied in
\citep{CohenGNS:ICML2026}, which showed that simple adaptive attacks force extreme gaps: deleting each key that is revealed to have been sampled empties the adaptive sample while the replay retains
$\Theta(pT)$ keys, an unbounded deflation gap; deleting the unsampled keys instead inflates by $\Theta(1/p)$; and even on insertion-only
streams, re-inserting revealed keys deflates by $\Theta(1/p)$ through
the resampling semantics. \citet{CohenGNS:ICML2026} proposed a
robustified sketch that protects the random coins (i.e. the sampled elements $S$) with tools from differential
privacy.

At the opposite extreme, a different, insertion-only sampler with a single coin per distinct key (equivalently, persistent per-key
randomness, as in Bernoulli sampling by hashing or MinHash sketches~\citep{FlajoletMartin85,ECohen6f})
has adaptivity gap exactly $1$: for every stopping time, the expected
sample size is $p$ times the expected number of distinct keys in both
runs, since each key contributes one coin, at its first arrival, that
no later operation can re-flip. The replay lens thus places these summaries on one scale: fully robust (single-coin insertions-only samplers, gap
$1$), robust at a modest instance-dependent price (the Meyerson
sketch, $\Theta(\log\Delta/\log\log\Delta)$), and fragile without
active robustification (resampling and resettable samplers,
polynomial or unbounded gaps).
\begin{openprob}\label{op:replaytheory}
 Is there a
general theory, such as a structural parameter of the sampling rule, that recovers these adaptivity gaps?
\end{openprob}

\paragraph{More robust OFL sketches.}
Our adaptivity gap bounds are specific to the (randomized) Meyerson sketch. Deterministic sketches have adaptivity ratio identically $1$. However, even under adaptive inputs, there is no known deterministic  OFL sketch with a better competitive ratio than Meyerson \cref{acompetitive:eq}, with $\mathrm{poly}(k, \log n)$ space, where $k$ is the number of open facilities.
In particular, the deterministic online facility location algorithm of \citet{fotakis2008competitive}, which attains the optimal $\Theta(\log n/\log\log n)$ competitive ratio, maintains space linear in $n$.
\begin{openprob}\label{op:bettersketch}
 Is there a more robust OFL sketch, with $\Theta(\log n/\log\log n)$ competitive ratio, adaptivity gaps $o(\log\Delta/\log\log\Delta)$ that uses 
 $\mathrm{poly}(k, \log n)$ space?
\end{openprob}

\paragraph{Routing adaptivity gaps.}
The known bounds on the routing ratio, all of them attacks, are summarized below. 
No upper bound on the routing ratio is known.
\begin{center}
\begin{tabular}{lll}
routing ratio & inflation & deflation \\
\hline
one-sided & $\Om(\Lambda) = \Om(\Delta)$ & $\Om(\log\Delta/\log\log\Delta)$ \\
 & {\small synthetic finite metric (\Cref{prop:onesidedstarve})} & {\small(\Cref{prop:reverse})} \\[3pt]
two-sided & $2^{\Omega(D)}$ at $\Delta = O(1)$ & $2^{\Omega(D)}$ at $\Delta = O(1)$ \\
 & {\small Euclidean $\mathbb{R}^{D+1}$ (\Cref{prop:starve})} & {\small(\Cref{cor:starveflip})} \\
\end{tabular}
\end{center}
The natural open questions are the gaps in this table:
\begin{openprob}\label{op:routing}
\begin{itemize}
\item \emph{One-sided lower bounds:} can a one-sided generator force an unbounded routing ratio at constant aspect ratio, matching \Cref{prop:starve}, or deflation $\Om(\Delta)$, matching the inflation of \Cref{prop:onesidedstarve}? The obstruction to both is the same: certifying the replay's \emph{absence} cheaply from one-sided observations.
\item \emph{Fixed dimension:} can the routing ratio be made unbounded at constant aspect ratio in fixed ambient dimension, say on the real line? The ratios of our two-sided constructions grow with the dimension.
\item \emph{Upper bounds:} does the aspect ratio together with the dimension bound the two-sided routing ratio, and does the aspect ratio alone bound the one-sided one?
\end{itemize}
\end{openprob}

\paragraph{Improved approximation factor in online $k$-clustering.}
A limitation of \Cref{thm:kmedianmain} is that its error bound is the distortion $\mathrm{dist}_t$ (\Cref{fact:distortion}), and the distortion admits no size--accuracy tradeoff: at price $f$ it is $O(\alpha\,\mathrm{OPT}_k + f\,k\log n)$, keeping more centers by lowering the price shrinks only the second term, and routing cost $\Omega(\mathrm{OPT}_k)$ is unavoidable under any assignment of the points to few representatives.\footnote{On the uniform metric on $m$ points, with all pairwise distances $r$, any assignment to $B \le m/2$ representatives routes $m - B \ge m/2$ points at distance $r$ each, while $\mathrm{OPT}_k = (m-k)r \le mr$ for every $k$. The weighted summary may still \emph{estimate} costs well on such benign instances; the point is that the architecture certifies and, adaptively, controls the error only through the distortion.} The distortion caps the certified accuracy at a constant factor already under oblivious inputs, and adaptively $\kappa_\Delta$ multiplies it (\Cref{lem:perlevel}; the inflation is realized by the routing claim of \Cref{thm:anticert}), giving the factor $O(\rho\,\kappa_\Delta)$. In sampling-based coresets, by contrast, space buys accuracy: the estimation error is $\eps\cdot\mathrm{OPT}_k$ with $\eps$ shrinking in the sample size, and no $\rho$ appears until an offline solver is applied to the coreset. Merge-and-reduce maintains such coresets, but not as an online summary. We \emph{suspect} that an online summary can be a $(1\pm\eps)$-coreset under adaptive inputs, by adding a sampling layer with respect to the online centers: points sampled with probabilities determined by their distances to the current center set, the Meyerson centers serving as an online bicriteria solution.
The layer size is larger by a factor depending on the ambient dimension or on $\log|\Mm|$. Its adaptive robustness requires a separate argument, as the sampling probabilities are state dependent. 
\begin{openprob}\label{op:samplinglayer}
Is there an online summary that is, under adaptive inputs, a $(1\pm\eps)$-coreset at all times, with space comparable to \Cref{thm:kmedianmain}?
\end{openprob}

\paragraph*{AI Disclosure.}
We used AI, mostly Claude (Anthropic) and also Chat GPT (OpenAI) and Gemini (Google) to assist in preparing this paper. For producing initial drafts of some arguments, for drafting and revising the exposition, and for formalizing and auditing proofs. The tools were used interactively and materially affected the results and writing throughout. The authors verified the correctness and originality of all content, and take full responsibility for all claims.

\section*{Acknowledgments}

\paragraph{Edith Cohen:}  Partially supported by the Israel Science Foundation (grant 1156/23). 

\paragraph{Elena Gribelyuk:} Partially supported by Huacheng Yu's NSF CAREER award CCF-2339942.

\paragraph{Uri Stemmer:} Partially supported by the Israel Science Foundation (grant 1419/24), and the Blavatnik Research Foundation.

\bibliographystyle{plainnat}
\bibliography{refs,adaptiverefs}

\appendix

\section{Extensions for the Meyerson Sketch on oblivious inputs}
\label{app:semiratio}
We present two self-contained extensions of prior results for the analysis of the Meyerson Sketch on oblivious (fixed in advance) inputs: We extend the competitive ratio bound of \citet{fotakis2008competitive} to semi-metrics (used to instantiate \cref{cor:transfer}) and extent the parameter range for the high probability bounds of \citet{BMORST11} (used in the proof  of \cref{thm:kmedianmain}).

\subsection{Competitiveness over Semi-Metrics}
The Meyerson sketch's optimal oblivious
competitive ratio was originally established for metrics
\citep{meyerson2001online,fotakis2008competitive}. In this section, we verify that it extends to semi-metrics which satisfy the $\alpha$-approximate triangle inequality (See \cref{def:semimetric}).

\begin{thm}[Competitive ratio over semi-metrics]\label{thm:semiratio}
Let $d$ satisfy the $\alpha$-approximate triangle inequality,
$\alpha \ge 1$. For every fixed sequence $x$ of $n$ points and every
price $f > 0$, the Meyerson sketch satisfies, for every $\mu \ge 2$,
\[
\E\bigl[C\bigr] \;\le\;
\Bigl(2\log_\mu n + 2\alpha\mu + 6\alpha + 4\Bigr)\cdot
\mathrm{OPT}_f(x) .
\]
Choosing $\mu = 2 + \lceil \ln n/\ln\ln n\rceil$ (for $n \ge 3$) yields
$\E[C] = O\bigl(\alpha \cdot \tfrac{\log n}{\log\log n}\bigr)\cdot
\mathrm{OPT}_f(x)$.
\end{thm}

We use the following two simple bounds:
\begin{fact}[Step Cost]\label{fact:stepcost}
   A step presented at current distance
$\delta = d(x_t, S_{t-1})$ has expected cost
\begin{equation}\label{eq:stepcost}
f\min\{1, \delta/f\} + \delta\bigl(1 - \min\{1, \delta/f\}\bigr)
\;\le\; 2\min\{\delta, f\} \;\le\; 2\delta .
\end{equation} 
\end{fact}

\begin{claim}[Waiting cost]\label{clm:wait}
Fix any subset $G$ of the presented points and let $\tau$ be the first
step at which the sketch opens a facility at a point of $G$. The
expected routing paid by points of $G$ presented before $\tau$ is at
most $f$.
\end{claim}

\begin{proof}
List the points of $G$ in arrival order with current distances
$\delta_1, \delta_2, \dots$ and opening probabilities
$q_l = \min\{1, \delta_l/f\}$; each $\delta_l$ (hence $q_l$) is a
function of the past, and the $l$-th coin is fresh, so
$\E[\mathbf{1}\{\tau > l\}] = \E[\mathbf{1}\{\tau > l-1\}(1 - q_l)]$.
Since $\delta_l(1 - q_l) \le f q_l$ pointwise (for $\delta_l \le f$
this is $f(q_l - q_l^2) \le f q_l$; for $\delta_l > f$ the left side
is $0$), the expected pre-$\tau$ routing telescopes:
\[
\E\Bigl[\sum_{l < \tau} \delta_l\Bigr]
= \sum_l \E\bigl[\mathbf{1}\{\tau > l-1\}(1-q_l)\,\delta_l\bigr]
\le f \sum_l \E\bigl[\mathbf{1}\{\tau > l-1\} - \mathbf{1}\{\tau > l\}\bigr]
\le f . \qedhere
\]
\end{proof}

\begin{proofof}{\Cref{thm:semiratio}}
Fix an optimal facility-location solution: facilities $F^*$ and, for
each $c^*_i \in F^*$, the cluster $C^*_i$ of points assigned to it,
with $d^*_p = d(p, c^*_i)$ for $p \in C^*_i$,
$A^*_i = \sum_{p \in C^*_i} d^*_p$, and average radius
$a^*_i = A^*_i/|C^*_i|$, so that
$\mathrm{OPT}_f(x) = f|F^*| + \sum_i A^*_i$.

Decompose each cluster into the \emph{core}
$S_0 = \{p \in C^*_i : d^*_p \le a^*_i\}$ and the \emph{rings}
$S_j = \{p : \mu^{j-1} a^*_i < d^*_p \le \mu^j a^*_i\}$ for $j \ge 1$.
Since a point of ring $j$ contributes more than $\mu^{j-1} a^*_i$ to
$A^*_i = |C^*_i|\,a^*_i$, ring $j$ is nonempty only if
$\mu^{j-1} < |C^*_i| \le n$; hence at most $2 + \log_\mu n$ of the sets
$S_0, S_1, \dots$ are nonempty. (If $a^*_i = 0$ then every $d^*_p = 0$
and all points lie in the core.)

Charge each nonempty set $S_j$ as follows.
By by \Cref{clm:wait}, we know that before the sketch first
opens a facility at a point of $S_j$, the routing paid by points of
$S_j$ is at most $f$ in expectation, and the opening
that ends the wait costs at most $f$; together, this corresponds to routing cost at most $2f$ per
nonempty set, hence at most $2f(2 + \log_\mu n)$ per cluster.

After a facility is open at some $q \in S_j$, every subsequent
$p \in S_j$ has current distance
$\delta \le d(p, q) \le \alpha\bigl(d^*_p + d^*_q\bigr)$. Note that this is the only application
of the approximate triangle inequality. For a ring $j \ge 1$:
$d^*_q \le \mu^j a^*_i < \mu\, d^*_p$, so
$\delta \le \alpha(1 + \mu)\, d^*_p$ and, by \eqref{eq:stepcost}, the
step costs at most $2\alpha(1+\mu)\, d^*_p$ in expectation. For the
core: $d^*_q \le a^*_i$, so the step costs at most
$2\alpha(d^*_p + a^*_i)$, and summing over the core,
$2\alpha\sum_{p \in S_0}(d^*_p + a^*_i) \le 4\alpha A^*_i$.
Summing the ring contributions,
$2\alpha(1+\mu)\sum_{p \notin S_0} d^*_p \le 2\alpha(1+\mu) A^*_i$.

Adding up the total cost over all clusters, we obtain
\[
\E[C] \;\le\; 2\bigl(2 + \log_\mu n\bigr)\, f|F^*|
\;+\; \bigl(2\alpha(1+\mu) + 4\alpha\bigr) \sum_i A^*_i
\;\le\; \bigl(2\log_\mu n + 2\alpha\mu + 6\alpha + 4\bigr)\,
\mathrm{OPT}_f(x) .
\]
With $\mu = 2 + \lceil \ln n/\ln\ln n \rceil$ we have
$\alpha\mu = O(\alpha \log n/\log\log n)$ and
$\log_\mu n = \ln n/\ln\mu = O(\log n / \log\log n)$, proving the rate.
\end{proofof}

\subsection{High-probability bounds}

\citet[Theorem~3.1]{BMORST11} established the high probability bounds below for semi-metrics but at the fixed price
$f = L/\bigl(k(1+\log_2 n)\bigr)$ with $L \le \mathrm{OPT}_k(P)$ and
failure probability $1/n$; see also \citet{ShindlerWM11}. We give a
self-contained price- and $\delta$-parametrized form, using the decomposition of \Cref{thm:semiratio} into core and ring sets.

\begin{thm}[High-probability fixed-sequence bounds]\label{thm:semiratio_whp}
Let $(\Mm,d)$ satisfy \cref{def:semimetric} with
$\alpha \ge 1$. Run the Meyerson sketch at price $f > 0$ on a
fixed sequence $P\subset \Mm$ of at most $n$ points, and fix $k \ge 1$ and
$\delta \in (0,1)$. With probability at least $1 - \delta$,
simultaneously,
\[
K \;\le\; k\bigl(2 + \lgp n\bigr)
+ \frac{15\alpha}{2}\,\frac{\mathrm{OPT}_k(P)}{f}
+ \frac{3}{2}\ln\frac{2}{\delta}
\qquad\text{and}\qquad
Q \;\le\; 5\alpha\,\mathrm{OPT}_k(P)
+ 2f\Bigl(k\bigl(2 + \lgp n\bigr) + \ln\frac{2}{\delta}\Bigr) ,
\]
with constants depending only on $\alpha$.
\end{thm}

\begin{claim}[Waiting-cost tail]\label{clm:waittail}
Partition the presented points into groups, of which at most $\Gamma$
are nonempty, and let $Q^{\mathrm{wait}}$ be the total routing paid by
points that arrive before the first opening at a point of their own
group. Then for every $m \ge 0$,
\[
\Pr\bigl[\,Q^{\mathrm{wait}} > 2f\,(\Gamma + m)\,\bigr] \;\le\; e^{-m} .
\]
\end{claim}

\begin{proof}
Let $Q^{\mathrm{wait}}_t$ and $O_t$ be, after step $t$, the
accumulated waiting routing and the number of groups holding a
facility, and let $M_t = \exp\bigl(Q^{\mathrm{wait}}_t/(2f) -
O_t\bigr)$, with $M_0 = 1$. A step in a group that already holds a
facility changes neither quantity. A step in a facility-less group at
current distance $\delta$ opens with probability
$q = \min\{1, \delta/f\}$, multiplying $M$ by $e^{-1}$, and otherwise
routes at cost $\delta = qf$ (the routing branch has probability zero
unless $\delta < f$), multiplying $M$ by $e^{q/2}$. The coin is fresh,
so the conditional expected multiplier is
\[
(1-q)\,e^{q/2} + q\,e^{-1}
\;\le\; e^{-q/2} + \tfrac38\,q
\;\le\; 1 - \tfrac{q}{2} + \tfrac{q^2}{8} + \tfrac38\,q
\;\le\; 1
\qquad\text{for } q \in [0,1],
\]
using $1-q \le e^{-q}$, then $e^{-x} \le 1 - x + x^2/2$ at $x = q/2$
together with $e^{-1} < \tfrac38$, and finally $q^2 \le q$. 
Hence $(M_t)$ is a supermartingale, so at the last step $N$ of the
(finite) sequence, $\E[M_N] \le M_0 = 1$. Since $O_N \le \Gamma$ on
every path, the event $Q^{\mathrm{wait}} > 2f(\Gamma + m)$ forces
$M_N \ge e^{Q^{\mathrm{wait}}/(2f) - \Gamma} > e^{m}$; applying
Markov's inequality to the nonnegative variable $M_N$ bounds its
probability by $\E[M_N]\, e^{-m} \le e^{-m}$.
\end{proof}

\begin{claim}[Late-opening tail]\label{clm:opentail}
Let $(o_t)$ be $\{0,1\}$-valued random variables adapted to a
filtration $(\mathcal{H}_t)$ such that
$p_t = \Pr[\,o_t = 1 \mid \mathcal{H}_{t-1}\,]$ is
$\mathcal{H}_{t-1}$-measurable and $\sum_t p_t \le m$ on every path.
Then for every $\delta \in (0,1)$,
\[
\Pr\Bigl[\,\sum_t o_t \;\ge\; \tfrac32\Bigl(m + \ln\tfrac1\delta\Bigr)\Bigr]
\;\le\; \delta .
\]
\end{claim}

\begin{proof}
Since $\E\bigl[2^{o_t} \mid \mathcal{H}_{t-1}\bigr] = 1 + p_t \le
e^{p_t}$, the process $2^{N_t} e^{-A_t}$, where $N_t = \sum_{s \le t}
o_s$ and $A_t = \sum_{s \le t} p_s$, is a supermartingale with initial
value $1$. With $A \le m$ pathwise, the event $N \ge y$ forces
$2^{N} e^{-A} \ge 2^{y} e^{-m}$, so Markov's inequality gives
$\Pr[N \ge y] \le 2^{-y} e^{m}$. At
$y = \bigl(m + \ln\tfrac1\delta\bigr)/\ln 2 \le
\tfrac32\bigl(m + \ln\tfrac1\delta\bigr)$ the bound is $\delta$.
\end{proof}

Finally, we use the lemmas developed above, which apply to semi-metrics with $\alpha$-approximate triangle inequality, to prove \Cref{thm:semiratio_whp}.

\begin{proofof}{\Cref{thm:semiratio_whp}}
Fix an optimal $k$-clustering: centers $c^*_i$, clusters $C^*_i$,
optimal distances $d^*_p$, cluster costs $A^*_i$, and average radii
$a^*_i$, as in the proof of \Cref{thm:semiratio}, so that
$\mathrm{OPT}_k(P) = \sum_i A^*_i$. Decompose each cluster into the
core $S_0 = \{p : d^*_p \le a^*_i\}$ and the rings
$S_j = \{p : 2^{j-1} a^*_i < d^*_p \le 2^{j} a^*_i\}$ for $j \ge 1$
(the decomposition of \Cref{thm:semiratio} at $\mu = 2$; its
quantitative optimization of $\mu$ is not needed here). By the
counting argument there, at most $2 + \lgp n$ groups per cluster are
nonempty, hence at most $\Gamma = k(2 + \lgp n)$ in total, and every
presented point belongs to exactly one group.

Write the opening indicators as $o_t = \ind{\{U_t \le q_t\}}$ with
$q_t = \min\{1,\, \delta_t/f\}$, where $\delta_t$ is the current
distance, and let $\zeta_t \in \{0,1\}$ indicate that the group of
$x_t$ already holds a facility; $\zeta_t$ is a function of
$U_{<t}$, hence predictable. Decompose
\[
K \;=\; \sum_t (1-\zeta_t)\, o_t \;+\; \sum_t \zeta_t\, o_t ,
\qquad
Q \;=\; \underbrace{\sum_t (1-\zeta_t)(1-o_t)\,\delta_t}_{Q^{\mathrm{wait}}}
\;+\; \underbrace{\sum_t \zeta_t (1-o_t)\,\delta_t}_{Q^{\mathrm{late}}} .
\]

\emph{Pathwise bounds.} Each group contributes at most one first
opening, so $\sum_t (1-\zeta_t)\, o_t \le \Gamma$ deterministically.
Once the group of a point $p$ holds a facility at some point $q$ of
the same group, the current distance satisfies $\delta_p \le d(p, q)
\le \alpha\bigl(d^*_p + d^*_q\bigr)$ --- the only use of the
(approximate) triangle inequality. For a ring point ($j \ge 1$):
$d^*_q \le 2^{j} a^*_i < 2\, d^*_p$, so $\delta_p \le 3\alpha\,
d^*_p$; for a core point, $\delta_p \le \alpha(d^*_p + a^*_i)$.
Summing within a cluster,
\[
Q^{\mathrm{late}}
\;\le\; \sum_i \Bigl(3\alpha \sum_{p \notin S_0} d^*_p
+ \alpha \sum_{p \in S_0} (d^*_p + a^*_i)\Bigr)
\;\le\; \sum_i \bigl(3\alpha A^*_i + 2\alpha A^*_i\bigr)
\;=\; 5\alpha\,\mathrm{OPT}_k(P)
\]
on every path, and, since $q_t \le \delta_t / f$, the conditional
means of the late openings satisfy $\sum_t \zeta_t q_t \le
5\alpha\,\mathrm{OPT}_k(P)/f$ on every path as well.

\emph{Concentration.} \Cref{clm:waittail} at $m = \ln\tfrac2\delta$
gives $Q^{\mathrm{wait}} \le 2f\bigl(\Gamma + \ln\tfrac2\delta\bigr)$
except with probability $\delta/2$. \Cref{clm:opentail}, applied to
the adapted indicators $(\zeta_t o_t)$ (whose conditional
probabilities are $\zeta_t q_t$, since $\zeta_t$ is predictable and
the coin $U_t$ is fresh) with $m = 5\alpha\,\mathrm{OPT}_k(P)/f$,
bounds the late openings by $\tfrac32\bigl(5\alpha\,
\mathrm{OPT}_k(P)/f + \ln\tfrac2\delta\bigr)$ except with probability
$\delta/2$. A union bound and $\Gamma \le k(2+\lgp n)$ complete the
proof.
\end{proofof}

\section{Routing Lower Bounds}\label{app:starvation}

This appendix establishes much stronger lower bounds on the routing adaptivity
ratio $\E[Q^{\A}_\tau]/\E[Q^{\OB}_\tau]$ (\emph{routing ratio} for
short) than \Cref{thm:anticert}:
The two-sided routing
\emph{inflation} and \emph{deflation} gaps are both unbounded already at constant aspect ratio.  The one-sided inflation gap is at least linear in the aspect ratio.

\begin{prop}[Routing Ratio; proofs in \Cref{app:starvation}]\label{prop:starvesummary}

\begin{itemize}
    \item [(i)] A \emph{two-sided} generator in $\mathbb{R}^{D+1}$ forces routing ratio $2^{\Omega(D)}$ at aspect ratio $\Delta \le 29$ (\Cref{prop:starve}). 
 \item [(ii)]
Symmetrically, a two-sided generator forces routing ratio
$2^{-\Omega(D)}$ at aspect ratio $\Delta \le 29$ (\Cref{cor:starveflip}). 
\item [(iii)] A deterministic
\emph{one-sided} generator in a finite metric forces ratio
$\Om(\Lambda) = \Om(\Delta)$ (\Cref{prop:onesidedstarve}).
\end{itemize}
In the constructions of (i) and (iii),
$\E[K^{\A}_\tau] \le \E[K^{\OB}_\tau]$ and
$\E[C^{\A}_\tau] \le \E[C^{\OB}_\tau]$: the generator relocates cost
between components rather than inflating it; the construction of
(ii) relocates in the opposite direction.
\end{prop}

The two-sided inflation and deflation generators are presented in \Cref{sec:twosidedrouting}. The one-sided inflation generator is presented in \Cref{app:onesided}.

The constructions work by manufacturing a location about which the two
runs' center sets disagree, then exploiting it. We call such a location a \emph{differential}. The constructions differ in what 
they utilize: The two-sided construction relies on a high Euclidean dimension. The one-sided one relies on higher aspect ratio in a synthetic finite metric.

\subsection{Two-Sided: Unbounded Inflation and Deflation at Constant
Aspect Ratio} \label{sec:twosidedrouting}

\subsubsection{Technical overview} 
We describe our two-sided routing inflation generator. The routing deflation generator is symmetric.

The generator manufactures a single center $A$ that the adaptive run holds and the replay lacks. This is achieved with a $1/4$ success probability by presenting a hub that both runs open, followed by a point $A$ at a distance of $f/2$ from the hub. This can be attempted multiple times to ensure a high success probability, with both runs paying an $O(f)$ routing cost on these attempts.

The generator then issues exponentially many (in the dimension) payout points that are at a distance of $0.9f$ from $A$, but at a distance $>f$ from each other and the hub. The replay pays for these payout steps entirely in \emph{openings}, incurring $0$ routing cost on the payout points. The adaptive run, holding $A$, pays in routing on the $0.1$-fraction of payout points which are rejected. Notably, the adaptive run opens \emph{fewer} centers and pays \emph{less} total cost in expectation: the generator relocates the adaptive run's payment from facilities into routing, and the replay's from routing into opening costs for new centers.

The routing ratio we obtain is exponential in the dimension, while the aspect ratio is constant.

\subsubsection{Vector packing property} 

We make use of the following lemma in the analysis.
\begin{lem}[Sign-vector packing]\label{lem:signvec}
For every $D \ge 2$ there exist $m = \lfloor e^{D/36} \rfloor$ unit
vectors $v_1, \dots, v_m \in \mathbb{R}^{D}$ with
$\langle v_i, v_j \rangle \le 1/3$ for all $i \ne j$.
\end{lem}

\begin{proof}
Let $v_i = D^{-1/2}(\xi_{i,1}, \dots, \xi_{i,D})$ with $(\xi_{i,k})$
i.i.d.\ uniform signs; each $v_i$ is a unit vector. For fixed
$i \ne j$ the coordinatewise products $(\xi_{i,k}\,\xi_{j,k})_{k \le D}$
are again i.i.d.\ uniform signs, so by Hoeffding's inequality
\[
\Pr\bigl[\langle v_i, v_j\rangle > 1/3\bigr]
\;=\; \Pr\Bigl[\sum_{k=1}^{D} \xi_{i,k}\,\xi_{j,k} > D/3\Bigr]
\;\le\; e^{-(D/3)^2/(2D)} \;=\; e^{-D/18} .
\]
A union bound over the $\binom{m}{2} \le m^2/2 \le e^{D/18}/2$ pairs
bounds the probability of failure by $1/2$, so there exists a set of $m$ unit vectors such  that all pairwise inner products $\le 1/3$.
\end{proof}

\subsubsection{The generator}
\medskip\noindent\textbf{Configuration.}
Fix $D \ge 2$, let $m$ and $v_1, \dots, v_m$ be as in
\Cref{lem:signvec}. For this lower bound construction, we consider the Euclidean metric in $\mathbb{R}^{D+1}$, and let $e_1$ denote the first standard basis vector and embedding the $v_i$ in the
orthogonal complement $\{0\} \times \mathbb{R}^{D}$. The construction uses four
\emph{regions}, and each region $c \in \{1, 2, 3, 4\}$ consists of a ``hub,'' an
attempt point, and $m$ payoff points, which we denote as follows:
\[
H_c \;=\; 4f(c-1)\, e_1, \qquad
A_c \;=\; H_c + \tfrac{f}{2}\, e_1, \qquad
w_{c,i} \;=\; A_c + 0.9f\, v_i \quad (i \in [m]) .
\]

\begin{claim}[Layout distances]\label{clm:layout}
Exactly, $d(A_c, H_c) = f/2$ and $d(w_{c,i}, A_c) = 0.9f$; moreover
$d(w_{c,i}, H_c) = f\sqrt{1.06} \in [\,1.02f,\ 1.03f\,]$ and
$d(w_{c,i}, w_{c,j}) \ge f\sqrt{1.08} \ge 1.03f$ for $i \ne j$. Every
point of region $c$ lies in $\bar B(H_c,\ 1.03f)$, so points of
distinct regions are at distance at least $4f - 2(1.03f) = 1.94f$.
Consequently all pairwise distances of the configuration lie in
$[\,f/2,\ 14.1f\,]$.
\end{claim}

\begin{proof}
The first two distances are immediate. Since $e_1 \perp v_i$,
Pythagoras gives
$d(w_{c,i}, H_c)^2 = (f/2)^2 + (0.9f)^2 = 1.06 f^2$, and with
\Cref{lem:signvec},
\[
d(w_{c,i}, w_{c,j})^2
= (0.9f)^2\,\|v_i - v_j\|^2
= 0.81 f^2 \bigl(2 - 2\langle v_i, v_j\rangle\bigr)
\;\ge\; 0.81 \cdot \tfrac43\, f^2 \;=\; 1.08 f^2 ;
\]
note $1.02^2 = 1.0404 \le 1.06 \le 1.0609 = 1.03^2 \le 1.08$. The
farthest point of region $c$ from $H_c$ is a payoff point, at
$f\sqrt{1.06} \le 1.03f$, so region $c$ lies in $\bar B(H_c, 1.03f)$;
the hubs are $4f$-spaced along $e_1$, which gives the inter-region
bound and, by the triangle inequality through the hubs,
$d_{\max} \le 12f + 2(1.03f) \le 14.1f$. The within-region distances
listed are the only distances below $1.94f$, and the smallest among
them is $d(A_c, H_c) = f/2$.
\end{proof}

\medskip\noindent\textbf{Generator and stopping time.}
For $c = 1, 2, 3, 4$: present the hub point $H_c$ followed by the attempt point $A_c$; if the adaptive run
opened $A_c$ \emph{and} the replay rejected it (i.e.\ $o = 1$ and
$o' = 0$ at that step --- this is the only two-sided observation the
generator makes), declare success, set $c^* = c$, and present
$w_{c^*,1}, \dots, w_{c^*,m}$ in order ( emit $\bot$ forever). 
Otherwise, continue with region $c+1$, emitting $\bot$ forever after
four failures. Let $\mathrm{SUCC}$ denote the event that some region
is successful. Also, observe that each decision is a deterministic function of
$(U_{<t}, U'_{<t})$, so this defines a deterministic two-sided generator in
the sense of \Cref{def:twosided}. A failed region occupies exactly
two steps and the generator advances immediately, so $H_c$ and $A_c$,
occupy the deterministic step indices $2c - 1$
and $2c$, and on $\mathrm{SUCC}$ the payoff points occupy steps
$2c^* + 1, \dots, 2c^* + m$. We stop at the deterministic time $\tau = m + 8$ (a fixed number of
steps); $\bot$-steps incur no cost.

\subsubsection{Analysis}

\begin{claim}[Run analysis]\label{clm:starverun}
In every realization:
\emph{(i)} every presented hub is at distance $\ge f$ from all
previously presented points, so both runs open it with probability
$1$ and it does not contribute routing cost for either run;
\emph{(ii)} when $A_c$ is presented, both runs' nearest-center
distances equal $f/2$ exactly, so the opening probability in each run is
$1/2$, where each decision is made using  fresh, mutually independent coins
$U_{2c}, U'_{2c}$;
\emph{(iii)} on $\mathrm{SUCC}$, when $w_{c^*,i}$ is presented the
adaptive nearest-center distance is exactly $0.9f$ (opening
probability $0.9$), while every oblivious center is at distance
$\ge 1.02f$, so the replay opens $w_{c^*,i}$ with probability $1$ and no routing costs are incurred at this step.
\end{claim}

\begin{proof}
Centers are always a subset of previously presented points, and by
\Cref{clm:layout} all cross-region distances are $\ge 1.94f \ge f$;
hence cross-region centers never realize a nearest-center minimum
below $f$. When a hub is presented, its own region contains no
earlier point (the first hub sees the empty set, at distance
$+\infty$), so both runs are dear there and open with probability
$1$, paying no routing: this is (i). For (ii): the only region-$c$
point presented before $A_c$ is $H_c$, which both runs hold by (i),
at distance exactly $f/2 < f$; the coins at step $2c$ are fresh and
mutually independent, and $\min\{1, (f/2)/f\} = 1/2$. For (iii): the
region-$c^*$ points presented before $w_{c^*,i}$ are $H_{c^*}$,
$A_{c^*}$, and $w_{c^*,j}$ for $j < i$, at distances $\ge 1.02f$,
$= 0.9f$, and $\ge 1.03f$ from $w_{c^*,i}$ respectively
(\Cref{clm:layout}). The adaptive run holds $A_{c^*}$, since success
requires $o_{2c^*} = 1$; hence its minimum is exactly $0.9f$ and its
opening probability is $0.9$. The replay does \emph{not} hold
$A_{c^*}$: success requires $o'_{2c^*} = 0$, and $A_{c^*}$ is
presented exactly once. Hence every oblivious center --- whichever
subset of the remaining presented points the replay happens to hold
--- is at distance $\ge 1.02f \ge f$, the step is dear for the
replay, and it opens with probability $1$, paying no routing. Note
that (iii) quantifies over all \emph{candidate} center sets, so no
induction on the replay's path is needed.
\end{proof}

\begin{prop}[Two-sided routing starvation]\label{prop:starve}
For every $D \ge 2$, with $m = \lfloor e^{D/36} \rfloor$, the above
configuration in $\mathbb{R}^{D+1}$, two-sided generator, and
deterministic stopping time $\tau = m + 8$ satisfy: in every realization
the presented stream has $d_{\min} = f/2$ and
$d_{\max} \in [\,f,\ 14.1f\,]$, hence $\Lambda = 2$ and
$\Delta \le 29$; and
\[
\E\bigl[Q^{\A}_\tau\bigr] \;\ge\; 0.06\, f\, m,
\qquad
\E\bigl[Q^{\OB}_\tau\bigr] \;\le\; f,
\]
while
$\E[K^{\A}_\tau] \le \E[K^{\OB}_\tau]$ and
$\E[C^{\A}_\tau] \le \E[C^{\OB}_\tau]$. Consequently the
two-sided routing \emph{inflation gap} is not bounded by
any function of the aspect ratio: it reaches $2^{\Omega(D)}$ at
$\Delta \le 29$, on instances where the adaptive run simultaneously
opens fewer centers and pays less total cost in expectation.
\end{prop}

\begin{proof}
\emph{Aspect ratios.} Every realization presents $H_1$ and $A_1$
(distance $f/2$, and no smaller distance exists in the configuration
by \Cref{clm:layout}), and then either $H_2$ (at distance
$\ge 3.5f$ from $H_1$) or, on success in region $1$, a payoff point
(at distance $f\sqrt{1.06} \ge f$ from $H_1$); so
$d_{\min} = f/2$ and $f \le d_{\max} \le 14.1f$ pathwise, giving
$\Lambda = \min\{f, d_{\max}\}/d_{\min} = 2$ and
$\Delta \le 14.1f/(f/2) \le 29$.

\emph{Success probability.} By \Cref{clm:starverun}(ii), region $c$,
if reached, succeeds iff $U_{2c} \le 1/2 < U'_{2c}$, an event of
probability $1/4$ determined by the fresh coin pair at the
deterministic step index $2c$; these events are independent across
$c$, and region $c$ is reached iff regions $1, \dots, c-1$ all
failed. Hence
\[
\Pr[\text{region } c \text{ reached}] = \Bigl(\tfrac34\Bigr)^{c-1},
\qquad
\Pr[\mathrm{SUCC}] = 1 - \Bigl(\tfrac34\Bigr)^{4}
= \tfrac{175}{256} \;\ge\; 0.68 .
\]

\emph{Oblivious routing.} Hubs and payoff points contribute zero
routing surely, by \Cref{clm:starverun}(i) and (iii). Attempt $A_c$,
if presented, contributes $(1 - o'_{2c}) \cdot \tfrac f2$; the event
that it is presented is determined by the coins at steps $< 2c$,
while $U'_{2c}$ is fresh, so the two are independent and
\[
\E\bigl[Q^{\OB}_\tau\bigr]
\;=\; \sum_{c=1}^{4} \Bigl(\tfrac34\Bigr)^{c-1}
\cdot \tfrac12 \cdot \tfrac f2
\;=\; \tfrac f4 \cdot \tfrac{175}{64}
\;=\; \tfrac{175}{256}\, f \;\le\; f .
\]

\emph{Adaptive routing.} Fix $i \in [m]$ and consider the step
$s = 2c^* + i$ on $\mathrm{SUCC}$. Conditionally on $\Gg_{s-1}$, on
the event that step $s$ presents $w_{c^*,i}$, the adaptive distance
is exactly $0.9f$ by \Cref{clm:starverun}(iii) and the coin $U_s$ is
fresh, so the step's conditional expected routing is
$(1 - 0.9) \cdot 0.9f = 0.09f$. Summing over the $m$ payoff steps by
the tower property and dropping the (nonnegative) contributions of
the attempts,
\[
\E\bigl[Q^{\A}_\tau\bigr]
\;\ge\; \Pr[\mathrm{SUCC}] \cdot m \cdot 0.09 f
\;\ge\; 0.68 \cdot 0.09\, f m \;\ge\; 0.06\, f m .
\]

\emph{Counts.} The two runs see the same presented sequence. Both
open every presented hub surely (\Cref{clm:starverun}(i)); each
presented attempt is opened by either run with probability exactly
$1/2$ (\Cref{clm:starverun}(ii), fresh coins); and on $\mathrm{SUCC}$
each payoff point is opened by the replay surely and by the adaptive
run with probability $0.9$ (\Cref{clm:starverun}(iii)). Taking
expectations termwise,
\[
\E\bigl[K^{\OB}_\tau\bigr] - \E\bigl[K^{\A}_\tau\bigr]
\;=\; \Pr[\mathrm{SUCC}] \cdot m \cdot (1 - 0.9) \;\ge\; 0 .
\]

\emph{Costs.} As in the proof of \eqref{eq:costcenter}, a step at
distance $d$ with opening probability $q = \min\{1, d/f\}$ has
conditional expected cost $f$ if $d \ge f$, and
$d + \tfrac{d}{f}(f - d)$ if $d < f$. Both runs pay $f$ per presented
hub and $\tfrac f2 + \tfrac12 \cdot \tfrac f2 = 0.75f$ per presented
attempt; on $\mathrm{SUCC}$ each payoff step costs the replay $f$
(dear) and the adaptive run $0.9f + 0.9 \cdot 0.1f = 0.99f$. Hence
\[
\E\bigl[C^{\OB}_\tau\bigr] - \E\bigl[C^{\A}_\tau\bigr]
\;=\; \Pr[\mathrm{SUCC}] \cdot m \cdot 0.01 f \;\ge\; 0 .
\qedhere
\]
\end{proof}

\begin{cor}[Two-sided routing deflation]\label{cor:starveflip}
Flipping the success test at the attempt from
$\{o_{2c} = 1,\, o'_{2c} = 0\}$ to $\{o'_{2c} = 1,\, o_{2c} = 0\}$
--- an event of the same probability $\tfrac14$ --- exchanges the
roles of the runs in \Cref{clm:starverun}(iii): the \emph{replay}
then holds $A_{c^*}$ at distance $0.9f$ from every payoff point while
every adaptive center is at distance $\ge 1.02f$, so the adaptive run
opens each payoff surely with no routing and the replay pays $0.09f$
per payoff in expectation. The identical accounting yields
$\E[Q^{\OB}_\tau] \ge 0.06fm$ and $\E[Q^{\A}_\tau] \le f$: the
two-sided routing \emph{deflation gap} is likewise unbounded ---
the ratio is driven below any function of the aspect ratio --- at
$\Delta \le 29$.
 Termwise, as in the \emph{Counts} and \emph{Costs} paragraphs of
the proof of \Cref{prop:starve}, the inequalities reverse: now
$\E[K^{\OB}_\tau] \le \E[K^{\A}_\tau]$ and
$\E[C^{\OB}_\tau] \le \E[C^{\A}_\tau]$ --- the starved run is
compensated in openings and pays more in total.
\end{cor}

\subsection{One-Sided: Routing Inflation $\Om(\Lambda)$}
\label{app:onesided}

\subsubsection{Technical overview} 
We describe our one-sided routing inflation generator. Unlike the two-sided construction, a one-sided generator cannot observe the replay outputs to guarantee a differential; instead, it creates a statistical differential in a custom finite metric space, paying for it with the aspect ratio.

The generator manufactures a center $A$ that the adaptive run holds and the replay likely lacks. It presents a hub $H$ that both runs open, followed by attempt points $A$ at a microscopic distance of $\epsilon' f$ from the hub. The generator stops at the first attempt $A$ that the adaptive run opens. Because the distance to the hub is so small, the oblivious replay opens $A$ with a tiny probability of exactly $\epsilon'$. Thus, with high probability ($1 - \epsilon'$), the generator successfully creates the differential: the adaptive run holds $A$, and the replay does not.

The generator then issues $m = 1/(\epsilon')^2$ payout points. For the oblivious replay, which likely lacks $A$, all payout points are at a distance of $f$ (from the hub). The replay finds these steps dear, opens a facility at almost every payout point, and pays a capped $O(f)$ in total routing. For the adaptive run, which holds $A$, the payout points are at a distance of $(1-\epsilon')f$. The adaptive run rejects each payout point with probability $\epsilon'$, paying $\approx f$ in routing each time it does so. Over $m$ payout points, it accumulates $\Omega(m \epsilon' f) = \Omega(\sqrt{m} f)$ in routing cost. As before, the adaptive run pays more in routing, but \emph{less} in total cost and \emph{fewer} total openings.

Because the minimum distance in this metric space is $\epsilon' f$, the aspect ratio is $\Delta = 1/\epsilon' = \sqrt{m}$. Therefore, the one-sided routing ratio we obtain is $\Omega(\sqrt{m}) = \Omega(\Delta)$, scaling linearly with the aspect ratio.

\subsubsection{The generator}

\medskip\noindent\textbf{Configuration.}
Fix an integer $m \ge 100$ and set
\[
\eps' = m^{-1/2} \;\le\; \tfrac1{10},
\qquad
n = \bigl\lceil 2/\eps' \bigr\rceil .
\]
The point set is
$\{H\} \cup \{A_1, \dots, A_n\} \cup
\{w^{(j)}_i : j \in [n],\ i \in [m]\}$
(one hub, $n$ attempts, and $m$ payoff points per attempt), with
distances, for all admissible distinct indices,
\[
d(H, A_j) = \eps' f, \qquad
d(A_j, A_k) = 2\eps' f, \qquad
d(A_j, w^{(j)}_i) = (1 - \eps') f,
\]
and $d(x, y) = f$ for every remaining pair --- that is,
$d(H, w^{(j)}_i) = d(A_k, w^{(j)}_i) = d(w^{(j)}_i, w^{(j')}_{i'}) = f$
for $k \ne j$ and distinct payoff pairs.

\begin{claim}[This is a metric]\label{clm:synthmetric}
The distances above satisfy the triangle inequality, so they define a
finite metric space with $d_{\min} = \eps' f$ and $d_{\max} = f$.
\end{claim}

\begin{proof}
All distances are positive and symmetric. A violation of
$d(x, z) \le d(x, y) + d(y, z)$ requires the largest side of a triple
to exceed the sum of the other two. If either of the other two sides
is at least $(1-\eps')f$, then, since every distance is at least
$\eps' f$, the sum is at least $(1-\eps')f + \eps' f = f$, which no
distance exceeds. Otherwise both other sides are at most $2\eps' f$;
but distances below $(1-\eps')f$ occur only within the cluster
$\{H, A_1, \dots, A_n\}$, so all three points lie in the cluster,
where the largest side, at most $2\eps' f$, does not exceed the sum
of the other two ($\ge \eps' f$ each). The extreme distances are read
off the configuration.
\end{proof}

\medskip\noindent\textbf{Generator and stopping time.}
Present $H$. Then for $j = 1, \dots, n$: present $A_j$; if the
adaptive run opens it, set $J = j$ and proceed to the payoff phase; if
no attempt is opened after $n$ tries, emit $\bot$ forever. In the
payoff phase, present $w^{(J)}_1, \dots, w^{(J)}_m$ in order, then
emit $\bot$ forever. Write $\mathrm{SEL} = \{J \text{ is defined}\}$.
Every decision depends only on the adaptive transcript, so this is a
deterministic \emph{one-sided} generator. Attempts are presented
consecutively from step $2$, so $A_j$, when presented, occupies the
deterministic step index $1 + j$, and on $\{J = j\}$ the payoff point
$w^{(j)}_i$ occupies step $1 + j + i$. We stop at the deterministic time $\tau = 1 + n + m$ (a fixed number
of steps); $\bot$-steps incur no cost.

\subsubsection{Analysis}

\begin{claim}[Run analysis]\label{clm:osrun}
In every realization:
\emph{(i)} $H$ is opened by both runs with probability $1$ and
contributes no routing to either;
\emph{(ii)} when $A_j$ is presented, the adaptive nearest-center
distance is exactly $\eps' f$, and the oblivious nearest-center
distance is exactly $\eps' f$ \emph{regardless of the replay's path};
hence both opening probabilities equal $\eps'$, decided by the fresh
independent coins $U_{1+j}, U'_{1+j}$;
\emph{(iii)} on $\{J = j\}$, when $w^{(j)}_i$ is presented, the
adaptive nearest-center distance is exactly $(1-\eps')f$ (opening
probability $1 - \eps'$), while the oblivious nearest-center distance
is exactly $(1-\eps')f$ if $A_j \in S^{\OB}$ and exactly $f$
otherwise; in the latter case the replay opens with probability $1$
and pays no routing.
\end{claim}

\begin{proof}
(i): $H$ is the first point, at distance $+\infty$ from the empty
center sets. (ii): the points presented before $A_j$ are $H$ and
$A_1, \dots, A_{j-1}$. The adaptive run holds $H$ (by (i)) and none
of the earlier attempts (it rejected each --- otherwise the attempt
phase would have ended earlier), so its distance is
$d(A_j, H) = \eps' f$ exactly. The replay's center set is a subset of
$\{H, A_1, \dots, A_{j-1}\}$ containing $H$; since
$d(A_j, H) = \eps' f < 2\eps' f = d(A_j, A_k)$, the minimum is
realized by $H$ whatever strays the replay holds, and equals
$\eps' f$ exactly. The coins at step $1+j$ are fresh and mutually
independent, and $\min\{1, \eps' f / f\} = \eps'$. (iii): the points
presented before $w^{(j)}_i$ are $H$, $A_1, \dots, A_j$, and
$w^{(j)}_{i'}$ for $i' < i$, at distances $f$, $f$ (for indices
$k < j$), $(1-\eps')f$ (for $A_j$), and $f$ respectively. The
adaptive run holds $A_j$ (this is the selection condition) and none
of $A_1, \dots, A_{j-1}$, so its minimum is exactly $(1-\eps')f$.
Every \emph{candidate} oblivious center other than $A_j$ is at
distance exactly $f$, so the replay's minimum is $(1-\eps')f$ or $f$
according to whether $A_j \in S^{\OB}$; at distance $f$ the opening
probability is $\min\{1, f/f\} = 1$. As in \Cref{clm:starverun},
(iii) quantifies over candidate center sets, so no induction on the
replay's path is needed.
\end{proof}

\begin{claim}[Deterministic lack-posterior]\label{clm:oscert}
Let $\chi_j = \ind{\{U'_{1+j} \le \eps'\}}$; the family
$(\chi_j)_{j \le n}$ is independent of $U$. On $\{J = j\}$ (a
$\sigma(U)$-measurable event), $A_j \in S^{\OB}_t$ if and only if
$\chi_j = 1$, for every $t \ge 1 + j$. In particular
$\Pr\bigl[A_J \in S^{\OB} \mid \sigma(U)\bigr] = \eps'$ on
$\mathrm{SEL}$.
\end{claim}

\begin{proof}
On $\{J = j\}$, the point $A_j$ is presented exactly once, at the
deterministic step $1+j$, where by \Cref{clm:osrun}(ii) its replay
opening probability is the \emph{deterministic constant} $\eps'$;
hence $o'_{1+j} = \ind{\{U'_{1+j} \le \eps'\}} = \chi_j$, a function
of $U'$ alone, independent of $\sigma(U)$. Selection is determined by
the adaptive coins, so $\{J = j\} \in \sigma(U)$, and the conditional
probability follows.
\end{proof}

\begin{prop}[One-sided routing starvation]\label{prop:onesidedstarve}
For every integer $m \ge 100$, the above finite metric space,
deterministic one-sided generator, and deterministic stopping time
$\tau = 1 + n + m$ satisfy: in every realization the presented stream
has $d_{\min} = \eps' f$ and $d_{\max} \le f$, hence
$\Delta = \Lambda \le \sqrt m$; and
\[
\E\bigl[Q^{\A}_\tau\bigr] \;\ge\; 0.75\,\sqrt m\, f,
\qquad
\E\bigl[Q^{\OB}_\tau\bigr] \;\le\; 2f,
\qquad
\E\bigl[Q^{\A}_\tau\bigr] \;\ge\; \frac{\sqrt m}{2}\,
\E\bigl[Q^{\OB}_\tau\bigr],
\]
while $\E[K^{\A}_\tau] \le \E[K^{\OB}_\tau]$ and
$\E[C^{\A}_\tau] \le \E[C^{\OB}_\tau]$. Consequently the
routing adaptivity ratio over one-sided generators grows at least
linearly in the capped aspect ratio: it is
$\Om(\Lambda) = \Om(\Delta)$.
\end{prop}

\begin{proof}
\emph{Aspect ratios.} Every realization presents $H$ and $A_1$
(distance $\eps' f$, the smallest distance of the configuration by
\Cref{clm:synthmetric}), so $d_{\min} = \eps' f$ pathwise, and
$d_{\max} \le f$ since no configuration distance exceeds $f$; hence
$\Lambda = \min\{f, d_{\max}\}/d_{\min} = d_{\max}/(\eps' f) \le
1/\eps' = \sqrt m$ and $\Delta = \Lambda$ pathwise.

\emph{Selection.} By \Cref{clm:osrun}(ii) the adaptive attempt coins
are fresh Bernoulli($\eps'$) at the deterministic steps
$2, \dots, 1+n$, so
\[
\Pr[\mathrm{SEL}] \;=\; 1 - (1 - \eps')^{n}
\;\ge\; 1 - e^{-n\eps'} \;\ge\; 1 - e^{-2} \;\ge\; 0.86 .
\]

\emph{Oblivious routing, exactly.} $H$ contributes zero. Attempt
$A_j$ is presented iff the adaptive run rejected
$A_1, \dots, A_{j-1}$ (probability $(1-\eps')^{j-1}$, an event of the
earlier adaptive coins, independent of the fresh $U'_{1+j}$), and
then contributes $(1 - o'_{1+j})\cdot\eps' f$, of conditional
expectation $(1-\eps')\eps' f$; summing the geometric series,
\[
\E\bigl[Q^{\OB}(\text{attempts})\bigr]
= \eps'(1-\eps') f \sum_{j=1}^{n} (1-\eps')^{j-1}
= (1-\eps')\bigl(1 - (1-\eps')^{n}\bigr) f
= (1-\eps')\,\Pr[\mathrm{SEL}]\, f .
\]
For the payoff phase, fix $i \in [m]$ and condition on
$\Gg_{s-1}$ at the step $s = 1+j+i$ on $\{J = j\}$: by
\Cref{clm:osrun}(iii) the step's conditional expected replay routing
is $\eps'\,(1-\eps')f$ on $\{\chi_j = 1\}$ and $0$ on
$\{\chi_j = 0\}$ (dear, opened surely). By \Cref{clm:oscert},
$\E[\ind{\{J=j\}}\chi_j] = \eps'\,\Pr[J = j]$, so summing over $j$
and the $m$ payoff indices, and using $m\,\eps'^2 = 1$,
\[
\E\bigl[Q^{\OB}(\text{payoffs})\bigr]
= m\,\eps'(1-\eps')f\cdot \eps'\,\Pr[\mathrm{SEL}]
= (1-\eps')\,\Pr[\mathrm{SEL}]\, f .
\]
Hence, exactly,
\[
\E\bigl[Q^{\OB}_\tau\bigr]
\;=\; 2\,(1-\eps')\,\Pr[\mathrm{SEL}]\, f \;\le\; 2f .
\]

\emph{Adaptive routing.} By \Cref{clm:osrun}(iii), on $\mathrm{SEL}$
each of the $m$ payoff steps has conditional expected adaptive
routing exactly $\eps'\,(1-\eps')f$, deterministically on that
branch. Dropping the (nonnegative) attempt contributions and using
$m\eps' = \sqrt m$,
\[
\E\bigl[Q^{\A}_\tau\bigr]
\;\ge\; \Pr[\mathrm{SEL}]\cdot m\,\eps'(1-\eps') f
\;=\; (1-\eps')\,\Pr[\mathrm{SEL}]\,\sqrt m\, f
\;\ge\; 0.86 \cdot 0.9\,\sqrt m f \;\ge\; 0.75\,\sqrt m\, f ,
\]
using $\eps' \le 1/10$. Dividing the two displays, the factors
$(1-\eps')\Pr[\mathrm{SEL}] > 0$ cancel:
\[
\frac{\E[Q^{\A}_\tau]}{\E[Q^{\OB}_\tau]}
\;\ge\; \frac{(1-\eps')\Pr[\mathrm{SEL}]\,\sqrt m\, f}
{2(1-\eps')\Pr[\mathrm{SEL}]\, f}
\;=\; \frac{\sqrt m}{2} .
\]

\emph{Counts.} $H$ is opened by both runs surely. Each presented
attempt is opened by either run with probability exactly $\eps'$
(\Cref{clm:osrun}(ii), fresh coins), and whether it is presented is
independent of both step-$(1{+}j)$ coins, so the attempts contribute
equally in expectation. Per payoff step, conditionally on
$\Gg_{s-1}$, the adaptive run opens with probability $1 - \eps'$
while the replay opens with probability $1$ on $\{\chi_J = 0\}$ and
$1 - \eps'$ on $\{\chi_J = 1\}$; in either case at least $1 - \eps'$.
Hence $\E[K^{\OB}_\tau] - \E[K^{\A}_\tau] =
m\,\eps'\,(1-\eps')\,\Pr[\mathrm{SEL}] \ge 0$.

\emph{Costs.} As in the proof of \eqref{eq:costcenter}, a step at
distance $d$ with opening probability $q = \min\{1, d/f\}$ has
conditional expected cost $f$ if $d \ge f$ and $d + \tfrac df(f - d)$
if $d < f$. The hub costs both runs $f$; each presented attempt costs
both runs $\eps' f + \eps'(f - \eps' f) = \eps'(2 - \eps')f$; each
payoff step costs the adaptive run
$(1-\eps')f + (1-\eps')\eps' f = (1 - \eps'^2)f$, and the replay the
same on $\{\chi_J = 1\}$ but $f \ge (1-\eps'^2)f$ on
$\{\chi_J = 0\}$. Hence
$\E[C^{\OB}_\tau] - \E[C^{\A}_\tau]
= m\,\eps'^2\,(1-\eps')\,\Pr[\mathrm{SEL}]\, f
= (1-\eps')\,\Pr[\mathrm{SEL}]\, f \ge 0$.
\qedhere
\end{proof}

\end{document}